\documentclass[envcountsect,,orivec]{llncs}

\usepackage{amsmath}
\usepackage{amssymb}
\usepackage{stmaryrd}
\usepackage[numbers,sort&compress]{natbib}
\usepackage{mathtools}
\usepackage{graphicx}
\usepackage[usenames,dvipsnames]{xcolor}  
\usepackage{xifthen}        
\usepackage{hyperref}     
\usepackage{cleveref}
\crefname{section}{Sect.}{Sect(s).}
\crefname{figure}{Fig.}{Fig(s).}
\crefname{theorem}{Thm.}{Thm(s).}
\crefname{definition}{Def.}{Def(s).}
\crefname{proposition}{Prop.}{Prop(s).}
\crefname{corollary}{Cor.}{Cor(s).}
\crefname{example}{Ex.}{Ex(s).}
\crefname{lemma}{Lemma}{Lemmas}
\crefname{example}{Ex.}{Ex(s).}

\usepackage{longtable}
\usepackage{TransCCS}
\usepackage{sessions}

\usepackage{savesym}      
\usepackage{vk}           
\savesymbol{todo}         
\usepackage{todonotes}    

\usepackage{enumerate}
\usepackage{logicproof}     
\usepackage{listings}       
\usepackage{environ}        
\usepackage{array}

\definecolor{mygreen}{rgb}{0,0.6,0}
\definecolor{mygreen}{rgb}{0,0.6,0}
\allowdisplaybreaks

\usepackage{comment}        
\excludecomment{hide}
\includecomment{icalp}

\title{Type-Based Analysis for Session Inference (Extended Abstract)%
\thanks{%
This research was supported, in part, by Science Foundation Ireland grant 13/RC/2094.
The first author was supported by MSR (MRL 2011-039).
}
}
\titlerunning{Type-Based Analysis for Session Inference}
\author{Carlo Spaccasassi \and Vasileios Koutavas}
\authorrunning{C. Spaccasassi \and V. Koutavas}
\institute{Trinity College Dublin}

\begin{document}

\maketitle

\begin{abstract}
  We propose a type-based analysis to infer the session protocols of channels in an ML-like concurrent functional language.
Combining and extending well-known techniques, we develop a type-checking system that separates the underlying ML type system from
the typing of sessions. Without using linearity, our system guarantees communication safety and partial lock freedom. It also
supports provably complete session inference for finite sessions with no programmer annotations. We exhibit the usefulness of our system with
interesting examples, including one which is not typable in substructural type systems.

\end{abstract}

  \section{Introduction}
   \label{sec:intro}

Concurrent programming often requires processes to communicate according to
intricate protocols. In mainstream programming languages these protocols are
encoded implicitly in the program's control flow, and no support is available
for verifying their correctness.

\begin{hide}
Honda \cite{Honda93} and others (e.g. \cite{Honda94,HondaEtal98,GayHole99})
\end{hide}
Honda \cite{Honda93} first
suggested the use of \emph{binary session types} to explicitly describe and check
protocols over communication channels with two endpoints.
Fundamentally, session type systems guarantee that a program respects
the order of communication events (session fidelity) and message types (communication safety)
described in a channel's session type.
A number of session type systems (e.g.,
\cite{Caires2010,PhenningSessionMonad,CastagnaEtal09})
also ensure
that processes fully execute the protocols of their open endpoints, as long as they do
not diverge
or block on opening new
sessions (partial lock freedom).

To date, binary session type disciplines have been developed for various process
calculi and high-level programming languages (see \cite{BETTYwg1} for an
overview) by following one of two main programming language design approaches:
using a single substructural type system for both session and traditional
typing
\begin{hide}
\cite{GV,Vasconcelos2006,Wadler2012,OOSessions06,SessionJava},
\end{hide}
\cite{GV,Vasconcelos2006, HondaEtal98, 
Wadler2012},
or using monads to separate the two
\cite{PhenningSessionMonad,PucellaT08}.

In this paper we propose a third design approach which uses \emph{effects}. Similar to previous work, our approach enables the
embedding of session types in programming languages with sophisticated type systems. Here we develop a high-level language where
intricate protocols of communication can be programmed and checked statically (\cref{sec:motivating-examples}). Contrary to both
monads and substructural type systems, our approach allows pure code to call library code with communication effects, without
having to refactor the pure code (e.g., to embed it in a monad or pass continuation channels through it---see \cref{ex:db}). We
apply our approach to \lang, a core of ML with session communication (\cref{sec:lang}).

Our approach separates traditional typing from session typing in a two-level system, which follows the principles of \emph{typed
based analysis} \cite{PalsbergTBA}. The first level employs a type-and-effect system, which adapts and extends the one of Amtoft,
Nielson and Nielson \cite{ANN} to session communication (\cref{sec:stage-1}). At this level the program is typed against an
ML type and a \emph{behaviour} which abstractly describes program structure and communication. Session protocols are not
considered here---they are entirely checked at the second level. Thus, each endpoint is given type $\tses{\rho}$, where $\rho$
statically approximates its source. The benefit of extending~\cite{ANN} is that we obtain a complete behaviour
inference algorithm, which extracts a behaviour for every program respecting ML types.

At the second level, our system checks that a behaviour, given an operational semantics, complies with the session types of channels and endpoints
(\cref{sec:session-types}). The session discipline realised here is inspired by the work of Castagna et al.~\cite{CastagnaEtal09}. This discipline guarantees that programs comply with session fidelity
and communication safety, but also, due to stacked interleaving of sessions, partial lock freedom.
However, one of the main appeals of our session typing discipline is that it enables a provably
\emph{complete session types inference} from behaviours which, with behaviour inference, gives us a complete method for
session inference from \lang, without programmer annotations (\cref{sec:inference}).
The two levels of our system only interact through behaviours, which we envisage will allow us to develop front-ends for
different languages and back-ends for different session disciplines.
%

To simplify the technical development we consider only sessions of finite interactions. However, we allow recursion in the source
language, as long as it is \emph{confined}: recursive code may only open new sessions and completely consume them
(see \cref{sec:motivating-examples}). In \cref{sec:extensions} we discuss an extension to recursive
types. Related work and conclusions can be found in \cref{sec:relwork}. Details missing from this extended abstract can be found
in the appendix for the benefit of the reviewers.

\begin{hide}
Concurrent programming often requires processes to communicate according to
intricate protocols. In mainstream programming languages these protocols are
encoded implicitly in the program's control flow, and no support is available
for verifying their correctness.

Honda \cite{Honda93} and then others (e.g. \cite{
})
suggested the use of \emph{binary session types} to make communication
protocols explicit and checkable in program typing. Binary session types
describe protocols over communication channels with two endpoints, whereby
values of different types can be exchanged, and branching choices can be
offered (external choice) and selected (internal choice). Since then, session
type disciplines have been developed for various process calculi and high-level
programming languages (see \cite{BETTYwg1} for an overview).

In this paper we are interested in providing a facility for static protocol
checking in existing high-level programming languages, and in particular ML, by
extending such languages with binary session types. We aim to an extension that
requires almost no programmer annotations, and a \emph{provably complete}
session type inference algorithm makes precise the few annotations needed.

Existing session type systems ensure that well-typed programs \emph{partially}
adhere to communication protocols, a property expressed as a subject reduction
lemma. A process $P$, containing an endpoint $x$ with session type
``$\inp \mathtt{int}.\eta $,'' can rely on the fact that \emph{if its environment
sends a value on $x$} then that value will be of type $\mathtt{int}$. Moreover,
$P$ may only use $x$ to receive---not send---a value, but it does not have to do
so. This is because the process can deadlock due to incorrect communication
interleaving, diverge, or block while waiting for a partner to start a new
connection. If the communication on $x$ happens, then $\eta$ describes the
remaining protocol for this endpoint.
Therefore a process partially adheres to a protocol up to divergence, deadlock
and absence of communication partners.

A number of session type systems (e.g.,
\cite{PhenningSessionMonad,
CastagnaEtal09
}) also remove the possibility of a deadlock, providing
a \emph{weak progress} guarantee. This roughly ensures that processes will
indeed fully execute the protocols of their open endpoints, as long as they do
not diverge\footnote{Some systems also remove the possibility of divergence by
considering only strongly normalising programs (e.g.,
\cite{PhenningSessionMonad}).} and they do not block for opening a new
connection. Even with these, arguably reasonable, exceptions, such systems are
particularly useful when communication is used to control a critical resource.

Here we also develop a session type system that guarantees weak progress.
Our approach is modular, allowing adaptations to weaker typing disciplines.

Many current efforts to developing session types for high-level languages use a
substructural type system which combines both expression and session typing
\cite{GV,Vasconcelos2006,Wadler2012
}. An alternative
approach is through the use of \emph{monads}
\cite{PhenningSessionMonad,PucellaT08}.

In this paper we put forward a new approach for adding session types to
high-level programming languages, which we present in a core of ML\@. Rather than
typing sessions directly on the source code, our approach is based on first
extracting the \emph{communication effect} of program expressions and then
imposing a session type discipline on this effect. This approach has the
following benefits:
\begin{itemize}
  \item It enables session types inference with minimal annotations which is
  sound but also provably complete. To our
    knowledge this is the first provably complete session types inference
    algorithm for ML\@. The only annotation necessary differentiates recursive
    functions that unfold recursive session types from those that use only
    finite session types.


  \item The session type system is modular. By decoupling the session type
  discipline from the precise syntax of the source
    language, we can in principle change one while keeping the other. We use our
    approach to develop a session type system with weak progress guarantees,
    similar to the systems mentioned previously; however with simple
    modifications our system can be weakened to resemble more the session type
    systems without such guarantees.
    Moreover, our approach could lead to a session type system for programs
    written in multiple programming languages.

  \item With this approach we are able to type programs where client code,
  oblivious to session communication, uses library
    functions containing open session endpoints (see~\ref{ex:db}). A
    substructural session type system would assign linear types to these
    functions, not allowing client code to use them more than once, thus
    requiring the restructuring of client code to thread linear endpoints to
    library function calls. With a monadic session type system, client code
    would need to be included in a monad, posing the challenge of composition
    with language features such as exceptions and mutable state.

\end{itemize}
\end{hide}
\begin{hide} 
To extract the communication effect of ML programs we adapt and extend the work
on type-and-effect systems developed by Amtoft, Nielson and Nielson
\cite{ANN,NNbehaviours}. Our extension provides a method for dealing with
aliasing of session endpoints using regions, obviating the need for a
substructural type system for ML\@. Furthermore, we develop a session type
discipline for communication effects inspired by Castagna et al.\
\cite{CastagnaEtal09} and show that it guarantees weak progress.

Here we focus on a pure core of ML\@. However, we believe that the techniques we use would work equally well in other high-level
languages---we leave this for future work. The language employs Hindley-Milner polymorphism and primitives for sending pure
monomorphic values, internal and external choices through choice labels $L$, and session
delegation and resumption. The ability to send functions containing communication effects would be an implicit form of delegation
and would unnecessarily complicate the exposition of our system.

For simplicity of presentation we develop our system for finite sessions and then show how it can be extended with recursive
sessions. However, even with finite session types, we use a novel typing for a class of useful, non-pure recursive functions,
which we call \emph{confined functions}. For example consider the function of type
$\tfun{\tunit}{\tunit}{}$:
\begin{equation}
 \efix{f}{\_}{\nbox{%
          \elet{z}{\eapp{\erequest{}{\keyw{init}}}{\cunit}}
          {\ematch{z}{%
              L_0 \Rightarrow \eapp{f}{\cunit},~~
              L_1 \Rightarrow \cunit
            }{}
           }}
           \tag{$\star$}\label{eq:intro-example}
}\end{equation}
When called, the function invokes $\erequest{}{\keyw{init}}$ requesting to open a session on a global channel $\keyw{init}$.
When the request is accepted by a partner process, a \emph{session endpoint} is bound to $z$ through which the processes can
communicate.
The process presents its partner with choices $L_0$ and $L_1$. If $L_0$ is chosen, the
function recurs; if $L_1$ is chosen the function terminates. In both cases no further communication on endpoint $z$ will occur,
and therefore $z$ can be given the finite session type $\Sigma\{\inp L_0. \tend,~\inp L_1. \tend\}$. This behaviour is representative of
systems that run a finite protocol an arbitrary number of times.
\end{hide}

\begin{hide}
The following section presents interesting examples using session communication that we would like to type in \lang.
\cref{sec:lang} gives the syntax and operational semantics of the language.
\cref{sec:types} presents the details of our typing system and the type
soundness result.
\cref{sec:inference} describes our type inference algorithm and its soundness
and completeness.
\cref{sec:extensions} discusses the extension of
our type system to a form of recursive session types. \cref{sec:relwork}
presents related work and conclusions.
Omitted details can be found in the Appendix.
\end{hide}

\section{Motivating Examples}\label{sec:motivating-examples}
\begin{hide}
Before presenting the details of the type system, we give and discuss two example implementations of a swap service, which
symmetrically exchanges values between pairs of processes connecting to it, and a database library which allows clients to send
queries to a concurrent database server. These are typical examples where process communication is used in programming.
\end{hide}

%

\begin{example}[A Swap Service]
\label{ex:swap1}
A coordinator process
uses the primitive $\eaccept{}{\keyw{swp}}$ to accept two connections on a channel \keyw{swp}
(we assume functions $\eaccept{}{c}$ and $\erequest{}{c}$ for every channel $c$),
opening two concurrent sessions with processes that want to exchange values.
It then coordinates the exchange and recurs.

\begin{minipage}{.45\textwidth}
\begin{lstlisting}[backgroundcolor=\color{white},numbers=none]
let fun coord(_) =
  let val p1 = $ \eaccept{}{\keyw{swp}}$ ()
      val x1 = $ \erecv{}{}$ p1
      val p2 = $ \eaccept{}{\keyw{swp}}$ ()
      val x2 = $ \erecv{}{}$ p2
  in $ \esend{}{}$ p2 x1; $\esend{}{} $ p1 x2; coord ()
in spawn coord;
\end{lstlisting}
\end{minipage}
\hfill
\begin{minipage}{.45\textwidth}
\begin{lstlisting}[backgroundcolor=\color{white},numbers=none]
let fun swap(x) =
  let val p = $\erequest{}{\keyw{swp}}$ ()
  in $\esend{}{} $ p x; $ \erecv{}{}$ p
in spawn (fn _ =$>$ swap 1);
   spawn (fn _ =$>$ swap 2);
\end{lstlisting}
\vspace{2em}
\end{minipage}

Each endpoint the coordinator receives from calling $\eaccept{}{\keyw{swp}}$ are used
according to the session type
$%
\inp T.\outp T.\tend
$.
This says that, on each endpoint, the coordinator will first read a value type $T$  ($\inp T$), then output a
value of the same type ($\outp T$) and close the endpoint ($\tend$). The interleaving of sends and receives
on the two endpoints achieves the desired swap effect.

\begin{hide}
To use this service we simply apply the
following
\textsf{swap} function to the value to be exchanged.
\[
  \begin{array}{@{}l@{\,}l@{\,}l@{}}
    \klet & \mathit{swap} = & \efunc{x}{\nbox{(
                \elet{z}{\eapp{\erequest{}{\keyw{swp}}}{\cunit}}{%
                \eapp{\esend{z}}{x};~
                \erecv{}{z})
    }}} \\
      \keyw{in} & \ldots
  \end{array}
\]
When applied to \textsf{x}, \textsf{swap} requests a connection on
$\keyw{swp}$,
receiving a session endpoint \textsf{p}, then sends \textsf{x} on \textsf{p}
and finally receives a value on
\textsf{p} which is returned as the result of the
function.
\end{hide}
Function $\keyw{swap}: \tfun{\tint}{T'}{}$
calls $\erequest{}{\keyw{swp}}$
and receives and endpoint which
is used
according to the session type $\outp \tint.\inp T'.\tend
$.
 By comparing the two session types above we
can see that the coordinator and the swap service can communicate without type errors, and indeed are typable, when $T=\tint=T'$.
Our type inference
algorithm automatically deduces the two session types from this code.

Because $\mathtt{swp}$ is a global channel, ill-behaved client code can
connect to it too:

\begin{minipage}{\textwidth}
\begin{lstlisting}[backgroundcolor=\color{white},numbers=none]
let val p1 = $\erequest{}{\keyw{swp}}$ () in send p1 1;
let val p2 = $\erequest{}{\keyw{swp}}$ () in send p2 2;
let val (x1$,$ x2) = ($\erecv{}{} $ p1$, \erecv{}{} $ p2) in $e_\mathsf{cl}$
\end{lstlisting}
\end{minipage}
\begin{hide}
\[
\nbox{
\elet{z_1}{%
  \eapp{\erequest{}{\keyw{swp}}}{\cunit}
}{\\
  \eapp{\esend{z_1}}{1};\\
  \elet{z_2}{%
    \eapp{\erequest{}{\keyw{swp}}}{\cunit}
  }{\\
    \eapp{\esend{z_2}}{2};\\
    \elet{x_1}{\erecv{}{z_1}}{\\
      \elet{x_2}{\erecv{}{z_2}}{\\
        e_\mathsf{cl}
      }
    }
  }
}
}
\]
\end{hide}
%
This client causes a deadlock, because the coordinator first sends
on \keyw{p2} and then on \keyw{p1}, but this code orders the corresponding receives in
reverse. The interleaving of sessions in this client is rejected by our type system because it is not \emph{well-stacked}:
$\erecv{}{\keyw{p1}}$ is performed before the most recent endpoint (\keyw{p2}) is closed. The interleaving in the coordinator, on
the other hand, is well-stacked.


\end{example}

\begin{example}[Delegation for Efficiency]
\label{ex:swap2}

In the previous example the coordinator is a bottleneck when exchanged values are large.
A more efficient implementation delegates exchange to the clients:

\begin{minipage}{.45\textwidth}
\begin{lstlisting}[backgroundcolor=\color{white},numbers=none]
let fun coord(_) =
  let val p1 = $\eaccept{}{\keyw{swp}}$ ()
  in $\eselectnew{\keyw{SWAP}}{}$ p1;
    let val p2 = $\eaccept{}{\keyw{swp}}$
    in $ \eselectnew{}{\keyw{LEAD}} $ p2;
       $ \edeleg{}{}$ p2 p1;
       coord()
\end{lstlisting}
\end{minipage}
\hfill
\begin{minipage}{.45\textwidth}
\begin{lstlisting}[backgroundcolor=\color{white},numbers=none]
let fun swap(x) =
  let val p = $ \erequest{}{\keyw{swp}} $ $ \cunit $
  in case p {
    $\keyw{SWAP}$: $\esend{}{}$ p x; $\erecv{}{} $ p
    $\keyw{LEAD}$: let val q = $ \eresume{}{}$ p
             val y = $\erecv{}{}$ q
           in $ \esend{}{}$ q x; y }
\end{lstlisting}
\end{minipage}

Function \keyw{swap} again connects to the coordinator over channel
$\keyw{swp}$, but now offers two choices with the labels $\keyw{SWAP}$ and $\keyw{LEAD}$.
If the coordinator selects the former, the swap method proceeds as
before; if it selects the latter, \keyw{swap}
resumes (i.e., inputs) another endpoint, binds
it to \keyw{q}, and performs a \keyw{rcv} and then a \keyw{send} on $\keyw{q}$.
The new coordinator accepts two sessions on $\keyw{swp}$, receiving two
endpoints: \keyw{p1} and \keyw{p2}. It selects \keyw{SWAP} on \keyw{p1}, $\keyw{LEAD}$ on \keyw{p2},
sends \keyw{p1} over \keyw{p2} and recurs.

When our system analyses the coordinator in isolation, it infers the protocol
$
\eta_{\keyw{coord}} = (\outp \keyw{SWAP}.\eta' \;\oplus\;\outp \keyw{LEAD}.\outp \eta'.\tend)
$
for both endpoints \keyw{p1} and \keyw{p2}.
When it analyses $\keyw{swap}:\tfun{T_1}{T_2}{}$, it infers
$
\eta_{\keyw{p}} = \Sigma\{%
\inp\keyw{SWAP}. \outp T_1. \inp T_2. \tend,\;
\inp\keyw{LEAD}. \inp\eta_{\keyw{q}}.\tend \}
$
 and
$\eta_{\keyw{q}}= \inp T_2. \outp T_1.\tend$
as the protocols of \keyw{p}
and \keyw{q}, respectively.
The former \emph{selects} either options $\keyw{SWAP}$ or $\keyw{LEAD}$
and the latter \emph{offers} both options.

%
%

If the coordinator is type-checked in isolation, then typing
succeeds with any $\eta'$: the coordinator can delegate any session.
However, because of duality, the typing of $\erequest{}{\keyw{swp}} $ in the
swap function implies that $\eta'=\eta_{\keyw{q}}$ and $T_1=T_2$.
Our inference algorithm can type this program and derive the above
session types.

\begin{hide}
It again connects to the coordinator over channel $\mathtt{swp}$, but now offers two choices: $\mathtt{Fst}$
and $\mathtt{Snd}$. If the coordinator selects the first one then the swap method behaves as before: it sends its value and
receives another which it returns. If the coordinator selects the second choice then $\mathit{swap}$ will resume (i.e., input) another endpoint
by calling $\eresume{}{z}$, bind it to $z'$, and then receive a value $y'$ from $z'$, send $x$ on $z'$ and finally return $y'$. Therefore,
the session type of endpoint $z$ is
\[
  \eta_{\texttt{swap}} =
  \Sigma\{ \nbox{
      \inp \mathtt{Fst}. \outp T_1. \inp T_2. \tend,
      ~\inp \mathtt{Snd}. \inp \eta. \tend \}
  }
  \qquad
  \eta_{\texttt{coord}} =
  \outp \mathtt{Fst}.\eta'
  ~\oplus~
  \outp \mathtt{Snd}.\outp \eta'.\tend
\]
denoting the choice between the two options $\mathtt{Fst}$ and $\mathtt{Snd}$, and the protocol followed in each one.
Here $\eta$ is the session type of endpoint $z'$ which is $\eta={\inp }T_2. \outp T_1.\tend$.
Again, $T_1$ and $T_2$ are the argument and return types of $\mathtt{swap}$.

The new coordinator is:
\[
  \begin{array}{@{}l@{\,}l@{\,}l@{}}
    \klet
         & \mathit{coord} = & \efix{f}{\_}{
    \nbox{\elet{z_1}{\eapp{\eaccept{}{\keyw{swp}}}{\cunit}
               \\} {\nbox{
                 \eapp{\eselectnew{\keyw{Fst}}}{z_1} \iseq \\
                \elet{z_2}{\eapp{\eaccept{}{\keyw{swp}}}{\cunit}
                \\}{\nbox{
                  \eapp{\eselectnew{\keyw{Snd}}}{z_2} \iseq
                \eapp{\edeleg{z_2}}{z_1} \iseq
                \eapp{f}{\cunit}
      }}}}}} \\
      \keyw{in} &
      \multicolumn{2}{@{}l@{}}{\!\!
        \espawn{\mathit{coord}} \iseq \ldots
      }
   \end{array}
\]

The new coordinator
accepts two sessions on $\mathtt{swp}$, receiving two endpoints: $z_1$ and $z_2$. It selects $\mathtt{Fst}$ on $z_1$
($\eapp{\eselectnew{\keyw{Fst}}}{z_1}$) and $\mathtt{Snd}$ on $z_2$ ($\eapp{\eselectnew{\keyw{Snd}}}{z_2}$).
The coordinator then sends $z_1$ over $z_2$ and recurs.

The protocol followed by the coordinator over the two endpoints is now more intricate, and in fact different for each one.
However, both endpoints must have the same session type because they are both generated by accepting a connection on
$\mathtt{swp}$.  This can be encoded with an \emph{internal choice}:
$
  \outp \mathtt{Fst}.\eta'
  ~~\oplus~~
  \outp \mathtt{Snd}.\outp \eta'.\tend
$.
In the case where the coordinator chooses the first choice, the rest of the session $\eta'$ over the endpoint is delegated, and
therefore it can be any session---$\eta'$ will be executed from the process that receives $z_1$ (running a $\mathit{swap}$ function).
If the coordinator selects the second choice then it simply delegates an endpoint with session type $\eta'$.

In our type system, if the coordinator is type-checked in isolation, then typing succeeds with any $\eta'$.
However, if both coordinator and swap function are typed in the same program, typing succeeds only when $\eta'=\eta$ and $T_1=T_2$.
These equalities are necessary to guarantee that the two dual endpoints of $\mathtt{swp}$ have dual session types.
Our inference algorithm is able to type this program and derive the above session types directly from the source code, without
any annotations.

Similarly to~\ref{ex:swap1}, a client program typable in our system cannot
cause the coordinator to deadlock by incorrectly ordering communication events.
\end{hide}
\end{example}

\begin{example}[A Database Library]
\label{ex:db}
In this example we consider the implementation of a library which allows clients to connect to a database.

\begin{minipage}{.45\textwidth}
\begin{lstlisting}[backgroundcolor=\color{white},numbers=none]
let fun coord(_) =
  let val p = $ \eaccept{}{\keyw{db}} $ $ \cunit $
      fun loop(_) = $\kmatch$ p {
        $\keyw{QRY}$: let val sql = $\erecv{}{}$ p
                val$\,$res = process sql
              in  $\esend{} $ p res; loop ()
        $\keyw{END}$: () }
  in spawn coord; loop ()
in spawn coord;
\end{lstlisting}
\end{minipage}
\hfill
\begin{minipage}{.45\textwidth}
\begin{lstlisting}[backgroundcolor=\color{white},numbers=none]
let fun clientinit(_) =
  let val con = $ \erequest{}{\keyw{db}} $ $ \cunit $
      fun query(sql) = $\eselectnew{\keyw{QRY}}$ con;
                       $\esend{\keyw{con}}{}$ sql;
                       $ \erecv{}{}$ con
      fun close(_) = $\eselectnew{\keyw{END}}$ con
  in (query$,$ close)
in $ e_{\mathsf{client}}$
\end{lstlisting}
\vspace{1em}
\end{minipage}

The coordinator accepts connections from clients on channel $\keyw{db}$. If a connection is established, after spawning a copy
of the coordinator to serve other clients, the coordinator enters a loop that serves the connected client. In this loop it offers
the client two options: $\keyw{QRY}$ and $\keyw{END}$. If the client selects $\keyw{QRY}$, the coordinator receives an SQL
query, processes it (calling $\keyw{process}:\keyw{sql}\rightarrow\keyw{dbresult}$, with these types are
defined in the library), sends back the result, and loops. If the client selects $\keyw{END}$ the connection with
the coordinator closes and the current coordinator process terminates.

Function $\keyw{clientinit}$ is exposed to the client, which can use it to request a connection with the database coordinator.
When called, it
establishes a connection $\keyw{con}$ and returns two functions to the client: $\keyw{query}$ and
$\keyw{close}$.
Then,
the client code $e_\mathsf{client}$
can apply the $\keyw{query}$ function to an $\keyw{sql}$ object and receive a
$\keyw{dbresult}$ as many times as necessary, and then invoke
$\keyw{close}$ to close the connection.
Using our two-level inference system with recursion \cref{sec:extensions}, we can infer the session type of the coordinator's
endpoint $\keyw{p}$:
$
\mu X. \Sigma\{ \inp \keyw{QRY}.\inp \keyw{sql}.\outp \keyw{dbresult}.X,~
\inp \keyw{END}.\tend \} $,
and check whether the client code $e_\mathsf{client}$ respects it.

This example is not typable with a substructural type system because
$\keyw{query}$ and $\keyw{close}$ share the same (linear) endpoint $\keyw{con}$.
Moreover, in a monadic system $e_\mathsf{client}$ will need to be converted to monadic form.
\begin{hide}
With out approach $e_\mathsf{client}$ can be an arbitrary functional expression (respecting the session protocol of \keyw{con}).
\end{hide}
\end{example}

  \section{Syntax and Operational Semantics of \lang}
  \label{sec:lang}
  
\begin{myfigure*}{\lang syntax and operational semantics.}{fig:lang}
\[\begin{array}{@{}l@{\;}r@{~}c@{~}l@{}}
    \textbf{Exp:}  &  e  & ::= & v
                                      \bnf  \epair{e}{e}
                                      \bnf  \eapp{e}{e}
                                      \bnf  \elet{x}{e}{e}
                                      \bnf  \eif{e}{e}{e} 
                                      \bnf  \espawn{e}
                                      \bnf \ematch{e}{L_i : e_i}{i \in I} 
\\[.5ex]
\textbf{Sys:} &  S  & ::= &  e \bnf S \parallel S 
\\[.5ex]
\textbf{Val:}  &  v  & ::= & 
                              x   
                              \bnf k \in \mathsf{Const}
                              \bnf \epair{v}{v}
                              \bnf \efunc{x}{e}
                              \bnf \efix{f}{x}{ e}
                              \bnf \eendp{p}{}
                              \\&&      \bnf &  \erequest{}{c} 
                              \bnf \eaccept{}{c}
                              \bnf \esend{}
                              \bnf \erecv{}{} 
                              \bnf  \eselectnew{L}
                              \bnf \edeleg{}
                              \bnf \eresume{}{}
\\[.5ex]
  \textbf{ECxt:}  &  E  & ::= & \hole
                                        \bnf  \epair{E}{e}
                                        \bnf  \epair{v}{E}
                                        \bnf  \eapp{E}{e}
                                        \bnf  \eapp{v}{E}
                                        \bnf  \elet{x}{E}{e}
                                        \bnf  \eif{E}{e}{e} 
                                        \\&&\bnf&  \espawn{E}
                                        \bnf \ematch{E}{L_i : e_i}{i \in I} 
\end{array}\]
\vspace{-2ex}
\[\begin{array}{%
    @{}l@{~}r@{\;}c@{\;}l@{~~}l@{~}l
    @{\qquad}
    @{}l@{~}r@{\;}c@{\;}l@{~~}l@{~}l
  }
    \irule![RIft][eift]{}{\eif{\ctrue}{e_1}{e_2} &\hookrightarrow& e_1}
  &
    \irule![RLet][elet]{}{\elet{x}{v}{e} &\hookrightarrow &e[v/x]}
    \\
    \irule![RIff][eiff]{}{\eif{\cfalse}{e_1}{e_2}  & \hookrightarrow& e_2}
                                                   &
    \irule![RFix][rfix]{}{\eapp{(\efix{f}{x}{e})}{v} &\hookrightarrow & e[\efix{f}{x}{e}/f][v/x]}
\end{array}\]
\vspace{-1ex}
\[\begin{array}{@{}l@{~}r@{\;}c@{\;}l@{~~}l@{~}ll@{~}r@{\;}c@{\;}l@{~~}l@{~}l@{}}
    \irule![RBeta][beta]{}{
      E[e] \parallel S &\arrow& E[e'] \parallel S
      \hfill
      \text{if } e \hookrightarrow e'
    }
    \\
    \irule![RSpn][spawn]{}{
      E[\espawn{v}] \parallel S &\arrow& E[()] \parallel \eapp{v}{()} \parallel S 
    }
    \\
    \irule![RInit][init]{}{
      E_1[\erequest{}{c}\;\cunit] \parallel E_2[\eaccept{}{c}\;\cunit] \parallel S 
      &\arrow&
      E_1[\eendp{p}{}] \parallel E_2[\eendp{\co p}{}] \parallel S 
      \quad
      \text{if } p,\co p \text{ fresh}
    }
    \\
    \irule![RCom][comm]{}{
      E_1[\esend{(\eendp{p}{}},{v})] \parallel E_2[\erecv{}{\eendp{\co p}{}}] \parallel S 
      &\arrow&
      E_1[()]   \parallel    E_2[v] \parallel S 
    }
    \\
    \irule![RDel][deleg]{}{
      E_1[\edeleg{(\eendp{p}{}}, {\eendp{p'}{})}] \parallel E_2[\eresume{}{\eendp{\co p}{}} ] \parallel S 
      &\arrow&
      E_1[()]\parallel E_2[\eendp{p'}{}] \parallel S 
    }
    \\
    \irule![RSel][select]{}{
      E_1[\eapp{\eselectnew{L_j}}{p}]
      \parallel
        E_2[\ematch{\eendp{\co p}{}}{L_i : e_i}{i\in I}] \parallel S 
      &\arrow&
      E_1[()]\parallel E_2[e_j] \parallel S 
      \hfill
      \text{if } j\in I
    }
\end{array}\]
\vspace{-1.2em}
\end{myfigure*} 

\cref{fig:lang} shows the syntax and operational semantics of \lang, a core of ML with session communication.
An 
expression can be one of the usual lambda expressions
or $\espawn{e}$ which evaluates $e$
to a function and asynchronously applies it to the unit value; it can also be
$\ematch{e}{L_i : e_i}{i \in I}$ which, as we will see, implements
finite \emph{external choice}.
We use standard syntactic sugar for writing
programs.
A \emph{system} $S$ is a parallel composition of 
closed expressions \emph{(processes)}.
\begin{hide}
A
running \emph{process} in \lang is a closed expression and a running \emph{system} $S$ is a parallel composition of
processes. We identify systems up to the reordering of parallel processes and
the removal of terminated, unit-value processes. A single process with no
active session endpoints is an \lang \emph{program}.
\end{hide}

\begin{icalp}
The operational semantics of \lang are standard; here we only discuss session-related rules.
\end{icalp}
\begin{hide}
The values of \lang contain the basic $\tunit$, $\tbool$, and $\tint$ constants and all standard integer and boolean operators.
Values also include pairs $(v,v')$, and first class recursive ($\efix{f}{x}{e}$) and non-recursive ($\efunc{x}{e}$) functions.
\end{hide}
Following the tradition of binary session types \cite{HondaEtal98}%
,
communication between processes happens over dynamically generated entities
called \emph{sessions} which have exactly two \emph{endpoints}. Thus, \lang values contain a countably infinite set of endpoints, ranged over by $p$. We assume a total involution ($\co{\cdot}$) over this set, with the property $\co p \not=p$, which
identifies \emph{dual endpoints}.
\begin{hide}
We write $o\freshfrom o'$ when the syntactic objects $o$ and $o'$ contain distinct endpoints.
\end{hide}

\begin{hide}
The language is equipped with a small-step, call-by-value operational semantics. 
\cref{fig:lang} \todo{cref conflicts with lipics style}
shows the \emph{redex}
expressions that perform beta reductions $\hookrightarrow$. Systems take small-step transitions $\arrow$ by decomposing a
system into an evaluation context $E$ with a beta redex in its hole (\iref{beta}), or by the effectful transitions discussed
below. Evaluation contexts include standard call-by-value contexts, but also the parallel system contexts $E \parallel S$ and
$S\parallel E$. An \iref{spawn} reduction generates new processes containing a single application.
\end{hide}

A process can request (or accept) a new session by calling $\erequest{}{c}$ (resp., $\eaccept{}{c}$) with the unit value, which
returns the endpoint (resp., dual endpoint) of a new session. Here $c$ ranges over an infinite set of global initialisation 
\emph{channels}. To simplify presentation, the language contains
$\erequest{}{c}$ and $\eaccept{}{c}$ for each channel $c$.

Once two processes synchronise on a global channel and each receives a fresh, dual endpoint (\iref{init} reduction),
they can exchange messages
\begin{icalp}
(\iref{comm}), \emph{delegate} endpoints (\iref{deleg}) and offer a number 
of choices $ L_{i \in I}$, from which the partner can select
one (\iref{select}).
Here $L$ ranges over a countably infinite set of
choice labels, and $I$ is a finite set of natural numbers;  $L_i$ denotes a unique label for each natural
number $i$ and we assume $\eselectnew{L_i}$ for each $L_i$.

The next two sections present the two-level type system of \lang.


\end{icalp}
\begin{hide}
. 
Applying $\esend{}$ to an
endpoint and a value will send this value over the session, whereas applying
$\erecv{}{}$ to an endpoint will receive a value over the session. The
synchronisation of these two applications leads to a synchronous communication
reduction (\iref{comm}).

A process can also offer a number of options to its dual with the construct
$\ematch{e}{L_i : e_i}{i \in I}$, implementing, as mentioned earlier, a finite external choice. Here $L$ ranges over a
countably infinite set of choice labels, and $I$ is a finite set of natural numbers. We assume a fixed enumeration of these
labels, thus $L_i$ denotes a unique label for each natural number $i$.
When a process offers a choice of labels on session endpoint $p$, its dual can select one of those labels with the expression
$\eapp{\eselectnew{L_j}}{\co p}$ (\iref{select}).

The intuition of session types is that once a session is open, the processes controlling its endpoints are in charge of executing
a predefined communication protocol. However, any of these processes may \emph{delegate} this obligation to another process, by
sending one endpoint $p$ over another endpoint $q$, with the application $\eapp{\edeleg{q}}{p}$. A process with an endpoint $\co q$ can
receive $p$ with the expression $\eresume{}{q}$, and continue executing the protocol over $p$ (\iref{deleg}). 

The example of the swap service in \ref{ex:swap2} used delegation to create a
direct connection between processes that run the swap function.
\end{hide}
  \begin{hide}
We give a type system for \lang organised in two levels.
The first level is a type-and-effect system adding communication effects and
endpoint flow-analysis to the Hindley-Milner type system.
This level extends the system of Amtoft, Nielson and Nielson \cite{ANN}.
The second level imposes a session discipline on the communication
effects of the first level.
\end{hide}

\section{First Level Typing: ML Typing and Behaviours} \label{sec:stage-1} 

\begin{myfigure*}{Syntax of types, behaviours, constraints, and session types.}{fig:type-syntax}
\[
  \begin{array}{@{}r@{~~}r@{\;}c@{\;}l}
    \textbf{Variables:} &
\multicolumn{3}{l}{%
  \alpha \textbf{(Type)}
  \quad
  \beta \textbf{(Behaviour)}
  \quad
  \psi \textbf{(Session)}
  \quad
  \rho \textbf{(Region)}
}
\\[.3ex]
\textbf{T.\ Schemas:} &  \TS  & ::= &  \typeschema{\vec\alpha\vec\beta\vec\rho
                                    \vec\psi
                                    }{C}{T}
\qquad \qquad \qquad \textbf{Regions:} ~~ r ~ ::= ~  l \bnf \rho
\\
  \textbf{Types:}            &  T  & ::= & \tunit
                                   \bnf  \tbool
                                   \bnf  \tint
                                   \bnf  \tpair{T}{T}
                                   \bnf  \tfun{T}{T}{\beta}
                                   \bnf  \tses{\rho}
                                   \bnf  \alpha
\\[.3ex]
\textbf{Constraints:} &  C  & ::= &       \su{T}{T}
                                  \bnf \cfd T
                                  \bnf \su{b}{\beta}
                                  \bnf \seql{\rho}{r}
                                  \bnf \seql{c}{\eta}
                                  \bnf \seql{\co c}{\eta}
                                  \bnf \eta \bowtie \eta  
                                  \bnf C,C \bnf \epsilon
\\[.3ex]
\textbf{Behaviours:}         &  b  & ::= &  \beta
                                       \bnf \tau
                                      \bnf \bseq{b}{b}
                                       \bnf \bichoice{b}{b}
                                       \bnf \orec{b}{\beta}
                                       \bnf \espawn{b}
                                       \bnf \pusho{l}{\eta}{}
\\[.3ex]
                        &&       \bnf &  \popo{\rho}{\outp T}
                                       \bnf \popo{\rho}{\inp T}
                                       \bnf \popo{\rho}{\outp \rho}
                                       \bnf \popo{\rho}{\inp l}
                                       \bnf \popo{\rho}{\outp L_i}
                                       \bnf \bechoice{}{\rho}{L_i}{b_i}
\\[.3ex]
\textbf{Type Envs:} &  \Gamma  & ::=  & x: \TS \bnf \Gamma,\Gamma\bnf\epsilon
\end{array}
\]
\end{myfigure*} 

Here we adapt and extend the type-and-effect system of Amtoft, Nielson and Nielson \cite{ANN} to session communication in \lang. A
judgement $\tjr{}{}{e}{T}{b}$ states that $e$ has type $T$ and behaviour $b$, under type environment $\Gamma$ and \emph{constraint
environment} $C$. The constraint environment relates type-level variables to terms and enables type inference. These components
are defined in \cref{fig:type-syntax}.

An \lang expression can have a standard type or an endpoint
type $\tses{\rho}$. Function types are annotated with a behaviour variable $\beta$.
\emph{Type variables} $\alpha$ are used for ML polymorphism.
As in~\cite{ANN}, Hindley-Milner polymorphism is extended with type schemas $\TS$ of the form
$\typeschema{\vec\gamma}{C_0}{T}$, where $\gamma$ ranges over variables
$\alpha,\beta,\rho,\psi$, and
$C_0$ imposes constraints on the quantified variables
with $\fv(C_0)\subseteq\{\vec\gamma\}$.
Type environments $\Gamma$ bind program variables to type schemas;
we let $\forall(\emptyset).T = T$.

\begin{myfigure*}{Type-and-Effect System for \lang Expressions (omitting rule for pairs).}{fig:typing-rules} 

\[  \begin{array}{@{}l@{\quad}l@{}}
\irule*[TLet][tlet]
  {\tjr{}{}{e_1}{\TS}{b_1}
    \quad \tjr{%
               }{,x:\TS}{e_2}{T}{b_2} }
  {\tjr{}{}{\elet{x}{e_1}{e_2}}{T}{\bseq{b_1}{b_2} } }
  &
\irule*[TVar][tvar]
  {}
  {\tjr{}{} {x} {\Gamma(x)} {\tau}}
\\[1em]
\irule*[TIf][tif]
  {%
    \tjr{}{}{e_1}{\tbool}{b_1}
    \quad
    \tjr{}{}{e_i}{T}{b_i} {~}_{(i\in\{1,2\})}
  }
  {\tjr{}{}{\eif{e_1}{e_2}{e_3}}{T}{\bseq{b_1}{(b_2 \oplus b_3)} } }
&
\irule* [TConst][tconst]
  {}
  {\tjr{}{}{k }{\mathit{typeof}(k)}{\tau} }
\\[1.4em]
\irule*[TApp][tapp]
  {\tjr{}{}{e_1}{\tfun{T'}{T}{\beta} }{b_1} \quad \tjr{}{}{e_2}{T'}{b_2} }
  {\tjr{}{}{\eapp{e_1}{e_2}}{T}{\bseq{\bseq{b_1}{b_2}}{\beta} } }
&
\irule*[TFun][tfun]
  {\tjr{}{,x:T}{e}{T'}{\beta} }
  {%
  \tjr{}{}
      {\efunc{x}{e}}
      {\tfun{T}{T'}{\beta}}{\tau}}
\\[1em]
\irule*[TMatch][tmatch]
  {%
    \tjr{}{}{e}{\tses{\rho}}{b}
    \quad \tjr{}{}{e_i}{T}{b_i}
    {~}_{(i\in I)}
  }
  {%
    \tjr{}{}
      {\ematch{e}{L_i : e_i}{i \in I} }
      {T}{\bseq{b}{\bechoice{}{\rho}{L_i}{b_i}} }
  }
&
\irule*[TEndp][tendp]
  {}
  {\tjr{}{}{\eendp{p}{l} }{\tses{\rho}{} }{\tau }}
  ~\condBox{$C\vdash\seql{\rho}{l}$}
\\[1.4em]
\irule[TSpawn][tspawn]{%
  \tjrBase {C} {\Gamma}
  {e}{\tfun{\tunit}{\tunit}{\beta}}{b}{}{}
}{%
  \tjr{}{}
      {\espawn{e}}
      {\tunit}
      {\bseq{b}{\espawn{\beta}}}
}
&
\irule*[TSub][tsub]
  {\tjr {} {} {e} {T} {b} }
  {\tjr {} {} {e} {T'} {\beta} }
  ~\condBox{\nbox[m]{%
      \coType{T}{T'}\\
      \sub{b}{\beta}
  }}
%
%
%
\\[1em]
\multicolumn{2}{@{}l@{}}{%
\irule[TRec][trec]
  {\tjrBase {C} {%
      \Gamma, f: \tfun{T}{T'}{\beta}, x: T
    }{%
      e
    }{T'}{b} {}{}
  }
  {\tjrBase {C} {\Gamma} {\efix
  {f} {x} {e} } {\tfun{T}{T'} {\beta}} {\tau} {}{}
}[~\condBox{\nbox[m]{%
      \under{C}{\confined(T,T')}\\
      \sub{\orec{b}{\beta}}{\beta}
    }
}]
}
\\[1.4em]
\multicolumn{2}{@{}l@{}}{%
\irule[TIns][tins]
  {\tjr{}{} {e} {\forall (\vec\gamma:C_0).T} {b} }
  {\tjr{}{} {e} {T\sigma}         {b} }
  [~\condBox{\nbox[m]{%
      dom(\sigma) \subseteq \{\vec\gamma\}\\
      \typeschema{\vec\gamma}{C_0}{T} \text{~is solvable by $ C $ and } \sigma
  }}]
}
\\[1.4em]
\multicolumn{2}{@{}l@{}}{%
\irule[TGen][tgen]
  {\tjr{\cup C_0}{} {e} {T}         {b} }
  {\tjr{}{} {e} {\forall (\vec\gamma:C_0).T} {b} }
  [~\condBox{\nbox[m]{%
      \{\vec\gamma\} \cap \fv(\Gamma, C, b) = \emptyset\\
      \forall (\vec\gamma:C_0).T  \text{~is solvable by $C$ and some } \sigma
  }}]
}
\end{array}\]
\end{myfigure*} 

The rules of our type-and-effect system are shown in \cref{fig:typing-rules}
which, as in \cite[Sec. 2.8]{ANN}, is a conservative extension of ML\@.
This system performs both ML type checking (including type-schema inference),
and behaviour checking (which enables behaviour inference). Rules
\iref{tlet}, \iref{tvar}, \iref{tif}, \iref{tconst}, \iref{tapp},
\iref{tfun}, \iref{tspawn} and the omitted rule for pairs perform standard type checking and
straightforward sequential ($\bseq{b_1}{b_2}$) and non-deterministic ($b_1 \oplus b_2$) behaviour composition;
$\tau$ is the behaviour with no effect.

Just as a type constraint $T \subseteq \alpha$ associates type $T$ with type variable $\alpha$, a behaviour constraint
$b \subseteq \beta$ associated behaviour $b$ to behaviour variable $\beta$. Intuitively, $\beta$ is the non-deterministic
composition of all its associated behaviours.
Rule \iref{tsub} allows the replacement of behaviour $b$ with variable $\beta$; such replacement in type annotations
yields a subtyping relation ($\coType{T}{T'}$).
Rules \iref{tins} and \iref{tgen} are taken from \cite{ANN} and extend ML's type schema instantiation and
generalisation rules, respectively.
Because we extend Hindley-Milner's let polymorphism, generalisation (\iref{tgen}) is only applied to the right-hand side
expression of the \keyw{let} construct.
The
following definition allows the instantiation of a type schema under a global
constraint environment $C$.
We write $C \vdash C'$ when $C'$ is included in the reflexive, transitive,
compatible closure of $C$.
\begin{definition}[Solvability]\label{def:solvability}
  $\typeschema{\vec\gamma}{C_0}{T}$ is \emph{solvable} by $C$ and substitution $\sigma$ when
  $\dom{\sigma}\subseteq\{\vec\gamma\}$ and $C \vdash C_0\sigma$.
\end{definition}

In \iref{trec}, the communication effect of the body of a recursive function should be \emph{confined},
which means it may only use endpoints it opens internally. For this reason, the function does not input nor return open endpoints
or other non-confined functions ($C \vdash \confined(T,T')$). Although typed under $\Gamma$ which may contain endpoints and
non-confined functions, the effect of the function body is recorded in its behaviour.
The second level of our system checks that if the function is called, no endpoints from its environment are affected.
It also checks that the function fully consumes internal endpoints
before it returns or recurs.

\begin{icalp}
A type $ T $ is confined
when it does not contain any occurrences of the endpoint type $ \tses \rho $ for
any $\rho$, and when any $ b $ in $ T $ is confined. A behaviour $ b $ is
confined when all of its possible behaviours are either $ \tau $ or
recursive.
\end{icalp}

\begin{hide}
Definition \ref{def:confined-types} (Confined Types) does not enforce the second
condition for confinement; i.e., it does not check whether the net communication effect of recursive behaviours is neutral.
This is done in the session typing discipline of the second level of our system
(\cref{sec:session-types}), and in fact only for those recursive functions that
are applied and whose behaviour is part of the communication effect of the program. This separation between the two levels simplifies type inference.
\end{hide}

To understand rule \iref{tendp}, we have to explain region variables ($\rho$), 
which are related to region constants through $C$.
Region constants are simple program annotations $l$ (produced during pre-processing) which uniquely identify the textual sources of
endpoints. We thus type an extended \lang syntax
\[
  \small
  \textbf{Values:}\qquad    v   ::= {\ldots}
                                \bnf \eendp{p}{l}
                                \bnf \erequest{l}{c}
                                \bnf \eaccept{l}{c}
                                \bnf \eresume{l}{}
\]
If a sub-expression has type $ \tses{\rho}$ and it evaluates to a
value $ \eendp{p}{l} $, then it must be that $C \vdash \rho \sim l$, denoting that
$p$ was generated from the code location identified by $l$.
This location will contain one of
$ \erequest{l}{c}$, $\eaccept{l}{c}$, or $\eresume{l}{}$.
These primitive functions (typed by \iref{tconst}) are given the following type schemas.
\\[1ex]$
~\hfill~
  \small
\begin{array}{@{}l@{\,:\,}l@{}}
\erequest{l}{c}   & \typeschema {\beta\rho\psi }
                                     {%
                                       \su{\pusho{l}{\psi}{%
                                       }
                                       }{\beta},
                                       \seql{\rho}{l},
                                       \seql{c}{\psi}
                                     }
                                     {%
                                      \tfun{\tunit}
                                      {\tses{\rho} }
                                        {\beta }
                                     }
\\ \eaccept{l}{c}    & \typeschema
                              {\beta\rho\psi}
                              {%
                                \su {\pusho{l}{\psi}{%
                                    }
                                } {\beta},
                                \seql{\rho}{l},
                                \seql {\co{c}}{\psi}
                              }
                              {%
                                \tfun{\tunit}{\tses{\rho}}{\beta}
                              }
\\ \eresume{l}{}             \qquad&
                              \typeschema{\beta\rho\rho'}{%
                                \su{\popo{\rho\inp }{l}}{\beta},~
                                \seql{\rho'}{l}
                              }{%
                                \tfun{\tses{\rho}}{\tses{\rho'}}{\beta}
                              }
\end{array}
~\hfill~
$\\[1ex]
An application of $\erequest{l}{c}$ starts a new session on the static endpoint $l$.
To type it, $C$ must contain its effect $\pusho{l}{\psi}{} \subseteq \beta$, where
$\psi$ is a session variable, representing the session type of $l$.
At this level session types are ignored (hence the use of a simple $\psi$); they become important in the second
level of our typing system. Moreover, $C$ must record that session variable $\rho$ is related to $l$ ($\rho\sim l$) and that
the ``request'' endpoint of channel $c$ has session type $\psi$ ($c\sim\psi$).
The only difference in the type schema of $\eaccept{l}{c}$ is that the ``accept'' endpoint of $c$ is related to $\psi$ ($\co c
\sim \psi$). Resume receives an endpoint ($\rho'$) over another one ($\rho$), recorded in its type schema
($\su{\popo{\rho\inp }{\rho'}}{\beta}$); $\rho$ is an existing endpoint but $\rho'$ is treated as an endpoint
generated by $\eresume{l}{}$, hence the constraint $\seql{\rho'}{l}$.

The following are the type schemas of the rest of the constant functions.
\\[1ex]$
~\hfill~
  \small
\begin{array}{@{}l@{\,:\,}l@{}}
   \erecv{}{}           &
                              \typeschema{\alpha\beta\rho}{%
                                \su{\popo{\rho}{\inp \alpha} }{\beta},
                                \cfd\alpha
                              }{%
                                \tfun{\tses{\rho}}{\alpha}{\beta}
                              }
\\ \esend{}              &
                              \typeschema{\alpha\beta\rho}{%
                                \su{\popo{\rho}{\outp \alpha}}{\beta},
                                \cfd\alpha
                              }{%
                                \tfun{\tpair{\tses{\rho}}{\alpha}}{\tunit}{\beta}
                              }
\\ \edeleg{}             &
                              \typeschema{\beta\rho\rho'}{%
                                \su{\popo{\rho}{\outp \rho'}}{\beta}
                              }{%
                                \tfun{\tpair{\tses{\rho}}{\tses{\rho'}}}{\tunit}{\beta}
                              }
\\ \eselectnew{L}       &
                              \typeschema{\beta\rho}{%
                                \su{\popo{\rho}{\inp L} }{\beta}
                              }{%
                                \tfun{\tses{\rho}}{\tunit}{\beta}
                              }
\end{array}
~\hfill~
$\\[1ex]
These record input ($\popo{\rho}{\inp \alpha}$), output ($ \rho\outp \alpha$), delegation ($\popo{\rho}{\outp \rho'}$), or selection
($\popo{\rho}{\outp L_i}$) behaviour. For input and output the constraint $\cfd{\alpha}$ must be in $C$, recording that the
$\alpha$ can be instantiated only with confined types.

\begin{hide}
\paragraph{Static Endpoints.}
First we require that textual sources of session endpoints to be
marked  with \emph{unique region labels} in a
pre-processing step, updating \lang syntax as follows.
\\[1ex]
  $\textbf{Values:}\qquad    v   ::=  {\ldots} \bnf \eendp{p}{l} \bnf
  \erequest{l}{c} \bnf \eaccept{l}{c} \bnf \eresume{l}{}$
\\[1ex]
With this extension, the system uses regions to statically approximate the endpoint that will be used at each communication at
runtime, effectively creating one type $\tses{l}$ for each endpoint source.\footnote{Technically, endpoint types are annotated
with session variables $\tses{\rho}$ which are related to endpoint labels through the  constraint environment $C$.} The
type-and-effect system uses \emph{region variables} $\rho$ to track the flow of labels at the type level. %
Dynamic endpoints generated at different source expressions are statically distinguished, but those generated from the same
expression are identified, resulting to the rejection of some type safe programs
(\cref{sec:session-types}). This can be remedied with standard context-sensitive solutions such as $k$-CFA \cite{kcfa}.

\paragraph{Functional Types.}
An expression in \lang can have a base type $\tunit$, $\tbool$, or $\tint$, a pair type $\tpair{T}{T'}$, a function type
$\tfun{T}{T'}{\beta}$, or a session endpoint type $\tses{\rho}$. \emph{Type variables} $\alpha$ are used for polymorphism (and type
inference). Each function type is annotated with a \emph{behaviour variable} $\beta$ and each session endpoint type with a
region variable $\rho$, respectively denoting the effect of the function body and the endpoint's textual source.

The types of session endpoints do not contain a session type because session types evolve during the
execution of the program.
In \ref{ex:swap1}, the two uses of $z$ in the body of $\mathit{swap}$ refer to the same endpoint but at
different states: at the first it can perform a send and then a receive, and at the second it can only perform the receive.
Therefore,
$\tses{\rho}$ only refers to the static identity of an endpoint through $\rho$,
ignoring its session type.

\paragraph{Communication Effects.}
Inspired by Castagna et al.\ \cite{CastagnaEtal09}, the behavioural effect of a \lang expression can be thought of as describing
operations on a \emph{stack of session endpoints} $\Delta$. This stack contains frames of the form $(l:\eta)$, where $l$ is a
static endpoint and $\eta$ a session type (described in
\cref{sec:session-types}).

The expression can push a new frame on the
stack ($\pusho{l}{\eta}{}$), or reduce the top session type by performing an input ($\popo{\rho}{\inp T}$) or output ($\popo{\rho}{\outp T}$)
of a value; a delegation ($\popo{\rho}{\outp \rho}$) or resumption ($\popo{\rho}{\inp l}$) of an endpoint; or an offer
($\bechoicetext{}{\rho}{L_i}{b_i}$) or selection ($\popo{\rho}{\outp L_i}$) of a choice. When the top session type of $\Delta$ is finished
then it is popped from the stack. The application of $\erequest{l}{c}$ or $\eaccept{l}{c}$ has a $\pusho{l}{\eta}{}$ effect; the
application of $\ksend$, $\krecv$, and $\kdeleg$, has the corresponding effect with $\rho$ calculated by the type of the first
argument ($\tses{\rho}$).
Departing from \cite{ANN}, function $\kresume$ is annotated with its own fresh label, instead of a variable which would be mapped
through the constraint environment $C$ to a single label. This allows typing programs where a $\kresume$ statement
inputs different endpoints.

\begin{example}
  Consider the following program $P$ that spawns two clients, one proxy, and one server; $P$
  is typable in \lang, provided $e_{1}$ and $e_{2}$ are.
  \[
  \begin{array}{@{}l@{\,}l@{\,}l@{}}
      \klet
      & \mathit{cli}_1 &= (\efunc{\_}{\elet{z_c}{\eapp{\erequest{l_{1}}{\keyw{c}}}{\cunit}} {e_{1}}}) \\
      & \mathit{cli}_2 &= (\efunc{\_}{\elet{z_c}{\eapp{\erequest{l_{2}}{\keyw{c}}}{\cunit}} {e_{2}}}) \\
      & \mathit{prx}  &= (\efunc{\_}{%
      \eletmany{z_{\mathit{c}} & = \eapp{\eaccept{l_c}{\keyw{c}}}{\cunit}\\
                z_{\mathit{s}} & = \eapp{\erequest{l_s}{\keyw{s}}}{\cunit}
      }{%
        \eapp{\edeleg{z_{\mathit{s}}}}{z_{\mathit{c}}}
    )} } \\
    & \mathit{srv}  &= (\efunc{\_}{%
      \eletstarmany{z_s &= \eapp{\eaccept{l_p}{\keyw{s}}}{\cunit} \\
                    x   &= \underline{\eresume{l}{z_s}}
      }{%
        \eapp{\esend{x}}{\keyw{1}};
        \eapp{\esend{x}}{\keyw{tt}}
    )}} \\
      \keyw{in} &
      \multicolumn{2}{@{}l@{}}{\!\!
      \espawn{\mathit{prx}_1};
      \espawn{\mathit{prx}_2};
      \espawn{\mathit{cli}};
      \espawn{\mathit{srv}};
      }
  \end{array}
\]
  Both clients request a session on $\keyw{c}$. The proxy accepts one of them and in turn requests a session with the
  server on $\keyw{s}$. Once the server accepts, the proxy delegates the client session over the server session. The server then sends two values to the connected client over the client session.

  In the absence of label $l$ on the underlined resume construct of the server, the type system would calculate that both $l_1$
  and $l_2$ endpoints can flow to $x$ at runtime. Therefore, $x$ will have type $\tses{\rho}$, with $\rho$ related to both $l_1$
  and $l_2$ in the constraint environment, which violates the unique-label requirement.
\end{example}

The remaining behaviours follow the structure of the code, as can be noted in
\cref{fig:typing-rules}:
$\tau$ is the silent behaviour of pure computations, $\bseq{b}{b'}$
allows sequencing, and $\bichoice{b}{b'}$ internal choice. Behaviour $\orec{b}{\beta}$ marks recursive behaviour.
Even though \lang does not have recursive session types, it does allow
recursive effectful functions such as the coordinator in Ex. \ref{ex:swap1}.
\end{hide}

\begin{hide} 
\paragraph{Constraints.}
Constraint sets $C$ have \emph{inclusion constraints} for types ($\su{T}{T'}$) and behaviours
($\su{b}{\beta}$), and equality constraints for regions ($\seql{\rho}{r}$). They also contain exactly two equality constraints
($\seql{c}{\eta}$ and $\seql{\co c}{\eta'}$) per global channel~$c$, one for the behaviour of the runtime endpoints returned by
$\eaccept{}{c}$ and one for those returned by $\erequest{}{c}$. These are the two \emph{kinds of endpoints} of $c$.

The simple inclusion constraints for behaviours are sufficient for typing any functionally-typable program without the
introduction of a sub-effecting relation~\cite{TalpinJ92}. Intuitively, $b \subseteq \beta$ means that $\beta$ \emph{may} behave
as $b$. Although not strictly necessary here, type constraints are more general to enable principal typing in the presence of
subtyping constraints $\su{\tint}{\keyw{real}}$ \cite[\S1.5.1]{ANN}, often appearing in the session-types literature. We write $C
\vdash \su{o}{o'}$ and $C \vdash \seql{o}{o'}$ for the reflexive and transitive closure of constraints in $C$; we say that these
constraints are \emph{derivable} from $C$. We write $C \vdash C'$ if all constraints of $C'$ are derivable from $C$.
The constraints $ \cfd T $ and $ \cfd b $ indicate that type $ T $
and behaviour $ b $ must be \emph{confined}. A type $ T $ is confined
when it does not contain any occurrences of the endpoint type $ \tses \rho $ for
any $\rho$, and when any $ b $ in $ T $ is confined. A behaviour $ b $ is
confined when all of its possible behaviours are either $ \tau $ or
recursive. The predicates $ C \vdash \confined(T) $ and $ C \vdash \confined(b)
$ are defined structurally as follows:
\todo{Hidden all definitions of confined. I'll add a brief description (i.e.
behaviours and types where no channel labels flow)}
\end{hide}

\begin{hide} 
\begin{definition}[Confined Behaviors]\label{def:confined-behav}
  $C \vdash \confined(o)$ is the least compatible relation on type schemas, types, and
  behaviours that admits the following axioms.
\begin{longtable}{@{}l@{\qquad}l@{}}
  \irule*[CTau][CTau]
    {b \in \{\tau,~\orec{b'}{\beta}\}}
    {C\vdash \confined(b)}
&
  \irule*[CAx-b][CAx-b]
    {\cfd b \in C }
    {C \vdash \confined(b) }
 \\[1.4em]
  \irule*[CICh][CICh]
    {C\vdash\confined(b_1) \quad
     C\vdash\confined(b_2)}
    {C\vdash\confined(b_1 \oplus b_2)}
&
  \irule*[CICh-Bw][CICh-Bw]
    {C\vdash\confined(b_1\oplus b_2)}
    {C\vdash\confined(b_i)}
    \condBox{$i\in \{1,2\}$}
\\[1.4em]
  \irule*[CSeq][CSeq]
    {C\vdash\confined(b_1) \quad
     C\vdash\confined(b_2)}
    {C\vdash\confined(b_1; b_2)}
&
  \irule*[CSeq-Bw][CSeq-Bw]
    {C\vdash\confined(b_1;b_2)}
    {C\vdash\confined(b_i)}
    \condBox{$i\in \{1,2\}$}
  \\[1.4em]
  \irule*[CSpw][CSpw]
    {C\vdash\confined(b)}
    {C\vdash\confined(\spawno{b}{})}
&
  \irule*[CSpw-Bw][CSpw-Bw]
    {C\vdash\confined(\spawno{b}{})}
    {C\vdash\confined(b)}
\\[1.4em]
  \irule*[CSub-b][CSub-b]
    {\sub{b_1}{b_2} \quad C\vdash \confined(b_2)  }
    {C \vdash \confined(b_1) }
    &
\end{longtable}
\end{definition}
\begin{definition}[Confined Types]\label{def:confined-types}
  $C \vdash \confined(o)$ is the least compatible relation on type schemas, types, and
  behaviours that admits the following axioms.
%
\begin{longtable}{@{}l@{\quad}l@{}}
  \irule*[CCons][CCons]
      {T \in \{\tint,\tbool,\tunit\} }
      {C \vdash \confined(T) }
  &
  \irule*[CAx-T][CAx-T]
      {\cfd T \in C }
      {C \vdash \confined(T) }
\\[1.4em]
    \irule*[CSub-T][CSub-T]
      {\sub{T_1}{T_2} \quad C\vdash \confined(T_2)  }
      {C \vdash \confined(T_1) }
&
    \irule*[CSub-T-Bw][CSub-T-Bw]
      {\forall T_1. (\sub{T_1}{T_2}) \Longrightarrow C\vdash \confined(T_1)  }
      {C \vdash \confined(T_2) }\footnote{\textcolor{red}{TODO}: is this rule
      necessary\inp 
      Consider $ T_1 = \tfun {T} {T'} {\beta_1} $,
      $ T_2 = \tfun {T} {T'} {\beta_2} $ and $ C = \{\su \tau {\beta_1},
      \su {\beta_1}{\beta_2}, \su {\rho\outp \tint} {\beta_2}, \cfd{T_1},
      \su{T_1}{T_2} \}$
      By Rule \iref{CSub-T-Bw}, $ T_2 $ has to be confined; but $ T_2 $ is not
      confined because of the constraint $ \su {\rho\outp \tint} {\beta_2} $. So this
      constraint set is rejected as not well-formed. I think that it should be
      accepted though, I don't see any danger in allowing it to be well-formed.
      }.
\\[1.4em]
    \irule*[CFun][CFun]{%
      C\vdash \confined(T,T')
      \quad
      C\vdash \confined(\beta)
    }{%
      C\vdash \confined(\tfun{T}{T'}{\beta})
    }
&
    \irule*[CFun-Bw][CFun-Bw]{%
      C\vdash \confined(\tfun{T}{T'}{\beta}) \quad o \in \{T,T',\beta\}
    }{%
      C\vdash \confined(o)
    }
\\[1.4em]
    \irule*[CTup][CTup]{%
      C\vdash \confined(T,T')
    }{%
      C\vdash \confined(\tpair{T}{T'})
    }
&
    \irule*[CTup-Bw][CTup-Bw]
    {%
      C\vdash \confined(\tpair{T}{T'}) \quad o \in \{T,T'\}
    }{%
      C\vdash \confined(o)
    }
\\[1.4em]
   \irule*[CTS][CTS]{%
     C,C_0 \vdash \confined(T)
    }{%
      C\vdash \confined( \forall (\vec\gamma: C_0). T )
    }
    &
\end{longtable}
  Moreover, $\confined_C(\Gamma)$ is the largest subset of $\Gamma$ such that
  for all bindings $ (x:\TS) \in \confined_C(\Gamma)$ we have
  $C\vdash\confined(\TS)$.
\end{definition}

The above definitions admit behaviours constructed by $\tau$ and
recursive behaviours (\iref{CTau}), and types that are constructed by such
behaviours and the base types $\tint,\tbool,\tunit$ (\iref{CCons}).
The definition allows for type and behaviour variables ($\alpha$, $\beta$) as
long as they are only related to confined types in $C$
(\iref{CAx-b}, \iref{CAx-T}). Sub-behaviours of confined behaviours and
sub-types of confined types are confined too (\iref{CSub-b}, \iref{CSub-T}).
Type schemas are confined if all their instantiations in $C$ are
confined (\iref{CTS}).
The definitions contain composition rules for composite
behaviours and types (\iref{CICh}, \iref{CSeq}, \iref{CSpw}, \iref{CFun},
\iref{CTup}); they also contain decomposition (or ``\textit{backward}'') rules
for the same composite constructs (\iref{CICh-Bw}, \iref{CSeq-Bw}, \iref{CSpw-Bw}, \iref{CFun-Bw},
\iref{CTup-Bw}).
The backward rules ensure that all the sub-components of a
confined behaviour or type in $ C $ are also confined.
\end{hide}

\begin{hide} 
We will work with well-formed constraints, satisfying
the following conditions.

\begin{definition}[Well-Formed Constraints] \label{def:wf-constraints} 
  $C$ is \emph{well-formed} if:
  \begin{enumerate}
    \item \emph{Type-Consistent:} for all type constructors $tc_1$, $tc_2$, if $(\su{tc_1(\vec t_1)}{tc_2(\vec t_2)})\in C$, then
      $tc_1=tc_2$, and for all $t_{1i} \in \vec t_1$ and $t_{2i} \in \vec t_2$, $(\su{t_{1i}}{t_{2i}})\in C$;
    \item \emph{Region-Consistent:} if $C \vdash \seql{l}{l'}$ then $l=l'$;
    \item \emph{Behaviour-Compact:} all cycles in behaviour constraints contain at least one $(\su{\orec{b}{\beta}}{\beta})\in C$;
      also if $(\su{\orec{b}{\beta}}{\beta'})\in C$ then $\beta=\beta'$ and $\forall(\su{b'}{\beta})\in C$, $b'=\orec{b}{\beta}$.
    \item \emph{Well-Confined:} for any type $ T $ and region $ \rho$, $
    C \vdash \confined(T) $ implies $ T \neq \tses \rho $. For all behaviours $
    b $, $ C \vdash \confined(b) $ implies that $ b \not\in\{%
    \popo{\rho}{\outp T}, \popo{\rho}{\inp T}, \popo{\rho}{\outp \rho}, \popo{\rho}{\inp l},
    \popo{\rho}{\outp L_i},$ $ \bechoice{}{\rho}{L_i}{b_i}%
    \} $.
  \end{enumerate}
\end{definition}
The first condition disallows constraints such as ($\su{\tint}{\tpair{T}{T'}}$) which lead to type
errors, and deduces ($\su{T_i}{T_i'}$) from ($\su{\tpair{T_1}{T_2}}{\tpair{T_1'}{T_2'}}$).

The second condition
requires that only endpoints from a single source can flow in each $\rho$. This
condition rejects Example \ref{ex:region-flow}.

\begin{example}
  The following program requests two session endpoints, bound to $x$ and $y$. It then binds one of these endpoints to $z$,
  depending on the value of $e$, and sends $\keyw{1}$ over $x$ and $\keyw{tt}$ over $z$.
  \[
    \begin{array}{@{}l@{\;}r@{\;}l@{}}
      \klet*
      & (x, y) &= (\eapp{\erequest{l_{1}}{\keyw{c}}}{\cunit},\eapp{\erequest{l_{2}}{\keyw{d}}}{\cunit})\\
      & z &= \eif{e}{x}{y}\\
      \multicolumn{3}{@{}l@{}}{%
        \kin~
        \eapp{\esend{x}}{\keyw{1}};
        \eapp{\esend{z}}{\keyw{tt}}
    }
  \end{array}
\]
  This program is not typable because communications on the $\keyw{c}$- and $\keyw{d}$-endpoints depend on  the
  value returned from $e$, which cannot be statically determined. In our framework, $z$ will have type $\tses{\rho}$
  and the constrain environment will contain $C \vdash \rho \sim l_1,~ \rho\sim l_2$. The program will be rejected by the second
  condition of \ref{def:wf-constraints}.
\end{example}

In related work (e.g., \cite{GayHole05}) such a program is rejected because of the use of substructural types for session
endpoints.

The third condition of Def.~\ref{def:wf-constraints} disallows recursive behaviours through the environment without the use of a
$\orec{b}{\beta}$ effect. The second part of the condition requires that there is at most one recursive constraint in the
environment using variable $\beta$. This condition is necessary to guarantee type preservation and the decidability of session
typing.

The fourth condition enforces that all types and behaviours marked as $
\cfd{-} $ in $ C $ are confined.
\paragraph{Polymorphism.}
The type system extends Hindley-Milner polymorphism with type schemas $\TS$ of the form
$\typeschema{\vec\gamma}{C}{T}$, where $\gamma$ ranges over any variable $\alpha,\beta,\rho,\psi$.
Type environments $\Gamma$ bind unique variable names to type schemas; we let $\forall(\emptyset).T = T$.
Besides type ($\vec\alpha$), behaviour ($\vec\beta$), and region
($\vec\rho$) variables, type schemas also generalise \emph{session variables} $\vec\psi$. A type schema contains a set $C$ which
imposes constraints on quantified variables. For $\TS$ to be \emph{well-formed}, it must be
$\fv(C)\subseteq\{\vec\gamma\}$. The types schemas of the constant \lang functions are:
\[
\begin{array}{@{}l@{\,:\,}l@{}}
\erequest{l}{c}   & \typeschema {\beta\rho\psi }
                                     {%
                                       \su{\pusho{l}{\psi}{%
                                       }
                                       }{\beta},
                                       \seql{\rho}{l},
                                       \seql{c}{\psi}
                                     }
                                     {%
                                      \tfun{\tunit}
                                      {\tses{\rho} }
                                        {\beta }
                                     }
\\ \eaccept{l}{c}    & \typeschema
                              {\beta\rho\psi}
                              {%
                                \su {\pusho{l}{\psi}{%
                                    }
                                } {\beta},
                                \seql{\rho}{l},
                                \seql {\co{c}}{\psi}
                              }
                              {%
                                \tfun{\tunit}{\tses{\rho}}{\beta}
                              }
\\ \esend{}              &
                              \typeschema{\alpha\beta\rho}{%
                                \su{\popo{\rho}{\outp \alpha}}{\beta},
                                \cfd\alpha
                              }{%
                                \tfun{\tpair{\tses{\rho}}{\alpha}}{\tunit}{\beta}
                              }
\\ \erecv{}{}           &
                              \typeschema{\alpha\beta\rho}{%
                                \su{\popo{\rho}{\inp \alpha} }{\beta},
                                \cfd\alpha
                              }{%
                                \tfun{\tses{\rho}}{\alpha}{\beta}
                              }
\\ \eselectnew{L}       &
                              \typeschema{\beta\rho}{%
                                \su{\popo{\rho}{\inp L} }{\beta}
                              }{%
                                \tfun{\tses{\rho}}{\tunit}{\beta}
                              }
\\ \edeleg{}             &
                              \typeschema{\beta\rho\rho'}{%
                                \su{\popo{\rho}{\outp \rho'}}{\beta}
                              }{%
                                \tfun{\tpair{\tses{\rho}}{\tses{\rho'}}}{\tunit}{\beta}
                              }
\\ \eresume{l}{}             \qquad&
                              \typeschema{\beta\rho\rho'}{%
                                \su{\popo{\rho\inp }{\rho'}}{\beta},~
                                \seql{\rho'}{l}
                              }{%
                                \tfun{\tses{\rho}}{\tses{\rho'}}{\beta}
                              }
\end{array}
\]
The effect of $\erequest{l}{c}$ ($\eaccept{l}{c}$) is to push a new static session endpoint $l$ on the stack.
For simplicity of typing, we assume functions $\erequest{l}{c}$ and $\eaccept{l}{c}$ for every global channel $c$.
The session type of
the endpoint is a variable $\psi$, to be substituted with a concrete session type (or a fresh variable in the case of inference)
at instantiation of the polymorphic type. This $\psi$ has to be equal to the session type of the static endpoint $c$ (resp.,
$\co{c}$), expressed by the constraint $\seql{c}{\psi}$ (resp., $\seql{\co c}{\psi}$). The return type of the function is
$\tses{\rho}$, where $\seql{\rho}{l}$. The types of the rest of the functions follow the same principles.
The following definition allows the instantiation of a type schema under a global constraint environment $C$.
\begin{definition}[Solvability]
  $\typeschema{\vec\gamma}{C_0}{T}$ is \emph{solvable} from $C$ using substitution $\sigma$ when
  $\dom{\sigma}\subseteq\{\vec\gamma\}$ and $C \vdash C_0\sigma$.
  \TS is solvable from $C$ if it exists $\sigma$ such that $\TS$ is solvable from $C$ using $\sigma$.
\end{definition}

\end{hide}

\begin{hide}
\paragraph{Typing Rules.}
The rules of our type-and-effect system are shown in \cref{fig:typing-rules}.
Most typing rules
are standard---we discuss only those different from \cite{ANN}.
\todo{I've hidden rule explanations and the def. of sub-typing - the def.
should go to the appendix}
\end{hide}
\begin{hide}
Rule \iref{tmatch} types a case expression with an external-choice behaviour of
the same number of branches. The choice labels $L_i$ ($i\in I$) in the code
determine those in the behaviour. The computation of each branch will return a
value of the same type $T$ but possibly have a different effect $b_i$.
Rule \iref{tsub} is for subtyping and sub-effecting. The latter can only
replace behaviours with variables, avoiding a more complex relation.
\begin{definition}[Functional Subtyping] 
  \label{def:func-subtyping}
  $\subType{T}{T'}$ is the least reflexive, transitive, compatible relation on types with the axioms:
  \begin{gather*}
    \irule*{(T_1 \subseteq T_2) \in C}{\subType{T_1}{T_2}}
    ~~~
    \irule*{C \vdash \seql{\rho}{\rho'} } {\subType{\tses \rho} {\tses {\rho'}} }
    ~~~
    \irule*{%
      \coType {T_1'}{T_1} \quad \sub{\beta}{\beta'} \quad \coType{T_2}{T_2'}
    }{%
      \coType{\tfun{T_1}{T_2}{\beta} }{\tfun {T_1'}{T_2'}{\beta'} }
    }
  \end{gather*}
\end{definition} 

\end{hide}


\section{Second Level Typing: Session Types} 
\label{sec:session-types}

\begin{hide}
The type-and-effect system we presented so far is parametric to session type annotations in behaviours $\pusho{l}{\eta}{}$ and
constraints ($\seql{c}{\eta}$, $\seql{\co c}{\eta}$).
\end{hide}
Session types describe the communication protocols of endpoints;
their syntax~is:
\\[1ex]$%
~\hfill~
  \small
\begin{array}{@{}r@{~}c@{~}l@{}}
  \eta & ::= &      \tend
               \bnf \outp T.\eta
               \bnf \inp T.\eta
               \bnf \outp \eta.\eta
               \bnf \inp \eta.\eta
               \bnf \sichoice {i}{I} {\eta}
               \bnf \sechoice {i}{I_1}{I_2}{\eta}
               \bnf \psi
\end{array}
~\hfill~
$\\[1ex]
A session type is finished ($\tend$) or it can describe further interactions:
the input ($\inp T.\eta$) or output ($\outp T.\eta$) of a \emph{confined} value $T$,
or the delegation ($\outp \eta'.\eta$) or resumption ($\inp \eta'.\eta$) of an
endpoint of session type $\eta'$, or the offering of  
non-deterministic selection ($\sichoicetext {i}{I} {\eta}$) of a label $L_i$,
signifying that session type $\eta_i$ is to be followed next.

Moreover, a session type can offer an external choice $\sechoicetext {i}{I_1}{I_2}{\eta}$ to its communication partner. Here $I_1$
contains the labels that the process \emph{must} be able to accept and $I_2$
the labels that it \emph{may} accept. We require that $I_1$ and $I_2$ are
disjoint and $I_1$ is not empty.
\begin{hide}
Hence, session types give a lower ($I_1$) and an upper ($I_1\cup I_2$)
bound of the labels in external choices.
These two sets of labels are not necessary for typing external choice---we
could use only the first set. However,
\end{hide}
\begin{icalp}
Although a single set would suffice,
\end{icalp}
the two sets make type
inference deterministic and independent of source code order.
\begin{hide}
It also makes typing more efficient, modular and intuitive.
\end{hide}
\begin{hide}
\begin{example} 
  Consider a program $P[e_1][e_2]$ containing the expressions:
  \[
    \begin{array}{@{}r@{~}c@{~}l@{}}
      e_1 &\defeq& \elet{x}{\eapp{\eaccept{l_1}{c}}{\cunit}}{\ematch{x}{L_1 \Rightarrow e,~L_2 \Rightarrow e^{\star}}{}}
      \\
      e_2 &\defeq& \elet{x}{\eapp{\erequest{l_2}{c}}{\cunit}}{\eapp{\eselectnew{L_2}}{x}}
    \end{array}
  \]
  Suppose $e^\star$ contains a type error, possibly because of a mismatch in session types with another part of $P$. If a type
  inference algorithm run on $P[e_1][e_2]$ first examines $e_1$, it will explore both branches of the choice, tentatively
  constructing the session type $\Sigma\{L_1.\eta_1,~L_2.\eta_2\}$, finding the error in $e^\star$. One strategy might then be to
  backtrack from typing $e^{\star}$ (and discard any information learned in the $L_2$ branch of this and possibly other choices in the code)
  and continue with the session type $\Sigma\{L_1.\eta_1\}$. However, once $e_2$ is encountered, the previous error in $e^\star$
  should be reported. A programmer, after successfully type checking $P[e_1][()]$, will be surprised to discover a type error in
  $e_1$ after adding in $e_2$. The type-and-effect system here avoids such situations by typing all choice branches, even if they
  are \emph{inactive}, at the expense of rejecting some---rather contrived---programs. A similar approach is followed in the
  type-and-effect system of the previous section by requiring \emph{all} branches to have the same type (Rule~\iref{tmatch} in
  \cref{fig:typing-rules}).
\end{example} 
\end{hide}

\begin{myfigure*}{Abstract Interpretation Semantics.}{fig:abstract-interpr}
\footnotesize

  $\begin{array}{@{}r@{\;}l@{~}l@{}}
\RefTirName{End}:\hfill
    \dstep {\St l \tend } {b
            &} {\Delta}{b}
            &
  \\[1ex]
\RefTirName{Beta}:\hfill
\dstep {\Delta}          {\beta
            &} {\Delta}{b}
            & \text{if } \sub{b}{\beta}
  \\[1ex]
\RefTirName{Plus}:\hfill
\dstep {\Delta}          {b_1 \oplus b_2
            &} {\Delta}{b_i}
            & \text{if } i \in \{1, 2\}
  \\[1ex]
\RefTirName{Push}:\hfill
\dstep {\Delta}           {\pusho{l}{\eta}{} &}
          {\St {l}{\eta}}    {\tau}
          & \text{if } l \not\in \Delta.\mathsf{labels}
  \\[1ex]
\RefTirName{Out}:\hfill
\dstep {\St {l}{\outp T.\eta}} {\popo{\rho}{\outp T'} &}
          {\St {l}{\eta}} {\tau}
          & \text{if } \nbox{C\vdash\seql{\rho}{l},~ 
          \subt{T'}{T}}
  \\[1ex]
\RefTirName{In}:\hfill
\dstep {\St {l}{\inp T.\eta}} {\popo{\rho}{\inp T'} &}
          {\St {l}{\eta}} {\tau}
          & \text{if }
          \nbox{C\vdash\seql{\rho}{l},~ 
          \subt{T}{T'}}
  \\[1ex]
\multicolumn{2}{@{}l@{}}{
  \RefTirName{Del}:
  \mcconf {\stBase {l}{\outp \eta_d.\eta}{\St {l_d}{\eta_d'}}}
          {\popo{\rho}{\outp \rho_d}}
}\\
&\xrightarrow{}_{C}
\mcconf {\St {l}{\eta}} {\tau}
          & \text{if }
          \nbox{%
          C\vdash\seql{\rho}{l},~ \seql{\rho_d}{l_d},~ \subt{\eta_d'}{\eta_d}}
  \\[1ex]
\RefTirName{Res}:\hfill
\dstep {(l:\inp \eta_r.\eta)}
          {\popo{\rho}{\inp l_r} &}
          {\stBase {l}{\eta}{(l_r:\eta_r)}} {\tau }
          & \text{if } (l\neq l_r),~ C\vdash\seql{\rho}{l}
  \\[1ex]
\RefTirName{ICh}:\hfill
\dstep {\St {l}{\sichoice{i}{I}{\eta}}}
            {\popo{\rho}{\outp L_j}
            & }
            {\St {l}{\eta_j} }    {\tau}
            & \text{if } (j \in I),~ C\vdash\seql{\rho}{l}
  \\[1ex]
\multicolumn{2}{@{}l@{}}{
  \RefTirName{ECh}:
  \mcconf{\St {l}{\sechoice{i}{I_1}{I_2}{\eta}}}{\bechoice{j\in J}{\rho}{L_j}{b_j}}
}\\
&\xrightarrow{}_{C}
\mcconf{\St {l}{\eta_k}}{b_k}
            &\text{if }\nbox{%
             k\in J,~ C\vdash\seql{\rho}{l},~
             \\I_1\subseteq J \subseteq I_1 \cup I_2
            }
  \\[1ex]
\RefTirName{Rec}:\hfill
\dstep {\Delta} {\orec{b}{\beta} &} {\Delta} {\tau}
            & \nbox{\text{if }
            \stackop {'} {b } {\epsilon } {},\\
            C' = (C\removeRHS \beta) {\cup} (\su{\tau}{\beta})
          }
  \\[1ex]
\RefTirName{Spn}:\hfill
\dstep {\Delta} {\espawn{b} &} {\Delta} {\tau}
          & \text{if } \stackop {} {b} {\epsilon}{}
  \\[1ex]
\RefTirName{Seq}:\hfill
\dstep {\Delta}          {b_1;b_2
            &} {\Delta'}{b_1';b_2}
            & \text{if } \dstep {\Delta} {b_1}{\Delta'}{b_1'}
  \\[1ex]
\RefTirName{Tau}:\hfill
  \dstep {\Delta}          {\bseq{\tau}{b}
  &} {\Delta}{b}
            &
  \end{array}$
\end{myfigure*}

We express our session typing discipline as an \emph{abstract interpretation
semantics} for behaviours shown
\begin{hide}
in \cref{fig:abstract-interpr}, which conservatively approximates the
communication effect of expressions at runtime.
\end{hide}
\begin{icalp}
in \cref{fig:abstract-interpr}.
\end{icalp}
It describes transitions of the form $\dstep {\Delta} {b} {\Delta'}{b'}$, where $b,b'$ are behaviours.
The $\Delta$ and $\Delta'$ are stacks on which static endpoint labels together with their corresponding session
types $(l:\eta)$ can be pushed and popped.
Inspired by Castagna et al. \cite{CastagnaEtal09}, in the transition $\dstep {\Delta} {b} {\Delta'}{b'}$, behaviour $b$ can
only use the top label in the stack to communicate, push another label on the stack, or pop the top label provided
its session type is $\tend$. This stack principle gives us a partial lock freedom property
(\cref{thm:type-soundness}).

Rule \RefTirName{End} from \cref{fig:abstract-interpr} simply removes a finished
stack frame, and rule \RefTirName{Beta} looks up behaviour variables in $C$; \RefTirName{Plus} chooses one of the branches of non-deterministic behaviour. The
\RefTirName{Push} rule extends the stack by adding one more frame to it, as long as the label has not been added before on the
stack (see \cref{ex:double-example}). Rules \RefTirName{Out} and \RefTirName{In} reduce the top-level session
type of the stack by an output and input, respectively. The requirement here is that the labels in the stack and the behaviour
match, the usual subtyping \cite{GayHole05} holds for the communicated types, and that the communicated types are \emph{confined}.
Note that sending confined (recursive) functions does not require delegation of
endpoints.

Transfer of endpoints is done by delegate and resume
(rules \RefTirName{Del} and \RefTirName{Res}). Delegate sends the second endpoint in the stack over the first; resume mimics
this by adding a new endpoint label in the second position in the stack. Resume requires a one-frame stack to guarantee
that the two endpoints of the same session do not end up in the same stack, thus avoiding deadlock \cite{CastagnaEtal09}. If we
abandon the partial lock freedom property guaranteed by our type system, then the conditions in \RefTirName{Res} can be relaxed
and allow more than one frame.

A behaviour reduces an internal choice session type by selecting one of its labels (\RefTirName{ICh}). A behaviour
offering an external choice is reduced non-deterministically to any of its branches (\RefTirName{ECh}). The behaviour must offer
all \emph{active} choices ($I_1\subseteq J$) and all behaviour branches must be typable by the session type
($J \subseteq I_1\cup I_2$).

As we previously explained, recursive functions in \lang must be confined. This means that the communication
effect of the function body is only on endpoints that the function opens internally, and the session type of these endpoints
is followed to completion (or delegated) before the function returns or recurs. This is enforced in Rule~\RefTirName{Rec}, where
$\orec{b}{\beta}$ must have no net effect on the stack, guaranteed by
$\stackop {'} {b } {\epsilon } {}$. Here $C' = (C\removeRHS \beta) {\cup} (\su{\tau}{\beta})$ is the original
$C$ with constraint $(\orec{b}{\beta}\subseteq\beta)$ replaced by $(\su{\tau}{\beta})$
(cf., Def.~\ref{def:wf-constraints}).
This update of $C$ prevents the infinite unfolding of $\orec{\beta}{b}{}$.
Spawned processes must also be confined (\RefTirName{Spn}).
We work with \emph{well-formed} constraints:

\begin{hide}
We will treat these stacks linearly, in the sense that a label can be pushed onto a stack only if it has not been previously
pushed on (and possibly popped from) that stack. Therefore, every stack
$\Delta$ contains an implicit set of the labels $\Delta.\mathsf{labels}$ that
have been pushed onto it.
\end{hide}

\begin{definition}[Well-Formed Constraints] \label{def:wf-constraints} 
  $C$ is \emph{well-formed} if:
  \begin{enumerate}
    \item \emph{Type-Consistent:} for all type constructors $tc_1$, $tc_2$, if $(\su{tc_1(\vec t_1)}{tc_2(\vec t_2)})\in C$, then
      $tc_1=tc_2$, and for all $t_{1i} \in \vec t_1$ and $t_{2i} \in \vec t_2$, $(\su{t_{1i}}{t_{2i}})\in C$.
    \item \emph{Region-Consistent:} if $C \vdash \seql{l}{l'}$ then $l=l'$.
    \item \emph{Behaviour-Compact:} behaviour constraints cycles contain
    a $(\su{\orec{b}{\beta}}{\beta})\in C$; also if
    $(\su{\orec{b}{\beta}}{\beta'})\in C$ then $\beta=\beta'$ and
    $\forall(\su{b'}{\beta})\in C$, $b'=\orec{b}{\beta}$.
    \item \emph{Well-Confined:} if $C \vdash \confined(T) $ then $ T \neq \tses \rho $; 
      also if $ C \vdash \confined(b) $ then $ b \not\in\{%
    \popo{\rho}{\outp T}, \popo{\rho}{\inp T}, \popo{\rho}{\outp \rho}, \popo{\rho}{\inp l},
    \popo{\rho}{\outp L_i},$ $ \bechoice{}{\rho}{L_i}{b_i}
    \} $.
  \end{enumerate}
\end{definition}

The first and fourth conditions are straightforward.
The third condition disallows recursive behaviours
through the environment without the use of a $\orec{b}{\beta}$ effect.
All well-typed \lang programs contain only such recursive behaviours because recursion is only possible through the use of a
recursive function.
The second part of the condition requires that there is at most one recursive constraint in the
environment using variable $\beta$. This is necessary for type preservation and decidability of session
typing.
The second condition of \cref{def:wf-constraints}
requires that only endpoints from a single source can flow in each $\rho$,
preventing aliasing of endpoints generated at different source locations.

\begin{myfigure*}{Examples of aliasing}{fig:aliasing}
  (a)~
\begin{minipage}[t]{.41\textwidth}\vspace{-1.5em}
\begin{lstlisting}[backgroundcolor=\color{white},numbers=none]
let val (p1$,$ p2) = ($\erequest{l_1}{c}, \erequest{l_2}{d}$)
    val p3 = if $ e $ then p1 else p2
in send p3 $\ctrue$
\end{lstlisting}
\end{minipage}
\hfill
  (b)~
  \begin{minipage}[t]{.40\textwidth}\vspace{-1.5em}
\begin{lstlisting}[backgroundcolor=\color{white},numbers=none]
let fun f  = $\erequest{l}{c}$
    val p1 = f ()
in send p1 1;
   let val p2 = f () in send p1 2;
\end{lstlisting}
\end{minipage}
\end{myfigure*}

\begin{example}[Aliasing of Different Sources]\label{ex:region-flow}
  Consider the program in \cref{fig:aliasing}~(a).
  Which endpoint flows to \keyw{p3} cannot be statically determined and therefore the program cannot yield a consistent session
  type for channels $c$ and $d$.
  The program will be rejected in our framework because \texttt{p3} has type $\tses{\rho}$
  and from the constrain environment $C \vdash \rho \sim l_1,~\rho\sim l_2$, which fails \cref{def:wf-constraints}.
\end{example}

Because endpoints generated from the same source
code location are identified in our system, 
stacks are treated \emph{linearly}: an endpoint label $l$ may
only once be pushed onto a stack. Every stack $\Delta$ contains an
implicit set of the labels $\Delta.\mathsf{labels}$ to record previously pushed labels.

\begin{example}[Aliasing From Same Source] \label{ex:double-example}
  Consider the program in \cref{fig:aliasing}~(b) where endpoint \keyw{p1} has type $\tses{\rho}$, with $C \vdash \rho\sim l$.
  The program has behaviour 
  $\pusho{l}{\eta}{}; \popo{\rho}{\outp \tint}; \pusho{l}{\eta}{}; \popo{\rho}{\outp \tint}; \tau$.
  Label $l$ is pushed on the stack twice and the behaviour complies with the session type $\eta=\outp \tint. \tend$.
  However the program does not respect this session type because it sends two integers on \keyw{p1} and none on \keyw{p2}.
  Our system rejects this program due to the violation of stack linearity.
\begin{hide}
Because region analysis identifies endpoints generated from the same source
code location, the above behaviour does not differentiate between the labels in
the second push and the final output.
Therefore, the only reason to reject this behaviour is because it pushes the
same location on the stack twice.
\end{hide}
\end{example}

Our system also rejects the correct version of the program in \cref{fig:aliasing}~(b), where the last send is replaced by
$\eapp{\esend{\keyw{p2}}}{\keyw{2}}$. This is because the label $l$ associated with the variable $\rho$ of a type $\tses{\rho}$ is
\emph{control flow insensitive}. Existing techniques can make labels control flow sensitive (e.g., \cite{TofteT94,kcfa}).

Using the semantics of \cref{fig:abstract-interpr} we define the following predicate which requires behaviours to follow to
completion or delegate all $(l:\eta)$ frames in a stack.
\begin{definition}[Strong normalization]
$\stackop{}{b}{\Delta}{\vec \Delta'}$ when
for all $b',\Delta'$ such that
$ \stackstepp{}{\Delta}{b}{\Delta'}{b'} \donearrow$ we have $ b' = \tau $
and $ \Delta' \in \{\vec\Delta'\} $.
We write
$\stackop{}{b}{\Delta}{}$
when
$\stackop{}{b}{\Delta}{\epsilon}$, where $\epsilon$ is the empty stack.
\end{definition}
\begin{hide}
 Because of finite session stacks and session types, behaviour-compact constraint environments, and
because of the rule for recursive behaviour explained below, there are no infinite sequences of reductions over this semantics.
Therefore $\stackop{}{b}{\Delta}{\vec \Delta'}$ is decidable.
\end{hide}
\begin{hide}
Note that according to the rules of \cref{fig:abstract-interpr}, in order for
$\stackop{}{b}{\Delta}{}$ to be true, session variables $\psi$ can only appear as delegated types not used in the reductions. For typable programs, our session type inference
algorithm can indeed infer session types that will satisfy this requirement
(\cref{sec:combine-types}).
\end{hide}

\begin{hide}
Similarly to recursive functions, a spawn will create a new, confined process.
This is guaranteed by Rule \iref{tspawn} in \cref{fig:typing-rules}, and Rule~
here which requires that the effect $b$ of the spawned process satisfies
$\stackop{}{b}{\Delta}{}$.
\end{hide}
\begin{hide}
\begin{remark}[Type annotations]
  The rule for recursive behaviour gives a method to safely bypass the linearity principle of the stacks, which requires that each
  label is pushed on the stack at most once, and type more programs. Because of this restriction, the $\mathit{swap}$ method of
  \ref{ex:swap1} can only be applied once in the body of the let.
  To explain this let us consider the type schema of the $\mathit{swap}$ function:
  $\forall(\alpha_1\alpha_2\beta:C_0)\tfun{\alpha_1}{\alpha_2}{\beta}$, where
  $$C = (\pusho{l}{\eta}{};\popo{\rho}{\outp \alpha_1};\popo{\rho}{\inp \alpha_2};\tau \subseteq \beta), (\rho \sim l),
  (\alpha_1\subseteq\alpha_2) $$
  The expression
  $(\eapp{\mathit{swap}}{\cunit}; \eapp{\mathit{swap}}{\cunit})$
  which applies $\mathit{swap}$ twice at type $\tfun{\tunit}{\tunit}{\beta}$ will have a behaviour (omitting some $\tau$s):
  \[\nbox{%
    \pusho{l}{\eta}{};\popo{\rho}{\outp \tunit};\popo{\rho}{\inp \tunit};\tau;~
    \pusho{l}{\eta}{};\popo{\rho}{\outp \tunit};\popo{\rho}{\inp \tunit};\tau
  }\]
  This behaviour clearly violates stack linearity because it pushes $l$ twice on the stack, and therefore the expression is not
  typable in our system.

  However, the communication effect of $\mathit{swap}$ is confined, and therefore we can define the function using the $\krec$
  construct instead of $\kfun$. In this case, the type schema of $\mathit{swap}$ will contain the constraint set
  $$C = (\orec{(\pusho{l}{\eta}{};\popo{\rho}{\outp \alpha_1};\popo{\rho}{\inp \alpha_2};\tau)}{\beta} \subseteq \beta), (\rho \sim l) $$
  and the above expression will have behaviour:
  \[
    b= \nbox{%
      \orec{(\pusho{l}{\eta}{};\popo{\rho}{\outp \tunit};\popo{\rho}{\inp \tunit};\tau)}{\beta};~
      \orec{(\pusho{l}{\eta}{};\popo{\rho}{\outp \tunit};\popo{\rho}{\inp \tunit};\tau)}{\beta}
  }
  \]
  It is easy to verify that this behaviour satisfies $\stackop{}{b}{\epsilon}{}$.

  Therefore the $\krec$ construct can be used by programmers as an annotation to mark the functions that are confined and
  therefore it is safe to apply multiple times in programs.
\end{remark}
\end{hide}
\begin{hide}
\paragraph{Endpoint Duality}
So far our type discipline does not check session types for duality. The
program: \[\nbox{%
    \espawn{(\efunc{\_}{\elet{x}{\eapp{\erequest{}{c}}{\cunit}}{\eapp{\esend{x}}{\keyw{\ctrue}}}})};~
  \elet{x}{\eapp{\eaccept{}{c}}{\cunit}}{\erecv{}{x} + \keyw{1}}
}\]
should not be typable because its processes use dual endpoints at incompatible
session types.
Therefore
\end{hide}
\begin{icalp}
Lastly,
\end{icalp}
session types on dual session endpoints ($\seql{c}{\eta}$,
$\seql{\co c}{\eta'}$) must be dual ($C \vdash \eta \dual{} \eta')$
\begin{icalp}
The definition of duality is standard, with the exception that
internal choice is dual to external choice only if
the labels in the former are included in the \emph{active} labels in the latter.
\end{icalp}

\begin{definition}[Valid Constraint Environment] \label{def:valid-type-solution}
  $C$ is \emph{valid} if there exists a substitution $\sigma$ of variables $\psi$
  with closed session types, such that  $C\sigma$ is well-formed and
  for all $(\seql{c}{\eta}), (\seql{\co c}{\eta'})\in C\sigma$ we have
  $C\vdash \eta \dual{} \eta'$.
\end{definition}

\begin{hide}
\begin{definition}[Duality] \label{def:duality}
$ C\vdash \eta \dual{} \eta'$ if the following rules and their symmetric ones are satisfied.
\todo{moved all rules in a single line}
\begin{gather*}
  \irule*{}{%
    C \vdash  \tend \dual{} \tend
  }
  \quad
  \irule*{%
    C \vdash  \subt{T}{T'}
    \\\\
    C \vdash  \eta \dual{} \eta'
  }{%
    C \vdash  {\outp T}.\eta \dual{} {\inp T'}.\eta'
  }
  \quad
  \irule*{%
    C \vdash  \subt{\eta_0}{\eta_0'}
    \\\\
    C \vdash  \eta \dual{} \eta'
  }{%
    C \vdash  {\outp \eta_0}.\eta \dual{} {\inp \eta_0'}.\eta'
  }
  \quad
  \irule*{%
    \forall i\in I_0.~~ C \vdash  \eta_i \dual{} \eta'_i
  }{%
    C \vdash  {\sichoiceBase{i}{I_0}{\eta}{i}} \dual{} {\sechoiceBase{i}{I_0I_1}{I_2}{\eta'}{i}})
  }
\end{gather*}
where $C \vdash  \subt{\eta}{\eta'}$ is Gay\&Hole \cite{GayHole05} subtyping, with $C$ needed for
inner uses of $C \vdash  \subt{T}{T'}$, extended to our form of external choice, where
$C \vdash\sechoicetext{i}{I_1}{I_2}{\eta'}
\subt{}
\sechoicetext{i}{J_1}{J_2}{\eta'}$
when
$I_1\subseteq J_1$
and
$J_1\cup J_2\subseteq I_1\cup I_2$
and
$\forall( i \in J_1\cup J_2).~
C \vdash  \subt{\eta_i}{\eta_i'}$.

\end{definition}
\end{hide}


\subsection*{Combining the Two Levels} \label{sec:combine-types}
\begin{hide}
We
\todo{Important: we need to state the level 1 result that reductions in
the op. semantics imply reductions in the abstract interpretation semantics.}
are interested in typing source-level programs which contain a single process
$e$ with no open endpoints. However, to prove type soundness, we need to type
running systems containing multiple running processes and open endpoints. For
each running process we need to maintain a stack $\Delta$, containing the
session types of the endpoints opened by the process.
  We write
  $\tjSysrVec{}{e}{b}{\Delta}$ if $C$ is well-formed and valid, $(\tjrBaseVec{}{}{e}{T}{b}{C}{\emptyset})$, and
  $(\stackopVec{}{b}{\Delta}{})$, for some type $T$.
  Program $e$ is well-typed if $\tjSysr{}{e}{b}{\epsilon}$ for some $C$ and $b$.
\end{hide}

The key property here is \emph{well-stackedness}, the fact that in a running system where each process has a corresponding stack
of endpoints,
there is a way to repeatedly remove pairs of endpoints with dual session types from the top of two stacks, until all stacks are
empty.
\begin{definition}[Well-stackedness]                \label{def:well-stackedness}
$C \wellstacked \c S$ is the least relation satisfying:
\begin{gather*}
   \irule[][]
     {~}
     {C \wellstacked \epsilon  }
~
   \irule[][] {%
      C \wellstacked  \c S,~ \tconf{\Delta}{b}{e},~ \tconf{\Delta'}{b'}{e'}
      \\
      C \vdash \eta\dual {} \eta'
      \\
        p,\co p\freshfrom \Delta, \Delta', \c S
    }{%
      C \wellstacked \c S,
      \tconf{(p^l:\eta) \cdot \Delta}{b}{e},
      \tconf{(\co p^{l'}:\eta') \cdot \Delta'}{b'}{e'}
    }
\end{gather*}
\end{definition}

Note that this does not mean that programs are deterministic. Multiple pairs of endpoints may be at the top of a set
of stacks.
\begin{icalp}
Duality of endpoints guarantees that communications are safe;
the ordering of endpoints in removable pairs implies the absence of deadlocks.
\end{icalp}

We let $P$, $Q$ range over tuples of the form $\tconf{\Delta}{b}{e}$ and
$\c S$ over sequences of such tuples.
In this section stack frames $ (p^l : \eta ) $ store both
endpoints and their labels.
\begin{hide}
(and trivially lift $\Downarrow_C$ to such stacks).
  Program $e$ is well-typed if $\tjSysr{}{e}{b}{\epsilon}$ for some $C$ and $b$.
\end{hide}
\begin{icalp}
We write
  $\tjSysrVec{}{e}{b}{\Delta}$ if $C$ is well-formed and valid, $(\tjrBaseVec{}{}{e}{T}{b}{C}{\emptyset})$, and
  $(\stackopVec{}{b}{\Delta}{})$, for some $\vect T$.
  We write $C \wellstacked \c S$ if $\vect{\Delta}$ is well-stacked.
\end{icalp}
\begin{hide}
\begin{definition}[Well-stackedness]                \label{def:well-stackedness}
$C \wellstacked \c S$ is the least relation satisfying the rules:
\begin{gather*}
   \irule[][]
     {}
     {C \wellstacked \epsilon  }
~
   \irule[][] {%
      C \wellstacked  \c S,~ \tconf{\Delta}{b}{e},~ \tconf{\Delta'}{b'}{e'}
      \\
      C \vdash \eta\dual {} \eta'
      \\
        p,\co p\freshfrom \Delta, \Delta', \c S
    }{%
      C \wellstacked \c S,
      \tconf{(p^l:\eta) \cdot \Delta}{b}{e},
      \tconf{(\co p^{l'}:\eta') \cdot \Delta'}{b'}{e'}
    }
\end{gather*}
\end{definition}
\end{hide}
\begin{icalp}
Well-typed systems enjoy session
fidelity and preserve typing and well-stackedness.

 \begin{theorem} \label{thm:preservation}
Let  $\c S = \vect{\mcconf{\Delta}{b}, e}$
and
$\tjParSys{}{\c S}$
and
$\wst{\c S}$
and
$\vect{e} \arrow \vect{e}'$; then
there exist $\vect{\Delta}',\vect{b}'$ such that
$\c S'=\vect{\tconf{\Delta'}{b'}{e'}}$ and:
\\
\begin{tabular}{@{\quad}l@{\quad}l}
  1. $\tjParSys{}{\c S'}$
  & (Type Preservation)
\\2. $\vect{\mcconf \Delta b} \rightarrow^*_C \vect {\mcconf{\Delta'}{b'}}$
  & (Session Fidelity)
\\3. $ \wst{\c S'} $
  & (Well-Stackedness Preservation)
\end{tabular}
\end{theorem}

Session fidelity and well-stackedness preservation imply \emph{communication safety},
since the former guarantees that processes are faithful to session types
in the stacks, while the latter that session types
are dual for each pair of open endpoints $ p $ and $ \bar p$.
Moreover, well-stackedness implies \emph{deadlock freedom}.
$P$ \emph{depends on} $Q$ if the endpoint at the top of $P$'s stack has dual endpoint in $Q$.

\begin{lemma}[Deadlock Freedom]
  $\wst{\c S}$; dependencies in $\c S$ are acyclic.
\end{lemma}
\end{icalp}

\begin{hide}
\begin{theorem}[Type Preservation] \label{thm:preservation} 
Let  $\tjParSys{}{\c S}$ and $\wst{\c S}$.
If $S \arrow \vect{e}'$ then there exist $\vect{\Delta}',\vect{b}'$ such that
$\c S'=\vect{\tconf{\Delta'}{b'}{e'}}$ and $\tjParSys{}{\c S'}$ and $ \wst{\c S'}$.
\end{theorem}
\end{hide}

Type soundness is more technical.
We divide system transitions to communication transitions between processes ($\comStep$)
and internal transitions ($\intStep$).
Let
$\c S \comStep \c S'$ ($\c S \intStep \c S'$) when
$ S \arrow S'$, derived by Rule
\iref{init}, \iref{comm}, \iref{deleg} or \iref{select} of \cref{fig:lang}
(resp., any other rule); $S \comStepW S' $ when $ S \intStep^*\comStep\intStep^* S'$.
\begin{hide}

We also define dependencies between processes of a running system according to the following definition.
\begin{definition}[Dependencies]
Let $ P = \tconf{\Delta}{b}{e}$ and $ Q = \tconf{\Delta'}{b'}{e'} $ be processes in $ \c S $.
\begin{description}
  \item[$ P $ and $ Q $ are \emph{ready} ($ P \leftrightharpoons Q$):] if
    $\Delta = (p^l:\eta)\cdot\Delta_0$ and
    $\Delta' = (\co p^{l'}:\eta') \cdot \Delta_0'$;

  \item[$ P $ is \emph{waiting} on $ Q $ ($P \mapsto Q$):] if
    $ \Delta = (p^l:\eta) \cdot \Delta_0$ and
    $ \Delta' = \Delta_1' \cdot (\co p^{l'}:\eta') \cdot \Delta_0' $
    and $\Delta_1' \neq \epsilon $;

  \item[$ P $ \emph{depends} on $ Q, R $ ($P \Mapsto (Q,R)$):] if
    $P = Q \leftrightharpoons R $, or
    $P \mapsto^+ Q \leftrightharpoons R $.

\end{description}
\end{definition}
\end{hide}
%
\begin{theorem}[Type Soundness] \label{thm:type-soundness}
 Let $ \tjParSys{}{\c S}$ and $ \wst{\c S}$. Then
 \begin{enumerate}
   \item $ \c S \comStepW  \c S'$, or
   \item $ \c S \intStep^* (\c F,\c D,\c W,\c B) $ such that:
\begin{hide}
 \begin{description}
   \item[Processes in $\c F$ are finished:] $ \forall\tconf{\Delta}{b}{e} \in \c F.~\Delta=\epsilon, b=\tau$ and $e=v$.
   \item[Processes in $\c D$ diverge:] $ \forall\tconf{\Delta}{b}{e} \in \c D.~\tconf{\Delta}{b}{e} \intStep^\infty$.
   \item[Processes in $\c W$ wait on channels:] $ \forall\tconf{\Delta}{b}{e} \in \c W.~e = E[\erequest{l}{c}] $ or
     $ e = E[\eaccept{l}{c}]$.
   \item[Processes in $\c B$ block on sessions:]
     $\forall P=\tconf{\Delta}{b}{e} \in \c B.~e = E[e_0]$ and
     $e_0$ is
     $\esend{v}{}$, $\eapp{\erecv{}{}}{v}$, $\edeleg{v}$, $\eapp{\eresume{}{}}{v}$,
   $\eapp{\eselectnew{L}}{v}$, or $\ematch{v}{L_i \Rightarrow e_i}{i \in I}$
     and
     $\exists Q \in (\c D, \c W).~\exists R \in (\c D, \c W, \c B).~\c S\vdash P \Mapsto (Q, R)$.
 \end{description}
\end{hide}
\begin{icalp}
  \begin{description}
    \item[Finished processes, $\c F$:] $
  \forall P \in \c F.~P = \tconf \epsilon \tau v $, for some $ v $;
   \item[Diverging processes, $\c D$:] $ \forall P \in \c D.~P \intStep^\infty$;
   \item[Waiting proc., $\c W$:] $
   \forall P \in \c W.~P = \tconf{\Delta}{b}{E[e]} $ and $ e \in \{%
   \erequest{l}{c}, \eaccept{l}{c} \}$;
   \item[Blocked processes, $\c B$:]
     $\forall P \in \c B.~P = \tconf{\Delta}{b}{E[e]} $ and
     $e \in \{%
     \esend{v}{}, \eapp{\erecv{}{}}{v},$ $ \edeleg{v}, \eapp{\eresume{}{}}{v},
   \eapp{\eselectnew{L}}{v}, \ematch{v}{L_i \Rightarrow e_i}{i \in I}\}$
   and $P$ transitively depends on a process in $\c D\cup\c W$.
  \end{description}
\end{icalp}
 \end{enumerate}
 \end{theorem}
 A well-typed and well-stacked \lang system will either be able to perform a communication, or, after finite internal steps, it
 will reach a state where some processes are values ($\c F$), some
 internally diverge ($\c D$), some are waiting for a partner process to open a session ($\c W$), and some are blocked on a session
 communication ($\c B$).
 Crucially, in states where communication is not possible,  $\c B$ transitively depends on $\c D \cup \c W$. Thus,
 in the absence of divergence and in the presence of enough 
 processes to start new sessions, no processes can be blocked; the system will either perform a communication or
 it will terminate \emph{(partial lock freedom)}.
 \begin{corollary}[Partial Lock Freedom]
 If $ \tjParSys{}{\c S}$,~$ \wst{\c S}$,
 and $\c S \not \Longrightarrow_{\mathsf c} $
 and $ \c S \intStep^* (\c F, \emptyset,\emptyset,\c B)$ then $\c B=\emptyset$.
 \end{corollary}







  \section{Inference Algorithm}
  \label{sec:inference}
  We use three inference algorithms, $\algW$, $ \algoSI $ and $ \algoDual $.
The first infers functional types and communication effects and corresponds to the first level of our type system. The other
two infer session types from the abstract interpretation rules of \cref{fig:abstract-interpr} ($\algoSI$) and the duality
requirement of \cref{def:valid-type-solution} ($\algoDual$), corresponding to
the second level of the type system.

Algorithm $ \algW $ is a straightforward adaptation of the homonymous
algorithm from \cite{ANN}: given an expression $ e $, $ \algW $ calculates its
type $ t $, behaviour $ b $ and constraints set $ C $; no session
information is calculated.  $\algW$ generates
pairs of fresh constraints $ \seql c \psi $ and $ \seql{\co c}{ \psi'} $ for
each global channel $ c $ in the source program; $\psi$ and $\psi'$ are
unique. Results of $\algW$'s soundness and completeness follow
from \cite{ANN}.

For all constraints $ (\seql c \psi) \in C $, Algorithm $ \algoSI $
infers a substitution $ \sigma $ and a refined set $ C' $ such that
$\stackop{'}{b\sigma}{\epsilon}{\epsilon}$.
The substitution only maps $\psi$ variables to session types.
The final $C'$ is derived from $C$ by applying $\sigma$ and
possibly adding more type constraints of the form $(T \subseteq T')$.
The core of this algorithm is the abstract
interpreter $ \algoMC $, which explores all possible
transitions from $\mcconf \epsilon b $.

Algorithm $ \algoMC $ is designed in a continuation-passing style, using
a \emph{continuation stack} $ K ::= \emptyK \bnf \consK{b}{K} $.

As transition paths are explored, previously discovered branches of internal and external choices in session types may need to be expanded. For example, 
if Algorithm $ \algoMC $ encounters a configuration $ \mcconf {(l :
\sichoicetext {i}{I} {\eta})}{l!L_j} $ where $ j \not\in I $,
the inference algorithm needs to add the newly discovered label $ L_j $ to 
the internal choice on the stack. 

To do this, internal and external choices are removed from the syntax of
sessions, and replaced with special variables $ \psiint $ and $ \psiext $.
These variables are bound by unique \emph{choice constraints}, 
extending the syntax of constraints (\cref{fig:type-syntax}):
\begin{align*}
  C ::= &  
  \ldots \bnf 
  \ceq{\sichoicetext{i}{I}{\eta}}{\psiint}
          \bnf \ceq{\sechoicetext{i}{I_1}{I_2}{\eta}}{\psiext}
\end{align*}
$\algoMC$ updates $\psiint$ and $\psiext$ constraints in $C$ with newly discovered branches.
For example it may add new labels to an internal choice, or move
active labels to inactive in an external choice.

We now give more detail for some inference steps of Algorithm $\algoMC$.
The full algorithm can be found in an online technical report\footnote{%
{Spaccasassi}, C., {Koutavas}, V.: {Type-Based Analysis for Session Inference}.
  ArXiv e-prints  (Oct 2015), \url{http://arxiv.org/abs/1510.03929v3}.
}.
Algorithm $ \algoMC $ terminates successfully when all sessions on the stack have terminated,
the input behaviour is $ \tau $ and the continuation stack is empty:
\begin{lstlisting}[
    backgroundcolor=\color{white}, 
    escapeinside={(*}{*)},
    firstnumber=6,
    firstline=8,
    lastline=12,
    numbers=none]
$ 
  \mcclause \Delta \tau C \emptyK    
  = (\sigma,~ C\sigma) 
$ (* \label{code:MC-delta-tau} *)
  if $  \sigma = $ finalize $ \Delta $
\end{lstlisting}
When this clause succeeds, $ \Delta $ may be empty or it may contain frames of the form $(l:\psi)$ or $(l:\tend)$.
The helper function \textsf{finalise} $ \Delta $ returns a
substitution $ \sigma $ that maps all such $\psi$'s to $\tend$. If this is not possible (i.e., a session on $\Delta$ is not
finished) \textsf{finalise} raises an error.

New frames are pushed on the stack when the behaviour is $
\pusho{l}{\eta}{}$:
\begin{lstlisting}[
    backgroundcolor=\color{white}, 
    escapeinside={(*}{*)},
    firstnumber=11,
    firstline=21,
    lastline=28,
    numbers=none]
$ 
  \mcclause \Delta {\pusho{l}{\eta}{}} C K 
  = (\sigma_2\sigma_1, C_2)
$ (* \label{code:MC-push} *)
  if  ($\sigma_1, \Delta_1$) = checkFresh$ (l, \Delta) $
  and ($\sigma_2, C_2$) = $
    \mcclause {\stBase {l}{\eta\sigma_1}{\Delta_1}} \tau {C\sigma_1}{K\sigma_1}
  $
\end{lstlisting}
where \textsf{checkFresh} 
checks that $ l $ has never been in $ \Delta $. 

When the behaviour is an operation that pops a session from the stack, such
as a send ($ l!T $), $ \algoMC $ looks up the top frame on the
stack, according to the stack principle. There are two cases to
consider: either the top frame contains a fresh variable $\psi$, or some type
has been already inferred. The algorithm here is:
\begin{minipage}{\textwidth}
\begin{lstlisting}[
    backgroundcolor=\color{white}, 
    escapeinside={(*}{*)},
    firstnumber=16,
    firstline=31,
    lastline=42,
    numbers=none]
$
  \mcclause { \St l \psi } {\popo{\rho}{!T}} C K
    = 
  (\sigma_2\sigma_1, C_2)
$  (*\label{code:MC-send-psi}*)
  if  $ \subEql{l}{\rho} $
  and $ \sigma_1 = \subst\psi {!\alpha.\psi'} $ where $ \alpha, \psi' $ fresh
  and $ (\sigma_2, C_2) 
          = 
        \mcclause { \St l {\psi'}\sigma_1 } \tau {
          C\sigma_1 \cup \singleSu{T}{\alpha} } { K\sigma_1 } 
      $
\end{lstlisting}
\end{minipage}
\begin{minipage}{\textwidth}
\begin{lstlisting}[
    backgroundcolor=\color{white}, 
    escapeinside={(*}{*)},
    firstnumber=16,
    firstline=44,
    lastline=48,
    numbers=none]
$ 
  \mcclause {\St l {!\alpha.\eta}}{\popo{\rho}{!T}} C K
    = \mcclause { \St l \eta } \tau {C \cup \singleSu{T}{\alpha}} K
$ (*\label{code:MC-send-eta}*)
  if  $ \subEql{l}{\rho} $ 
\end{lstlisting}
\end{minipage}
In the first case, $ \algoMC $ checks that $\rho$ in the behaviour corresponds to $l$ at the top of the stack.
It then produces the substitution
$ \subst{\psi}{!\alpha. \psi'} $, where $ \alpha $ and $ \psi'$ are fresh, and 
adds $(T \subseteq \alpha)$ to $ C $. The second case produces no substitution.

The clauses for delegation are similar:
\begin{lstlisting}[
    backgroundcolor=\color{white}, 
    escapeinside={(*}{*)},
    firstline=68,
    lastline=87,
    numbers=none]
$
 \mcclause {\stBase{l}{\psi}{\St{l_d}{\eta_d}}} {\popo{\rho}{!\rho_d}} C K
= (\sigma_2\sigma_1, C_2)
$(*\label{code:MC-del-psi}*)
  if $ \subEql{l}{\rho} $ and $ \subEql{l_d}{\rho_d} $
  and $ \sigma_1 = \subst \psi {!\eta_d.\psi'} $ where $ \psi' $ fresh
  and $ (\sigma_2, C_2) = \mcclause {\St l {\psi'} \sigma} \tau {C\sigma_1
                            } {K\sigma_1} 
  $

$
  \mcclause {\stBase l {!\eta_d.\eta} {\St{l_d}{\eta_d'}}
  } {\popo{\rho}{!\rho_d}} C K
  = (\sigma_2\sigma_1, C_2)
$(*\label{code:MC-del-eta}*)
  if  $ \subEql{l}{\rho} $ and $ \subEql{l_d}{\rho_d} $
  and $ (\sigma_1, C_1) = \nsub(\eta_d', \eta_d, C) $
  and $ (\sigma_2, C_2) = 
    \mcclause {\St l {\eta} \sigma} \tau {C_1} {K\sigma_1} 
$
\end{lstlisting}
The main difference here is that, in the second clause, the \textsf{sub} function checks that 
$C \vdash  \subt{\eta_d'}{\eta_d}$ and performs relevant inference.
Moreover, the input $ \Delta $ must contain at least two
frames (the frame below the top one is delegated).

The cases for receive, label selection and offer, and resume are similar (see online report). In the cases for label selection and
offering, the algorithm updates the $ \psiint $ and $ \psiext $ variables, as discussed above. In the case of resume, the
algorithm checks that the stack contains one frame.

In behaviour sequencing and branching,
substitutions are applied eagerly and composed iteratively, and new constraints
are accumulated in $ C $:
\begin{lstlisting}[
    backgroundcolor=\color{white}, 
    escapeinside={(*}{*)},
    firstnumber=16,
    firstline=206,
    lastline=209,
    numbers=none]
$ \mcclause \Delta {\bseq {b_1} {b_2}} C K
  = 
  \mcclause \Delta {b_1} C {\consK {b_2} K}
  $ (*\label{code:MC-seq}*) 
\end{lstlisting}
\begin{lstlisting}[
    backgroundcolor=\color{white}, 
    escapeinside={(*}{*)},
    firstnumber=16,
    firstline=212,
    lastline=220,
    numbers=none]
$ \mcclause  \Delta {b_1\oplus b_2} C K
    = (\sigma_2\sigma_1, C_2) 
    $ (*\label{code:MC-plus}*)
  if  $ (\sigma_1, C_1) = \mcclause \Delta {b_1} C K $
  and $ (\sigma_2, C_2) = \mcclause {\Delta\sigma_1 
                                  } {b_2\sigma_1
                                  } {C_1
                                  } {K\sigma_1}
      $
\end{lstlisting}

When a recursive behaviour $ \orec{\beta}{b}{}$ is encountered,
Algorithm $ \algoMC $ needs to properly setup the input constraints $ C $
according to Rule \RefTirName{Rec} of \cref{fig:abstract-interpr}:
\begin{lstlisting}[
    backgroundcolor=\color{white}, 
    escapeinside={(*}{*)},
    firstnumber=16,
    firstline=229,
    lastline=237,
    numbers=none]
$ \algoMC(\mcconf{\Delta}{\orec{b}{\beta}},~ C)
   =  (\sigma_2\sigma_1, C_2) $ (*\label{code:MC-rec}*)
  if  $ C = C' \uplus \singleSu{b'}{\beta} $
  and $ (\sigma_1, C_1) = \mcclause \epsilon b {C' \cup
                                      \singleSu{\tau}{\beta}} \epsilon $ 
  and $ (\sigma_2, C_2) = \mcclause {\Delta\sigma_1
                                  } \tau {
 (C_1 \backslash \singleSu{\tau}{\beta}) \cup
 (\singleSu{b'}{\beta})\sigma_1 } { K\sigma_1 } $
\end{lstlisting}
Here the algorithm first calls $\algoMC$ on $\mcconf{\epsilon}{b}$,
checking that the recursion body $ b $ is self-contained under
$ C' $, in which the recursion variable $\beta$ is bound to $\tau$.
This update of $C$ prevents the infinite unfolding of $\orec{\beta}{b}{}$.
It then restores back the constraint on $
\beta $, applies the substitution $ \sigma_1$, and continues inference.

The clause for $\espawn{b}$ is similar, except that $ C $ is unchanged.
Variables $ \beta $ are treated as the internal choice 
of all behaviours $ b_i$ bound to $ \beta $ in $ C $:
\begin{lstlisting}[
    backgroundcolor=\color{white}, 
    escapeinside={(*}{*)},
    firstnumber=16,
    firstline=240,
    lastline=243,
    numbers=none]
$ \mcclause \Delta \beta C K
  = \mcclause \Delta b C K $(*\label{code:MC-beta}*)
  where  $ b = \bigoplus \{b_i \where \exists i. ~(b_i \subseteq \beta) \in C\}
  $ 
\end{lstlisting}
    
Inference fails when $ \algoMC $ reaches a stuck configuration
$\mcconf \Delta b$ other than $\mcconf \epsilon \tau$,
corresponding to an error in the session type discipline.
\begin{hide}
\todo{TODO: describe MC with rec\inp }
\end{hide}

To prove termination of $\algoSI$,
we first define the translation $\groundify{b}$, that
replaces $ \beta$ variables in $ b $ with the internal choice $
\bigoplus\set{b_i}{(\singleSu{b_i}{\beta}) \in C}$.
Due to behaviour-compactness (\cref{def:wf-constraints}), $\groundify b$ is a finite \emph{ground} term,
i.e.\ a finite term without $ \beta $ variables.
Except for Rule \RefTirName{Beta}, transitions in
\cref{fig:abstract-interpr} never expand $ b $; they either consume $
\Delta $ or $ b $.
Since $ \groundify{b} $ is finite when $ C$ is well-formed, $
\mcconf \epsilon {\groundify{b}} $ generates a finite state space and Algorithm
$ \algoMC $ always terminates.

\begin{icalp}
Similar to ML type inference,
\begin{hide}
  \todo{Cite Kfoury, J.; Tiuryn; Urzyczyn, P. (1990).
  "ML typability is dexptime-complete".}
\end{hide}
the worst-case complexity of $ \algoMC $ is exponential to program size:
$\algoMC$ runs in time linear to the size of $\groundify{b}$, which
in the worst case is exponentially larger than $b$, which is linear to program size.
The worst case appears in pathological programs where, e.g., each function calls all previously defined functions.
We intend to explore whether this is an issue in practice, especially with an optimised
dynamic programming implementation of $\algoMC$.
\end{icalp}

\begin{hide}
Thanks to $ \beta$-compactness, $\groundify b$ is a finite \emph{ground} term,
i.e. a finite term without $ \beta $ variables; moreover, the ground
translation $\groundify -$ is fully abstract w.r.t. the abstract interpretation
semantics.

Let $\execStates \Delta b C $ denote the set of all states reachable from
configuration $ \mc \Delta b $ under $C$, and $ \execSize{\Delta}{b}{C}$ denote
its cardinality. Let $ \behavDomain$ be the set of all behaviours $ b $.

Except for Rule \RefTirName{Beta}, transitions in the abstract
interpretation semantics never expand the behaviour $ b $, but rather they
consume it.
Therefore, when $ b$ is a finite ground behaviour, it is
possible to provide a finite upper bound $\bsize - ::
\behavDomain \rightarrow \mathcal N $ to the size of $\execStates \Delta b C $:
\begin{theorem}[Finite state-space]
Let $ \mc \Delta b $ be a configuration such that $ b $ is a finite ground
behaviour. For any well-formed $ C $, $ \execSize{\Delta}{b}{C} \leq \bsize{\mc
\Delta b} $.
\end{theorem}
\end{hide}

\begin{hide}
As a corollary of this theorem and the definition of $ \bsize{-}$,
we can derive a useful induction principle, whereby
the set of states reachable from $ \mc \Delta b$ is strictly decreasing after
any transition:
\begin{corollary}
Let $ b $ be a finite term, and $ C $ well-formed.
If $ \dstep{\Delta}{b}{\Delta'}{b'}$, then $ \execStates{\Delta'}{b'}{C} \subset
\execStates{\Delta}{b}{C} $.
\end{corollary}
\end{hide}

Soundness and completeness of $\algoSI$ follow from the these properties of $\algoMC$.
\begin{icalp}
\begin{lemma}[Soundness of $ \algoMC $]\label{prop:sound-mc}
Let $ C $ be well-formed
and $ \algoMC(\mcconf{\Delta}{b}, C) = (\sigma_1, C_1) $;
then $ \Delta\sigma_1 = \Delta' $ and $\stackopBase
{\Delta'}{b\sigma_1}{C_1}{}$.
\end{lemma}
\end{icalp}
\begin{hide}
\begin{lemma}[Soundness of $ \algoMC $]\label{prop:sound-mc}
Let $ C $ be well-formed in $ \mcconf{\Delta}{b} $.
If $ \algoMC(\mcconf{\Delta}{b}, C, K) = (\sigma_1, C_1) $,
then $ \Delta\sigma_1 \equiv \Delta' $ and $\stackopBase
{\Delta'}{K[b]\sigma_1}{C_1}{}$.
\end{lemma}
\end{hide}
\begin{hide}
The proof is by induction on the execution size of $ \mc \Delta b$,
where the equivalence relation $ \Delta \equiv \Delta'$ removes closed sessions
$(l: \tend)$ from the stack $ \Delta $, in order to account for $ \algoMC$'s
lazy inference of finished sessions.
\end{hide}

\begin{hide}
Completeness of Algorithm $ \algoSI $ also depends on the completeness of $\algoMC$,
and it is also proved by induction on the execution size:
\end{hide}
\begin{icalp}
\begin{lemma}[Completeness of $\algoMC $]
Let 
$ C $
be well-formed
and
$ \stackopBase { (\Delta} { b)\sigma } { C } { } $;
then
$ \algoMC(\mcconf{\Delta}{b}, C_0) = (\sigma_1, C_1) $
and
$\exists \sigma' $ such that
$ \cj {C} {C_1\sigma'} $
and
$ \forall \psi \in dom(\sigma)$,
$\coType{\sigma(\psi)}{\sigma'(\sigma_1(\psi))}$.
\end{lemma}
\end{icalp}
\begin{hide}
\begin{lemma}[Completeness of $\algoMC $]
Let $ \mc \Delta {K[b]} $ be well-formed in $ C $ and in $ C_0\sigma $. If
$ \stackopBase
{ (\Delta}
{ K[b])\sigma }
{ C }
{ } $
and $ \cj { C } { C_0\sigma }$, then $\exists \sigma' $ such that
$ \algoMC(\mcconf{\Delta}{b}, C_0)
= (\sigma_1, C_1) $ terminates
and
$ \cj {C} {C_1\sigma'} $
and
$ \forall \psi \in dom(\sigma).~\coType{\sigma(\psi)}{\sigma'(\sigma_1(\psi))}
$.
\end{lemma}
\end{hide}
Completeness states that $ \algoMC $ computes the most
general constraints $ C_1 $ and substitution $ \sigma_1 $, because, for any $ C
$ and $ \sigma $ such that $(\mcconf \Delta b)\sigma $ type checks, $ C $ specialises $ C_1 $ and
$ \sigma $ is an instance of $ \sigma_1 $, after some extra substitution
$ \sigma' $ of variables (immaterial for type checking).


Algorithm $ \algoDual $ collects all $ \seql c \eta_1 $ and $ \seql {\bar c}
{\eta_2}$ constraints in $ C' $, generates duality constraints $
\duals{\eta_1}{\eta_2} $ and iteratively checks them,
possibly substituting $\psi$ variables. It ultimately returns a
$ C'' $ which is a valid type solution according to \cref{def:valid-type-solution}.
\begin{hide}
For example, if $ \eta_1 = \outp T_1.\eta_1' $ and $ \eta_2 = \inp T_2.\eta_2' $,
the empty substitution and the constraints $ \atomSu{\su{T_1}{T_2},
\duals{\eta_1'}{\eta_2'} }$ are generated.
When comparing internal and external choices, $ \algoDual $ checks that
all the branches in the internal choice are included in the
set of active branches in the external choice.
When one of the two sessions is a variable $ \psi $, the
algorithm collects all constraints $ \duals{\psi}{\eta_i} $ and calculates the
least supertype of the dual of all $ \eta_i $ sessions.
Algorithm $ \algoDual $ succeeds when no more simplification step can be taken,
and only duality constraints among session variables remain.
\end{hide}
Soundness and completeness of Algorithm $ \algoDual $ is straightforward.


We now show how $\algoSI$ infers the correct session types for \cref{ex:swap1} from \cref{sec:motivating-examples}. 
We assume that Algorithm $\algW $ has already produced a behaviour $ b $ and constraints $ C $ for this example.
For clarity, we simplify $b$ and $C$: we remove spurious $ \tau $s from behaviour sequences, replace region
variables $\rho$ with labels (only one label flows to each $\rho$), and perform simple substitutions of $\beta$ variables.
 
\begin{example}[A Swap Service]

There are three textual sources of endpoints in this example: the two
occurrences of $\eaccept{}{\keyw{swp}}$  in \textsf{coord}, and
$\erequest{}{\keyw{swp}}$ in \textsf{swap}. A pre-processing step automatically
annotates them with three unique labels $l_1, l_2 $ and $ l_3 $.
Algorithm $ \algW $ infers $ b $ and $ C $ for
\cref{ex:swap1}; the behaviour $ b $ (simplified) is:
$$ \espawn{ (\beta_{coord})}; \espawn{ (\beta_{
swap
})}; \espawn{ (\beta_{
swap
})}$$

In this behaviour three processes are spawned: one with a $
\beta_{coord} $ behaviour, and two with a $ \beta_{swap}$ behaviour. 
The behaviour associated to each of these variables is described in $ C $, along with
other constraints:
\begin{enumerate}
  \item 
$\orec{ (\pusho{ l_{1}} {\psi_{1}}{}; l_{1}?\alpha_{1}; 
         \pusho{l_{2}} {\psi_{1}}{}; 
         l_{2}?\alpha_{2}; l_{2}!\alpha_{1};
         l_{1}!\alpha_{2}); \beta_{coord}}{\beta_{coord}} 
\subseteq \beta_{coord}$
\item $ \pusho{ l_{3}} {\psi_{2}}{}; l_{3}!\tint; l_{3}?\alpha_{3} \subseteq
\beta_{swap}$
\item $ \overline{{swap}} \sim\psi_{1} $
\item $ {swap} \sim\psi_{2} $
\end{enumerate}

The above behaviour and environment are the inputs to Algorithm $\algoSI$, implementing session type inference according to the
second level of our framework.
The invocation $
\algoSI(b, C) $ calls $\mcclause
\epsilon b C \emptyK $, where the first $\epsilon$ is the empty endpoint stack $\Delta$ and the second $\epsilon$ is the empty
continuation stack. Behaviour $ b $ is decomposed as $ b = K[b']$, where 
 $ b' = \espawn{ (\beta_{ coord })} $ 
 and $ K 
$ is the continuation $ [~];\espawn{ (\beta_{ swap })}; \espawn{ (\beta_{ swap
})} $.
The algorithm thus calls
$\mcclause {\epsilon} {\espawn{ (\beta_{ coord })}} C K $, which, after replacing $\beta_{coord}$ and unfolding its
inner recursive behaviour becomes:
$$
\mcclause \epsilon {
\pusho{ l_{1}} {\psi_{1}}{}; l_{1}?\alpha_{1}; 
         \pusho{l_{2}} {\psi_{1}}{}; 
         l_{2}?\alpha_{2}; l_{2}!\alpha_{1};
         l_{1}!\alpha_{2}; \beta_{coord}}
 {C_1} \epsilon
$$

Here $ C_1 $ is equal to $ C $ above, with the exception of replacing Constraint 1 with the 
constraint $(\tau \subseteq \beta_{coord})$. Inference is now
straightforward: the frame $(l_1:\psi_1)$ is first pushed on the endpoint stack. From 
behaviour $ l_1?\alpha_1$ the algorithm applies substitution
$[\psi_1 \mapsto ?\alpha_4.\psi_4]$, where $ \psi_4$ and $ \alpha_4
$ are fresh, and generates constraint $ (\alpha_4 \subseteq \alpha_1) $
obtaining $C_2$. We thus get:
$$
\mcclause{
 (l_1:\psi_4) 
}{
         \pusho{l_{2}} {?\alpha_4.\psi_4}{}; 
         l_{2}?\alpha_{2}; l_{2}!\alpha_{1};
         l_{1}!\alpha_{2}; \beta_{coord}}
 {C_2} \epsilon
$$

After the next push, the endpoint stack becomes
$(l_2:?\alpha_4.\psi_4)\cdot(l_1:\psi_4)$. The next behaviour $ l_2?\alpha_2$
causes $ \algoMC $ to create constraint ($ \alpha_4 \subseteq \alpha_2$) obtaining $ C_3 $, and to
consume session $ ?\alpha_4$ from the top frame of the endpoint stack.
$$
\mcclause{
 (l_2:\psi_4)\cdot(l_1:\psi_4)
}{
         l_{2}!\alpha_{1}; l_{1}!\alpha_{2}; \beta_{coord}}
 {C_3} \epsilon
$$
Because of $l_2!\alpha_1$, $\algoMC $ generates
$[\psi_3 \mapsto {!}\alpha_5.\psi_5]$ and $ (\alpha_1
\subseteq \alpha_5)$ obtaining $C_4$. 
$$
\mcclause{
  (l_2:\psi_5)\cdot(l_1:!\alpha_5.\psi_5)
}{
         l_{1}!\alpha_{2}; \beta_{coord}}
 {C_4} \epsilon
$$
Since $l_1 $ in the behaviour and $l_2$ at the top of the endpoint stack do not match, $\algoMC$ infers that $ \psi_5$ must be the
terminated session $ \tend$. Therefore it substitutes
$[\psi_5 \mapsto \tend]$ obtaining $C_5$.
Because of the substitutions, $C_5$ contains 
$ \overline{{swap}} \sim {?}\alpha_4.!\alpha_5.\tend$. 
After analysing $ \beta_{swap} $, $\algoMC$ produces $C_6$ where $ swap \sim
{!}\tint.!\alpha_6.\tend$.

During the above execution $\algoMC$ verifies that 
the stack principle is respected and no endpoint label is pushed on the stack twice.
Finally the algorithm calls $\algoDual(C_6)$ which performs a duality check between the constraints of $\co{swap}$ and $swap$,
inferring
substitution $[\alpha_4 \mapsto \tint,~\alpha_6\mapsto\alpha_5]$.
The accumulated constraints on type variables $\alpha$ give the resulting session types of the swap channel endpoints:
$ (\overline{swap} \sim ?\tint.!\tint.\tend)$ and $(swap \sim !\tint.?\tint.\tend)$.

\end{example}
%
%

  \section{A Proposal for Recursive Session Types}
  \label{sec:extensions}
  \newcommand{\finarrowBase}[5]{#1\vDash #2\Downarrow_{#3, #4}^\mathsf{fin} #5}
\newcommand{\recarrowBase}[5]{#1\vDash #2\Downarrow_{#3, #4}^\mathsf{rec} #5}  

The system we have presented does not include recursive session types. Here we propose an extension to the type
system with recursive types. The inference algorithm for this extension is non-trivial and we leave it to future work.

In this extension, a recursive behaviour may partially use a recursive session type and rely on the continuation behaviour to
fully consume it.
First we add \emph{guarded} recursive session types:
$\eta ::= \ldots \bnf \mu X. \eta \bnf X$.
The first level of our type system remains unchanged, as it is parametric to
session types, and already contains recursive functions and behaviours.

A recursive behaviour $\orec{b}{\beta}$ operating on an endpoint $l$ with session type $\mu X. \eta$ may: (a) run in an infinite loop,
always unfolding the session type; (b) terminate leaving $l$ at type $\tend$; (c) terminate leaving $l$ at type $\mu X. \eta$.
Behaviour $b$ may have multiple execution paths, some terminating, ending at $\tau$, and some recursive, ending
at a recursive call $\beta$. They all need to leave $l$ at the same type, either $\tend$ or $\mu X. \eta$;
the terminating paths of $b$ determine which of the two session types $l$ will have after $\orec{b}{\beta}$.
If $b$ contains no terminating paths then we assume that $l$ is fully consumed by $\orec{b}{\beta}$ and type
the continuation with $l$ at $\tend$.

To achieve this, we add a 
\emph{stack environment} $D$ in the rules of \cref{fig:abstract-interpr}, which maps labels $ l $ to stacks $ \Delta $.
If $\Delta_1 = (l:\mu X.\eta)$,
we call an $l$-\emph{path} from $ \mc {\Delta_1} b_1$ any finite sequence of
transitions such that $\mc {\Delta_1} {b_1}\rightarrow_{C, D} \ldots 
\rightarrow_{C, D}
\mc {\Delta_n} {b_n} \not\rightarrow_{C, D} $.
A $l$-path is called 
$l$-\emph{finitary} if there is no $ b_i = \tau^l$ for any configuration $i $ in
the series; otherwise we say that the path is $l$-\emph{recursive}. 
We write $ \finarrowBase { (l:\mu X.\eta) }{b}{C} {D}
{\Delta'}$ when the last configuration of all $l$-finitary paths from 
$ \mc { (l:\mu X.\eta) }{b}  $ is $\mc{\Delta'}{\tau} $.
Similarly, we write $\recarrowBase { (l:\mu X.\eta) }{b}{C} {D} {\Delta'}$
when the last configuration of all $l$-recursive paths from 
$ \mc { (l:\mu X.\eta) }{b}  $ is $\mc{\Delta'}{\tau^l} $. 
When no $l-$ paths from $ \mc
{(l:\mu X.\eta)} b $ is $l$-finitary, we stipulate $ \finarrowBase { (l:\mu
X.\eta) }{b}{C} {D} {(l:\tend)}$ holds.
We add the following rules to those of \cref{fig:abstract-interpr}.

\begin{center}
  \begin{tabular}{l}
\begin{minipage}{.5\textwidth}
\irule*[Rec2][Rec2]
  {
   \nbox{
       \finarrowBase { (l:\mu X.\eta) }{b}{C'} {D'} {\Delta'}
    \\ \recarrowBase { (l:\mu X.\eta) }{b}{C'} {D'} {\Delta'}
   }
  }
  {
     \dstepBase{ \stBase{l}{\mu X. \eta}{\Delta}  }{  \orec{b}{\beta} }
           { \Delta'\cdot\Delta  }{  b' }{C,D}
  }
\end{minipage}
  \condBox{  
  \nbox{
              \Delta' \in \{(l:\tend), (l:\mu X. \eta)\}
            \\ C' = (C\removeRHS{\beta}) {\cup} (\su{\tau^l}{\beta})
            \\ D' = D \subst l {\Delta'} 
  }
  }
\\[4em]
\irule*[RCall][RCall]
  { \  }
  {
   \dstepBase {(l:\mu X.\eta)}{\tau^l}
              {D(l)}{\tau}
              {C, D} 
  }
\irule*[Unf][Unf]
  {
    \dstepBase{ \stBase l {\eta \subst{X}{\mu X.\eta } }  \Delta }{  b }
          { \Delta'}{ b'}{C,D}
  }
  {
    \dstepBase{ \stBase{l}{\mu X. \eta}{\Delta}  }{  b }
          { \Delta'  }{  b' }{C,D}
  }
  \end{tabular}
\end{center}
 
Rule \iref{Rec2}  requires that both $l$-finitary and $l$-recursive paths converge to the same stack $\Delta'$, either $ (l:\tend) $ or $
(l:\mu X.\eta) $. In this rule, similarly to rule \RefTirName{Rec} in \cref{fig:abstract-interpr}, we replace the recursive
constraint $(\orec{b}{\beta} \subseteq \beta)$ with $(\tau^l \subseteq \beta)$, representing a trivial recursive call of $\beta$.
This guarantees that all $l$-paths have a finite number of states. The $D$ environment is extended with $l \mapsto \Delta'$,
used in Rule \iref{RCall} to obtain the session type of $l$ after a recursive call.
Rule \iref{Unf} simply unfolds a recursive session type.

  \section{Related Work and Conclusions}
  \label{sec:relwork}
  We presented a new approach for adding binary session types to high-level programming languages, and applied it to a core of ML with
session communication. In the extended language our system checks the session protocols of interesting programs, including one
where pure code calls library code with communication effects, without having to refactor the pure code (\cref{ex:db}).
Type soundness guarantees partial lock freedom, session fidelity and
communication safety. 

Our approach is modular, organised in two levels, the first focusing on the type system of the source language and second on
typing sessions; the two levels communicate through \emph{effects}. In the fist level we adapted and extended the work of Amtoft,
Nielson and Nielson~\cite{ANN} to session communication, and used it to extract the communication effect of programs.
In the second level we developed a session typing discipline inspired by Castagna et al.~\cite{CastagnaEtal09}.
This modular approach achieves a provably complete session inference for finite sessions without programmer
annotations.

\begin{hide}
This extends foundational work, such as Milner's original polymorphic type inference 
and Tofte and Talpin's region analysis~\cite{TofteT94}. This effect is a term in a restricted process
algebra which lends itself well for session type checking.

Our session type discipline is presented as an abstract interpretation where $\stackop{}{b}{\Delta}{}$ means that all paths
from $b$ with stack $\Delta$ and environment $C$ reduce to a terminal configuration $\mcconf{\epsilon}{\tau}$. An equivalent
definition for $\stackop{}{b}{\Delta}{}$ can be given using inference rules.
\end{hide}

Another approach to checking session types in high-level languages is to use substructural type systems. For
example, Vasconcelos et al.~\cite{Vasconcelos2006} develop such a system for a functional language with
threads, and Wadler~\cite{Wadler2012} presents a linear functional language with effects. Type soundness in the former guarantees
session fidelity and communication safety, and in the latter also lock freedom and strong normalisation.
Our system is in between these two extremes: lock freedom is guaranteed only when processes do not diverge and their
requests for new sessions are met. Other systems give similar guarantees (e.g.,
\cite{PhenningSessionMonad,
CastagnaEtal09%
}).
\begin{hide}
The modularity of the two levels in our system may enable the use of other, perhaps more permissive disciplines.
\end{hide}


Toninho et al.~\cite{PhenningSessionMonad} add session-typed communication to a functional language using a monad. Monads, similar
to effects, cleanly separate session communication from the rest of the language features which, unlike effects, require parts of
the program to be written in a monadic style. Pucella and Tov~\cite{PucellaT08} use an
indexed monad to embed session types in Haskell,
however with
limited endpoint delegation: delegation relies on moving
capabilities, which cannot escape their static scope.
Our \cref{ex:swap2} is not typable in that system
 because of this.
 In~\cite{PucellaT08} session types are inferred by Haskell's type inference. However, the programmer must
guide inference with expressions solely used to
manipulate type structures.

Tov~\cite{TovThesis} has shown that session types can be encoded in a language with a general-purpose substructural type system.
Type inference alleviates the need for typing annotations in the examples
considered. Completeness of session inference relies on completeness of
inference in the general language, which is not 
clear. 
\begin{hide}
Our work gives
a provably complete session types inference system for \lang without the use of
substructural types, making it applicable to high-level programming languages
without such types.
\end{hide}

Igarashi et al.~\cite{IgarashiK00} propose a reconstruction
algorithm for finite types in the linear $ \pi $ calculus.
Inference is complete and requires no annotations. 
Padovani~\cite{Padovani15} extends this work to pairs, disjoint sums and regular recursive types.

Mezzina~\cite{Mezzina08} gives an inference algorithm for session types in a calculus of services. The type system does
not support recursive session types and endpoint delegation. It does allow, however to type replicated processes that only use
finite session types, similar to our approach.

\begin{hide}
We are currently working on extending our system to include more ML features such as exceptions and mutable state. We also intend
to improve the accuracy of our regions for approximating the endpoints in a program, using context-sensitive techniques from static
analysis (e.g., $k$-CFA~\cite{kcfa}). We are also currently working on a prototype implementation of our system.
\end{hide}

\begin{hide}
Scribble \cite{YoshidaHNN13} is a programming language with multiparty session types~\cite{HondaYC08}
\end{hide}
\begin{hide}
\todo{Added Tofte's regions and more relaxed typing for recursive sessions. 
I would reference Kobayashi's static analysis of deadlock freedom.}
We are currently working on extending our system to include more ML features, 
such as exceptions and mutable state, and to relax both the recursive session
typing and the stack discipline to accept more safe programs.
We also intend to improve the accuracy of our regions for approximating the
endpoints in a program, using the region inference of \cite{TofteT94},
which involves the non-trivial task of adding behaviours to their effect
systems. We believe that regions will provide more precise analysis of endpoint
life-times. We are also considering using context-sensitive techniques from
static analysis (e.g., $k$-CFA \cite{kcfa}).
\end{hide}

\bibliographystyle{splncs03}    
\bibliography{refs-short}
\clearpage
\appendix
\centerline{\Large\bf Appendix}
%
%
\newcommand{\qedhere}{\qed}
\section{Omitted definitions and examples}
\label{sec:omitted}
The following section presents the complete definitions of concepts and
examples that could not be spelled out in full details because of space
constraints.

\subsection{Full type system and definitions for the first level}
\[  \begin{array}{@{}l@{\quad}l@{}}
\irule*[TPair][tpair]
  {
    \tjr  {} {} {e_1} {T_1} {b_1}
    \quad
    \tjr  {} {} {e_2} {T_2} {b_2}
  } {
    \tjr{} {} {(e_1, e_2)} {T_1 \times T_2} { \bseq{b_1}{b_2} }
  }
  &
\irule*[TVar][tvar]
  { }
  { \tjr{}{} {x} {\Gamma(x)} {\tau}} 
\\[1em]
\irule*[TIf][tif]
  {
    \tjr{}{}{e_1}{\tbool}{b_1}
    \quad
    \tjr{}{}{e_i}{T}{b_i} {~}_{(i\in\{1,2\})}
  }
  { \tjr{}{}{\eif{e_1}{e_2}{e_3}}{T}{\bseq{b_1}{(b_2 \oplus b_3)} } }
&
\irule* [TConst][tconst]
  { }
  { \tjr{}{}{ k }{\mathit{typeof}(k)}{\tau} }
\\[1.4em]
\irule*[TApp][tapp]
  { \tjr{}{}{e_1}{ \tfun{T'}{T}{\beta} }{b_1} \quad \tjr{}{}{e_2}{T'}{b_2} } 
  { \tjr{}{}{\eapp{e_1}{e_2}}{ T}{ \bseq{\bseq{b_1}{b_2}}{\beta} } }
&
\irule*[TFun][tfun]
  { \tjr{}{,x:T}{e}{T'}{\beta} }
  {
  \tjr{}{}
      {\efunc{x}{e}}
      {\tfun{T}{T'}{\beta}}{\tau}}
\\[1em]
\irule*[TMatch][tmatch]
  {
    \tjr{}{}{e}{\tses{\rho}}{b}
    \quad \tjr{}{}{e_i}{T}{b_i}
    {~}_{(i\in I)}
  }
  { 
    \tjr{}{}
      { \ematch{e}{L_i : e_i}{i \in I} } 
      {T}{ \bseq{b}{\bechoice{}{\rho}{L_i}{b_i}} } 
  }
&
\irule*[TEndp][tendp]
  { }
  { \tjr{}{}{ \eendp{p}{l} }{ \tses{\rho}{} }{ \tau }}
  ~\condBox{$C\vdash\seql{\rho}{l}$}
\\[1.4em]
\irule*[TLet][tlet]
  { \tjr{}{}{e_1}{\TS}{b_1}
    \quad \tjr{
               }{,x:\TS}{e_2}{T}{b_2} }
  { \tjr{}{}{\elet{x}{e_1}{e_2}}{T}{\bseq{b_1}{b_2} } }
&
\irule*[TSub][tsub]
  { \tjr {} {} {e} {T} {b} }
  { \tjr {} {} {e} {T'} {\beta} }
  ~\condBox{\nbox[m]{
      \coType{T}{T'}\\
      \sub{b}{\beta}
  }}
\\[1em]
\irule[TSpawn][tspawn]{
  \tjrBase {C} { \confined_{C}(\Gamma) }
  {e}{\tfun{\tunit}{\tunit}{\beta}}{b}{}{}
}{
  \tjr{}{}
      {\espawn{e}}
      {\tunit}
      {\bseq{b}{\espawn{\beta}}}
}
&
\\[1em]
\multicolumn{2}{@{}l@{}}{
\irule[TRec][trec]
  { \tjrBase {C} {
      \confined_{C}(\Gamma), f: \tfun{T}{T'}{\beta}, x: T
    }{
      e
    }{T'}{b} {}{}
  }
  { \tjrBase {C} {\Gamma} {\efix
  {f} {x} {e} } {\tfun{T}{T'} {\beta}} { \tau} {}{}
}[~\condBox{\nbox[m]{
      \under{C}{\confined(T,T')}\\
      \sub{\orec{b}{\beta}}{\beta}
    }
}]
}
\\[1.4em]
\multicolumn{2}{@{}l@{}}{
\irule[TIns][tins]
  { \tjr{}{} {e} {\forall (\vec\gamma:C_0).T} {b} }
  { \tjr{}{} {e} {T\sigma}         {b} }
  [~\condBox{\nbox[m]{
      dom(\sigma) \subseteq \{\vec\gamma\}\\ 
      \typeschema{\vec\gamma}{C_0}{T} \text{ is solvable from $ C $ by } \sigma
  }}]
}
\\[1.4em]
\multicolumn{2}{@{}l@{}}{
\irule[TGen][tgen]
  { \tjr{\cup C_0}{} {e} {T}         {b} }
  { \tjr{}{} {e} {\forall (\vec\gamma:C_0).T} {b} }
  [~\condBox{\nbox[m]{
      \{\vec\gamma\} \cap \fv(\Gamma, C, b) = \emptyset\\
      \forall (\vec\gamma:C_0).T  \text{ is WF, solvable from } C
  }}]
}
\end{array}\]

\begin{definition}[Functional Subtyping] 
  \label{def:func-subtyping}
  $\subType{T}{T'}$ is the least reflexive, transitive, compatible relation on types with the axioms:
  \begin{gather*}
    \irule*{ (T_1 \subseteq T_2) \in C}{\subType{T_1}{T_2}}
    ~~~
    \irule*{ C \vdash \seql{\rho}{\rho'} } { \subType{\tses \rho} {\tses {\rho'}} }
    ~~~
    \irule*{
      \coType {T_1'}{T_1} \quad \sub{\beta}{\beta'} \quad \coType{T_2}{T_2'}
    }{
      \coType{ \tfun{T_1}{T_2}{\beta} }{ \tfun {T_1'}{T_2'}{\beta'} }
    }
  \end{gather*}
\end{definition}

\subsection{Confined types and behaviours}
\begin{definition}[Confined Behaviors]\label{def:confined-behav}
  $C \vdash \confined(b)$ is the least compatible relation on 
  behaviours that admits the following axioms:
  
\begin{longtable}{@{}l@{\qquad}l@{}} 
  \irule*[CTau][CTau]
    { b \in \{\tau,~\orec{b'}{\beta}\}}
    {C\vdash \confined(b)}
&
  \irule*[CAx-b][CAx-b]
    { \cfd b \in C }
    { C \vdash \confined(b) }
 \\[1.4em] 
  \irule*[CICh][CICh]
    {C\vdash\confined(b_1) \quad 
     C\vdash\confined(b_2)}
    {C\vdash\confined(b_1 \oplus b_2)}
&
  \irule*[CICh-Bw][CICh-Bw]
    {C\vdash\confined(b_1\oplus b_2)}
    {C\vdash\confined(b_i)}
    \condBox{$i\in \{1,2\}$}
\\[1.4em] 
  \irule*[CSeq][CSeq]
    {C\vdash\confined(b_1) \quad 
     C\vdash\confined(b_2)}
    {C\vdash\confined(b_1; b_2)}
& 
  \irule*[CSeq-Bw][CSeq-Bw]
    {C\vdash\confined(b_1;b_2)}
    {C\vdash\confined(b_i)}
    \condBox{$i\in \{1,2\}$}
  \\[1.4em]
  \irule*[CSpw][CSpw]
    {C\vdash\confined(b)}
    {C\vdash\confined(\spawno{b}{})}
& 
  \irule*[CSpw-Bw][CSpw-Bw]
    {C\vdash\confined(\spawno{b}{})}
    {C\vdash\confined(b)}
\\[1.4em]
  \irule*[CSub-b][CSub-b]
    { \sub{b_1}{b_2} \quad C\vdash \confined(b_2)  }
    { C \vdash \confined(b_1) }
    &
\end{longtable}
\end{definition}

\begin{definition}[Confined Types]\label{def:confined-types}
  $C \vdash \confined(o)$ is the least compatible relation on type schemas and
  types that admits the following axioms.
%
\begin{longtable}{@{}l@{\quad}l@{}}
  \irule*[CCons][CCons]
      { T \in \{\tint,\tbool,\tunit\} }
      { C \vdash \confined(T) }
  &
  \irule*[CAx-T][CAx-T]
      { \cfd T \in C }
      { C \vdash \confined(T) }
\\[1.4em]
    \irule*[CSub-T][CSub-T]
      { \sub{T_1}{T_2} \quad C\vdash \confined(T_2)  }
      { C \vdash \confined(T_1) }
&
    \irule*[CSub-T-Bw][CSub-T-Bw]
      { \forall T_1. (\sub{T_1}{T_2}) \Longrightarrow C\vdash \confined(T_1)  }
      { C \vdash \confined(T_2) }
\\[1.4em]
    \irule*[CFun][CFun]{
      C\vdash \confined(T,T')
      \quad
      C\vdash \confined(\beta)
    }{
      C\vdash \confined(\tfun{T}{T'}{\beta})
    }
&
    \irule*[CFun-Bw][CFun-Bw]{
      C\vdash \confined(\tfun{T}{T'}{\beta}) \quad o \in \{T,T',\beta\}
    }{
      C\vdash \confined(o)
    }
\\[1.4em]
    \irule*[CTup][CTup]{
      C\vdash \confined(T,T')
    }{
      C\vdash \confined(\tpair{T}{T'})
    }
&
    \irule*[CTup-Bw][CTup-Bw]
    {
      C\vdash \confined(\tpair{T}{T'}) \quad o \in \{T,T'\}
    }{
      C\vdash \confined(o)
    }
\\[1.4em]
   \irule*[CTS][CTS]{
     C,C_0 \vdash \confined(T) 
    }{
      C\vdash \confined( \forall (\vec\gamma: C_0). T )
    }
    &
\end{longtable}
  We define $\confined_C(\Gamma)$ as the largest subset of $\Gamma$ such that
  for all bindings $ (x:\TS) \in \confined_C(\Gamma)$ we have
  $C\vdash\confined(\TS)$.
\end{definition}

The above definitions admit behaviours constructed by $\tau$ and
recursive behaviours (\iref{CTau}), and types that are constructed by such
behaviours and the base types $\tint,\tbool,\tunit$ (\iref{CCons}).
The definitions allow for type and behaviour variables ($\alpha$, $\beta$) as
long as they are only related to confined types in $C$
(\iref{CAx-b}, \iref{CAx-T}). Sub-behaviours of confined behaviours and
sub-types of confined types are confined too (\iref{CSub-b}, \iref{CSub-T}).
Type schemas are confined if all their instantiations in $C$ are
confined (\iref{CTS}).
The definitions contain composition rules for composite
behaviours and types (\iref{CICh}, \iref{CSeq}, \iref{CSpw}, \iref{CFun},
\iref{CTup}); they also contain decomposition (or ``\textit{backward}'') rules
for the same composite constructs (\iref{CICh-Bw}, \iref{CSeq-Bw}, \iref{CSpw-Bw}, \iref{CFun-Bw},
\iref{CTup-Bw}).
The backward rules ensure that all the sub-components of a
confined behaviour or type in $ C $ are also confined.


\subsection{Active and inactive labels}
  Consider a program $P[e_1][e_2]$ containing the expressions:
  \[
    \begin{array}{@{}r@{~}c@{~}l@{}}
      e_1 &\defeq& \elet{x}{\eapp{\eaccept{l_1}{c}}{\cunit}}{\ematch{x}{L_1 \Rightarrow e,~L_2 \Rightarrow e^{\star}}{}}
      \\
      e_2 &\defeq& \elet{x}{\eapp{\erequest{l_2}{c}}{\cunit}}{\eapp{\eselectnew{L_2}}{x}}
    \end{array}
  \]
  Suppose $e^\star$ contains a type error, possibly because of a mismatch in session types with another part of $P$. If a type
  inference algorithm run on $P[e_1][e_2]$ first examines $e_1$, it will explore both branches of the choice, tentatively
  constructing the session type $\Sigma\{L_1.\eta_1,~L_2.\eta_2\}$, finding the error in $e^\star$. One strategy might then be to
  backtrack from typing $e^{\star}$ (and discard any information learned in the $L_2$ branch of this and possibly other choices in the code)
  and continue with the session type $\Sigma\{L_1.\eta_1\}$. However, once $e_2$ is encountered, the previous error in $e^\star$
  should be reported. A programmer, after successfully type checking $P[e_1][()]$, will be surprised to discover a type error in
  $e_1$ after adding in $e_2$. The type-and-effect system here avoids such situations by typing all choice branches, even if they
  are \emph{inactive}, at the expense of rejecting some---rather contrived---programs. A similar approach is followed in the
  type-and-effect system of the previous section by requiring \emph{all} branches to have the same type (Rule~\iref{tmatch} in
  \cref{fig:typing-rules}).

\subsection{Duality}
The
program: \[\nbox{
    \espawn{(\efunc{\_}{\elet{x}{\eapp{\erequest{}{c}}{\cunit}}{\eapp{\esend{x}}{\keyw{\ctrue}}}})};~
  \elet{x}{\eapp{\eaccept{}{c}}{\cunit}}{\erecv{}{x} + \keyw{1}}
}\]
should not be typable because its processes use dual endpoints at incompatible
session types. Therefore session types on dual session endpoints
($\seql{c}{\eta}$, $\seql{\co c}{\eta'}$) must be dual, where duality is
defined as follows:

\begin{definition}[Duality] \label{def:duality}
$ C\vdash \eta \dual{} \eta'$ if the following rules and their symmetric ones are satisfied.
\begin{gather*}
  \irule*{ }{
    C \vdash  \tend \dual{} \tend
  }
  \quad
  \irule*{
    C \vdash  \subt{T}{T'}
    \\\\
    C \vdash  \eta \dual{} \eta'
  }{
    C \vdash  {!T}.\eta \dual{} {?T'}.\eta'
  }
  \quad
  \irule*{
    C \vdash  \subt{\eta_0}{\eta_0'}
    \\\\
    C \vdash  \eta \dual{} \eta'
  }{
    C \vdash  {!\eta_0}.\eta \dual{} {?\eta_0'}.\eta'
  }
  \\
  \irule*{
    \forall i\in I_0.~~ C \vdash  \eta_i \dual{} \eta'_i
  }{
    C \vdash  {\sichoiceBase{i}{I_0}{\eta}{i}} \dual{} {\sechoiceBase{i}{I_0I_1}{I_2}{\eta'}{i}})
  }
\end{gather*}
where $C \vdash  \subt{\eta}{\eta'}$ is Gay\&Hole \cite{GayHole05} subtyping, with $C$ needed for 
inner uses of $C \vdash  \subt{T}{T'}$, extended to our form of external choice, where
$C \vdash\sechoicetext{i}{I_1}{I_2}{\eta'}
\subt{}
\sechoicetext{i}{J_1}{J_2}{\eta'}$
when
$I_1\subseteq J_1$
and
$J_1\cup J_2\subseteq I_1\cup I_2$
and
$\forall( i \in J_1\cup J_2).~
C \vdash  \subt{\eta_i}{\eta_i'}$.

\end{definition}

\subsection{Dependencies}
Thanks to the definition of well-stackedness, we can give a more precise account
of type soundness. We can define dependencies between processes of a running
system according to the following definition.
\begin{definition}[Dependencies]
Let $ P = \tconf{\Delta}{b}{e}$ and $ Q = \tconf{\Delta'}{b'}{e'} $ be processes in $ \c S $.
\begin{description}
  \item[$ P $ and $ Q $ are \emph{ready} ($ P \leftrightharpoons Q$):] if
    $\Delta = (p^l:\eta)\cdot\Delta_0$ and
    $\Delta' = (\co p^{l'}:\eta') \cdot \Delta_0'$;

  \item[$ P $ is \emph{waiting} on $ Q $ ($P \mapsto Q$):] if
    $ \Delta = (p^l:\eta) \cdot \Delta_0$ and
    $ \Delta' = \Delta_1' \cdot (\co p^{l'}:\eta') \cdot \Delta_0' $
    and $\Delta_1' \neq \epsilon $;

  \item[$ P $ \emph{depends} on $ Q, R $ ($P \Mapsto (Q,R)$):] if
    $P = Q \leftrightharpoons R $, or
    $P \mapsto^+ Q \leftrightharpoons R $.

\end{description}
\end{definition}

The type soundness theorem (\cref{thm:type-soundness}) can be reformulated to
show that blocked processes depend on a process that is either waiting,
diverging or blocked:
\begin{corollary}[Type Soundness with dependencies]
\label{thm:type-soundness}
 Let $ \tjParSys{}{\c S}$ and $ \wst{\c S}$. Then
 \begin{enumerate}
   \item $ \c S \comStepW  \c S'$, or 
   \item $ \c S \intStep^* (\c F,\c D,\c W,\c B) $ such that:
 \begin{description}
   \item[Processes in $\c F$ are finished:] $ \forall\tconf{\Delta}{b}{e} \in \c F.~ \Delta=\epsilon, b=\tau$ and $e=v$.
   \item[Processes in $\c D$ diverge:] $ \forall\tconf{\Delta}{b}{e} \in \c D.~ \tconf{\Delta}{b}{e} \intStep^\infty$.
   \item[Processes in $\c W$ wait on channels:] $ \forall\tconf{\Delta}{b}{e} \in \c W.~ e = E[\erequest{l}{c}] $ or
     $ e = E[\eaccept{l}{c}]$.
   \item[Processes in $\c B$ block on sessions:]
     $\forall P=\tconf{\Delta}{b}{e} \in \c B.~ e = E[e_0]$ and 
     $e_0$ is
     $\esend{v}{}$, $\eapp{\erecv{}{}}{v}$, $\edeleg{v}$, $\eapp{\eresume{}{}}{v}$,
   $\eapp{\eselectnew{L}}{v}$, or $\ematch{v}{L_i \Rightarrow e_i}{i \in I}$
     and
     $\exists Q \in (\c D, \c W).~ \exists R \in (\c D, \c W, \c B).~  \c S\vdash P \Mapsto (Q, R)$.
 \end{description}
 \end{enumerate}
 \end{corollary}

\subsection{Subtyping in recursive session types}
We add a brief discussion about subtyping for the recursive session
types in \cref{sec:extensions}. We only consider recursive session types that
are guarded, and therefore contractive; thus the subtyping relation for the
extended session types only needs to perform a standard unfolding of recursive
types \cite{GayHole05}.
Note that the definition of functional subtyping (\cref{def:func-subtyping}) is
unaffected because it involves behaviours but not session types.

Subtyping is used in the definition of duality (\cref{def:duality}), which remains unaffected. Our inference algorithm relies on a
decision procedure for duality ($\mathcal{D}$) which now needs to use a decision procedure for subtyping of equi-recursive types,
such as the algorithm of Kozen, Parsberg and Schwartzbach%
\begin{icalp}
\footnote{Kozen, D., Palsberg, J., Schwartzbach, M.I.: Efficient recursive subtyping.
  Mathematical Structures in Computer Science  5(1),  113--125 (1995)}.
\end{icalp}
\begin{hide}
\cite{KozenEtalSubtyping}.
\end{hide}
This is the main modification of our
inference algorithm when adding recursive session types.

\clearpage
\section{Type Soundness}
\label{sec:type-soundness}
\def\nxt{\mathsf{nxt}}

To prove type soundness (progress and preservation) we first define the typed reductions of $\c S$ configurations shown in
\cref{fig:typed-reductions}. Recall that we let $\c S$ range over $\vect{\tconf{\Delta}{b}{e}}$, 
write $S$ for $\vect{e}$ when $\c S = \vect{\tconf{\Delta}{b}{e}}$, and identify $\c S$ and $S$ up to reordering.
Here we write $\vec b$ for an arbitrary sequential composition of behaviours, which may be empty ($\epsilon$). We also
superscripts in sequences of terms to identify them (e.g., $\vec\tau^1$ may be different than $\vec\tau^2$), and we
identify sequential compositions up to associativity and the axiom $\bseq{\epsilon}{b}=\bseq{b}{\epsilon}=b$.

In \cref{sec:stage-1} we assumed that programs are annotated by unique region
labels in a pre-processing step. This is necessary to achieve the maximum accuracy of our system (and reject fewer programs). Because beta reductions can duplicate annotations here
we drop the well-annotated property. Type soundness for general annotated programs implies type soundness for uniquely annotated
programs.

\begin{myfigure*}{Typed Reductions.}{fig:typed-reductions}\small 
  \[\begin{array}{@{}l@{\qquad}l@{}}
\irule*[TREnd][trend]{
  \tconf {\Delta} {b} {e},~ \c S \tred \c S'
}{
  \tconf {(l:\tend)\cdot\Delta} {b} {e},~ \c S \tred \c S'
}
&
\irule*[TRBeta][trbeta]{ 
  \sub{b'}{\beta}
  \\
  \tconf {\Delta} {b[b'/\beta]} {e},~ \c S \tred \c S'
}{
  \tconf {\Delta} {b} {e},~ \c S \tred \c S'
}
\\[1.5em]
\multicolumn{2}{@{}l@{}}{
\irule*[TRIft][trift]{ }{
  \tconf {\Delta} {\bseq{\vec\tau}{\bseq{\tau}{\bseq{(b_1 \oplus b_2)}{\vec b_\nxt}}}} {E[\eif{\ctrue}{e_1}{e_2}]}
  \tred
  \tconf {\Delta} {\bseq{\vec\tau}{\bseq{b_1}{\vec b_\nxt}}} {E[e_1]}
}
}
\\[1.5em]
\multicolumn{2}{@{}l@{}}{
\irule*[TRIff][triff]{ }{
  \tconf {\Delta} {\bseq{\vec\tau}{\bseq{\tau}{\bseq{(b_1 \oplus b_2)}{\vec b_\nxt}}}} {E[\eif{\cfalse}{e_1}{e_2}]}
  \tred
  \tconf {\Delta} {\bseq{\vec\tau}{\bseq{b_2}{\vec b_\nxt}}} {E[e_2]}
}
}
\\[1.5em]
\multicolumn{2}{@{}l@{}}{
\irule*[TRLet][trlet]{ }{
  \tconf{\Delta}{\bseq{\vec\tau}{\bseq{\tau}{\vec b_\nxt}}}{E[\elet{x}{v}{e}]}
  \tred
  \tconf{\Delta}{\bseq{\vec\tau}{\vec b_\nxt}}{E[e[v/x]]}
}
}
\\[1.5em]
\multicolumn{2}{@{}l@{}}{
\irule*[TRApp][trapp]{ }{
  \tconf{\Delta}{\bseq{\vec\tau}{\bseq{\tau}{\bseq{\tau}{\vec b_\nxt}}}}{E[\eapp{(\efunc{x}{e})}{v}]}
  \tred
  \tconf{\Delta}{\bseq{\vec\tau}{\vec b_\nxt}}{E[e[v/x]]}
}
}
\\[1.5em]
\multicolumn{2}{@{}l@{}}{
\irule*[TRRec][trrec]{ }{\nbox{
  \tconf {\Delta} {\bseq{\vec\tau}{\bseq{\tau}{\bseq{\tau}{\bseq{\orec{b}{\beta}}{\vec b_\nxt}}}}} {E[\eapp{(\efix{f}{x}{e})}{v}]}
  \\\tred
  \tconf {\Delta} {\bseq{\vec\tau}{\bseq{b[\orec{b}{\beta}/\beta]}{\vec b_\nxt}}} {E[e[\efix{f}{x}{e}/f][v/x]]}
}}
}
\\[3em]
\multicolumn{2}{@{}l@{}}{
\irule*[TRSpn][trspwn]{ }{
  \tconf {\Delta} {\bseq{\vec\tau}{\bseq{\tau}{\bseq{(\espawn{b})}{\vec b_\nxt}}}} {E[\espawn{v}]},~ \c S
  \tred
  \tconf {\Delta} {\bseq{\vec\tau}{\bseq{\tau}{\vec b_\nxt}}} {E[()]},~
  \tconf {\epsilon} {\bseq{\tau}{\bseq{\tau}{b}}} {\eapp{v}{()}},~ \c S
}
}
\\[1.5em]
\multicolumn{2}{@{}l@{}}{
\irule*[TRInit][trinit]{
   p,\co p\freshfrom E_1,E_2,\c S,\Delta_1,\Delta_2
}{\nbox{
    \tconf{\Delta_1} {\bseq{\vec\tau^1}{\bseq{\tau}{\bseq{\tau}{\bseq{\pusho{l_1}{\eta_1}{}}{\vec b_{\nxt 1}}}}}} {E_1[\eapp{\erequest{l_1}{c}}{()}]},~
    \\
    \tconf{\Delta_2} {\bseq{\vec\tau^2}{\bseq{\tau}{\bseq{\tau}{\bseq{\pusho{l_2}{\eta_2}{}}{\vec b_{\nxt 2}}}}}} {E_2[\eapp{\eaccept{l_2}{c}}{()}]},~ \c S
   \\\tred
   \tconf{(p^{l_1} : \eta_1)\cdot\Delta_1} {\bseq{\vec\tau^1}{\bseq{\tau}{\vec b_{\nxt 1}}}} {E_1[\eendp{p^{l_1}}{}]},~
   \tconf{(\co p^{l_2} : \eta_2)\cdot\Delta_2} {\bseq{\vec\tau^2}{\bseq{\tau}{\vec b_{\nxt 2}}}} {E_2[\eendp{\co p^{l_2}}{}]},~ \c S
}}
}
\\[4.5em]
\multicolumn{2}{@{}l@{}}{
\irule*[TRCom][trcom]{ }{\nbox{
    \tconf{(p^{l_1}:{!T_1}.\eta_1)\cdot\Delta_1}{
      \bseq{\vec\tau^1}{\bseq{\tau}{\bseq{\tau}{\bseq{\tau}{\bseq{\popo{\rho_1}{!T_1'}}{\vec b_{\nxt 1}}}}}}
    }{
      E_1[(\esend{(\eendp{p^{l_1}}{}},{v}))]
    },~
  \\
  \tconf{(\co p^{l_2}:{?T_2}.\eta_2)\cdot\Delta_2}{
    \bseq{\vec\tau^2}{\bseq{\tau}{\bseq{\tau}{\bseq{\popo{\rho_2}{?T_2'}}{\vec b_{\nxt 2}}}}}
  }{
    E_2[\erecv{}{\eendp{\co p^{l_2}}{}}]
  },~ \c S
  \\\tred
  \tconf{(p^{l_1}:\eta_1)\cdot\Delta_1} {\bseq{\vec\tau^1}{\bseq{\tau}{\vec b_{\nxt 1}}}} {E_1[()]},~
  \tconf{(\co p^{l_2}:\eta_2)\cdot\Delta_2} {\bseq{\vec\tau^2}{\bseq{\tau}{\vec b_{\nxt 2}}}} {E_2[v]},~ \c S
}}
}
\\[4.5em]
\multicolumn{2}{@{}l@{}}{
  \irule*[TRDel][trdel]{
  }{\nbox{
  \tconf{
    (p^{l_1}:{!\eta_{d}}.\eta_1)\cdot(p'^{l_1'}:\eta_{d}')\cdot\Delta_1
  }{
    \bseq{\vec\tau^1}{\bseq{\tau}{\bseq{\tau}{\bseq{\tau}{\bseq{\popo{\rho_1}{!\rho_d}}{\vec b_{\nxt 1}}}}}}
  }{
    E_1[\edeleg{(\eendp{p}{l_1}}, {\eendp{p'}{l_1'})}]
  },~
  \\
  \tconf{
    (\co p^{l_2}:{?\eta_{r}}.\eta_2)
  }{
    \bseq{\vec\tau^2}{\bseq{\tau}{\bseq{\tau}{\bseq{\popo{\rho_2}{?\rho_r}}{\vec b_{\nxt 2}}}}}
  }{
    E_2[\eresume{l_r}{\eendp{\co p}{l_2}} ]
  },~ \c S
  \\\tred
  \tconf{(p^{l_1}:\eta_1)\cdot\Delta_1} {\bseq{\vec\tau^1}{\bseq{\tau}{\vec b_{\nxt 1}}}} {E_1[()]},~
  \tconf{(\co p^{l_2}:\eta_2)\cdot(p'^{l_r}:\eta_r)} {\bseq{\vec\tau^2}{\bseq{\tau}{\vec b_{\nxt 2}}}} {E_2[\eendp{p'}{l_r}]},~ \c S
}}
}
\\[4.5em]
\multicolumn{2}{@{}l@{}}{
\irule*[TRSel][trsel]{
  k\in I
}{\nbox{
  \tconf{
    \left(p^{l_1} : {\sichoicetext{i}{I}{\eta}}\right)\cdot\Delta_1
  }{
    \bseq{\vec\tau^1}{\bseq{\tau}{\bseq{\tau}{\bseq{\popo{\rho_1}{!L_k}}{\vec b_{\nxt 1}}}}}
  }{
    E_1[\eapp{\eselectnew{L_k}}{\eendp{p}{l_1}}]
  },~
  \\
  \tconf{
    \left(\co p^{l_2}:{\sechoicetext{i}{I_1}{I_2}{\eta'}}\right)\cdot\Delta_2
  }{
    \bseq{\vec\tau^2}{\bseq{\tau}{\bseq{\left(\bechoicetext{j\in J}{\rho_2}{L_j}{b_j}\right)}{\vec b_{\nxt 2}}}}
  }{
    E_2[\ematch{\eendp{\co p}{l_2}}{L_j \Rightarrow e_j}{j\in J}]
  },~ \c S
  \\\tred
  \tconf{(p^{l_1}:\eta_k)\cdot\Delta_1}{\bseq{\vec\tau^1}{\bseq{\tau}{\vec b_{\nxt 1}}}}{E_1[()]},~
  \tconf{(\co p^{l_2}:\eta_k')\cdot\Delta_2}{\bseq{\vec\tau^2}{\bseq{b_k}{\vec b_{\nxt 2}}}}{E_2[e_k]},~ \c S
}}
}
\end{array}\]
\end{myfigure*}

\begin{lemma}[Weakening] \label{lem:type-weakening} 
  Suppose $\tjr{}{}{e}{\TS}{b}$. Then 
  $\tjr{}{\uplus \Gamma'}{e}{\TS}{b}$.
\end{lemma}
\begin{proof}
  By induction on $\tjr{}{}{e}{\TS}{b}$.
\end{proof}
\begin{lemma}[Type Decomposition] \label{lem:type-decomp} 
  Suppose $\tjr{}{}{E[e]}{\TS}{b}$. Then there exist $\TS'$, $b'$, $\vec b_\nxt$ and fresh $x$ such that
  $b = \bseq{\vec\tau}{\bseq{b'}{\vec b_\nxt}}$
  and
  $\tjr{}{}{e}{\TS'}{b'}$
  and
  $\tjr{}{, x{:}\TS'}{E[x]}{\TS}{\bseq{\vec\tau}{\bseq{\tau}{b_\nxt}}}$.
\end{lemma}
\begin{proof}
  By structural induction on $E$.
\end{proof}
\begin{lemma}[Type Composition] \label{lem:type-comp} 
  Suppose
  $\tjr{}{, x{:}\TS'}{E[x]}{\TS}{b}$
  and
  $\tjr{}{}{e}{\TS'}{b'}$.
  Then there exists $\vec b_\nxt$ such that
  $b = \bseq{\vec\tau}{\bseq{\tau}{\vec b_\nxt}}$
  and
  $\tjr{}{}{E[e]}{\TS}{\bseq{\vec\tau}{\bseq{b'}{\vec b_\nxt}}}$.
\end{lemma}
\begin{proof}
  By structural induction on $E$ using Lem.\ \ref{lem:type-weakening}.
\end{proof}
\begin{lemma} \label{lem:wst-subtype} 
  Suppose
  $\wst{\c S, \tconf{\Delta \cdot (p^l:\eta) \cdot \Delta'}{b}{e}}$
  and $C \vdash \subt{\eta}{\eta'}$.
  Then
  $\wst{\c S, \tconf{\Delta \cdot (p^l:\eta') \cdot \Delta'}{b}{e}}$.
\end{lemma}

\subsection{Preservation} 

Preservation relies on two lemmas, the first is that typed reductions  of \cref{fig:typed-reductions} preserve well-typedness,
well-stackedness and well-annotatedness; the other is that untyped reductions can be simulated by the typed reductions.

\begin{lemma}
  Suppose $\tjParSys{}{\c S}$ and $\wst{\c S}$ and $\c S \tred \c S'$. Then $\tjParSys{}{\c S'}$ and $\wst{\c S'}$.
\end{lemma}
\begin{proof}
  By induction on $\c S \tred \c S'$ using the type composition and decomposition (Lem.(s) \ref{lem:type-comp} and
  \ref{lem:type-decomp}).
  
  The most interesting case is that of delegation (\iref{trdel}). In this case we have
  \[\c S = 
    \nbox{
      \tconf{
        (p^{l_1}:{!\eta_{d}}.\eta_1)\cdot(p'^{l_1'}:\eta_{d}')\cdot\Delta_1
      }{
        \bseq{\vec\tau^1}{\bseq{\tau}{\bseq{\tau}{\bseq{\tau}{\bseq{\popo{\rho_1}{!\rho_d}}{\vec b_{\nxt 1}}}}}}
      }{
        E_1[\edeleg{(\eendp{p}{l_1}}, {\eendp{p'}{l_1'})}]
      },~
      \\
      \tconf{
        (\co p^{l_2}:{?\eta_{r}}.\eta_2)
      }{
        \bseq{\vec\tau^2}{\bseq{\tau}{\bseq{\tau}{\bseq{\popo{\rho_2}{?\rho_r}}{\vec b_{\nxt 2}}}}}
      }{
        E_2[\eresume{l_r}{\eendp{\co p}{l_2}} ]
      },~ \c S_0
      \\\tred
      \tconf{(p^{l_1}:\eta_1)\cdot\Delta_1} {\bseq{\vec\tau^1}{\bseq{\tau}{\vec b_{\nxt 1}}}} {E_1[()]},~
      \tconf{(\co p^{l_2}:\eta_2)\cdot(p'^{l_r}:\eta_r)} {\bseq{\vec\tau^2}{\bseq{\tau}{\vec b_{\nxt 2}}}}
      {E_2[\eendp{p'}{l_r}]},~ \c S_0
     = \c S'}
  \]
  and by $\tjParSys{}{\c S}$:
  \begin{gather*}
    \tjrBase{C}{\emptyset}{
      E_1[\edeleg{(\eendp{p}{l_1}}, {\eendp{p'}{l_1'})}]
    }{T_1}{
      \bseq{\vec\tau^1}{\bseq{\tau}{\bseq{\tau}{\bseq{\tau}{\bseq{\popo{\rho_1}{!\rho_d}}{\vec b_{\nxt 1}}}}}}
    }{}{}
    \\
    \tjrBase{C}{\emptyset}{
      E_2[\eresume{l_r}{\eendp{\co p}{l_2}} ]
    }{T_2}{
      \bseq{\vec\tau^2}{\bseq{\tau}{\bseq{\tau}{\bseq{\popo{\rho_2}{?\rho_r}}{\vec b_{\nxt 2}}}}}
    }{}{}
    \\
    \stackop{}{
      \bseq{\vec\tau^1}{\bseq{\tau}{\bseq{\tau}{\bseq{\tau}{\bseq{\popo{\rho_1}{!\rho_d}}{\vec b_{\nxt 1}}}}}}
    }{
      (p^{l_1}:{!\eta_{d}}.\eta_1)\cdot(p'^{l_1'}:\eta_{d}')\cdot\Delta_1
    }{}
    \\
    \stackop{}{
      \bseq{\vec\tau^2}{\bseq{\tau}{\bseq{\tau}{\bseq{\popo{\rho_2}{?\rho_r}}{\vec b_{\nxt 2}}}}}
    }{
      (\co p^{l_2}:{?\eta_{r}}.\eta_2)
    }{}
  \end{gather*}
  By type decomposition Lem.\ \ref{lem:type-decomp} and inversion on the rules of \cref{fig:typing-rules,fig:abstract-interpr}:
  \begin{align*}
    &
    \tjrBase{C}{\emptyset}{
      \edeleg{(\eendp{p}{l_1}}, {\eendp{p'}{l_1'})}
    }{\tunit}{
      \bseq{\tau}{\bseq{\tau}{\bseq{\tau}{\popo{\rho_1}{!\rho_d}}}}
    }{}{}
    &&
    \tjrBase{C}{x{:}\tunit}{
      E_1[x]
    }{T_1}{
      \bseq{\vec\tau^1}{\bseq{\tau}{\vec b_{\nxt 1}}}
    }{}{}
     \\
    &
    \tjrBase{C}{\emptyset}{
      \eresume{}{\eendp{\co p}{l_2}} 
    }{\tses{\rho_r}}{
      \bseq{\tau}{\bseq{\tau}{\popo{\rho_2}{?\rho_r}}}
    }{}{}
    &&
    \tjrBase{C}{x{:}\tses{\rho_r}}{
      E_2[x]
    }{T_2}{
      \bseq{\vec\tau^2}{\bseq{\tau}{\vec b_{\nxt 2}}}
    }{}{}
    \\
    &
    \stackop{}{
      \bseq{\tau}{\vec b_{\nxt 1}}
    }{
      (p^{l_1}:\eta_1)\cdot\Delta_1
    }{}
    &&
    C \vdash \seql{\rho_1}{l_1},~
    \seql{\rho_d}{l_1'},~
    \subt{\eta_d'}{\eta_d}
    \\
    &
    \stackop{}{
      \bseq{\tau}{\vec b_{\nxt 2}}
    }{
      (\co p^{l_2}:\eta_2) \cdot
      (p'^{l_r}:\eta_{r})
    }{}
    &&
    C \vdash \seql{\rho_2}{l_2},~
    \seql{\rho_r}{l_r}
  \end{align*}
  Note that the transition rules considered in $\Downarrow_C$ do not take into account the concrete endpoints $p$, $p'$ and $\co
  p$---they are existentially quantified in these rules. By Lem.\ \ref{lem:type-comp} and the rules of \cref{fig:abstract-interpr}:
  \begin{align*}
    &
    \tjrBase{C}{\emptyset}{
      E_1[\cunit]
    }{T_1}{
      \bseq{\vec\tau^1}{\bseq{\tau}{\vec b_{\nxt 1}}}
    }{}{}
    &&
    \stackop{}{
      \bseq{\vec\tau^1}{\bseq{\tau}{\vec b_{\nxt 1}}}
    }{
      (p^{l_1}:\eta_1)\cdot\Delta_1
    }{}
     \\
    &
    \tjrBase{C}{\emptyset}{
      E_2[p'^{l_r}]
    }{T_2}{
      \bseq{\vec\tau^2}{\bseq{\tau}{\vec b_{\nxt 2}}}
    }{}{}
    &&
    \stackop{}{
      \bseq{\vec\tau^2}{\bseq{\tau}{\vec b_{\nxt 2}}}
    }{
      (\co p^{l_2}:\eta_2) \cdot
      (p'^{l_r}:\eta_{r})
    }{}
  \end{align*}
  Therefore $\tjParSys{}{\c S'}$.

  From $\wst{\c S}$ we deduce:
  \begin{align*}
    &
    \wst{\c S_0,~ \tconf{
      (p^{l_1}:{!\eta_{d}}.\eta_1)\cdot(p'^{l_1'}:\eta_{d}')\cdot\Delta_1
    }{b_1}{e_1},~ \tconf{
      (\co p^{l_2}:{?\eta_{r}}.\eta_2)
  }{b_2}{e_2}}
  \\
  &
    \wst{\c S_0,~ \tconf{
      (p'^{l_1'}:\eta_{d}')\cdot\Delta_1
    }{b_1}{e_1},~ \tconf{
      \epsilon
  }{b_2}{e_2}}
  \\
  &
  C \vdash {!\eta_{d}}.\eta_1 \dual{} {?\eta_{r}}.\eta_2,~ \eta_1 \dual{} \eta_2
  \qquad
  p,\co p\freshfrom p',\Delta_1, \c S_0
  \qquad
  p'\freshfrom \Delta_1, \c S_0
  \end{align*}
  where $b_1$, $e_1$, $b_2$ and $e_2$ are the appropriate behaviours and expressions shown above.
  Therefore, $ C \vdash \subt{\eta_d'}{\eta_r}$ and from Lem.\ \ref{lem:wst-subtype} we deduce
  \begin{align*}
    &
    \wst{\c S_0,~ \tconf{
      \Delta_1
    }{b_1}{e_1},~ \tconf{
      (p'^{l_1'}:\eta_{d}')\cdot\epsilon
    }{b_2}{e_2}}
  \\\text{and~~}&
     \wst{\c S_0,~
       \tconf{
         (p^{l_1}:\eta_1)\cdot\Delta_1
       }{b_1}{e_1},~ \tconf{
         (\co p^{l_2}:\eta_2) \cdot (p'^{l_r}:\eta_{r})
       }{b_2}{e_2}}
  \end{align*}
  which completes the proof for this case. The rest of the cases are similarly proven.
\end{proof}

\begin{lemma}
  Suppose $\tjParSys{}{\c S}$ and $S \arrow \vect{e'}$. There exists $\c S'$ such that $\c S \tred \c S'$ and $S'=\vect{e}$.
\end{lemma}
\begin{proof}
  By the definitions of the reduction relations $(\tred)$ and $(\arrow)$ in \cref{fig:lang,fig:typed-reductions}. In this proof
  the structure of behaviours needed for establishing the $(\tred)$ reductions are deduced using Lem.\ \ref{lem:type-decomp} and
  inversion on the typing rules; the necessary structure of the stacks is deduced by inversion on the rules of the $\Downarrow_C$
  relation (\cref{fig:abstract-interpr}).
\end{proof}

\emph{
The proof of preservation (\cref{thm:preservation}) is a direct consequence of the preceding two lemmas.
}

\subsection{Type Soundness} 

We first extend the notion of internal and communication steps to the reductions of \cref{fig:typed-reductions}.
\begin{lemma}
  Let
  $\c S \tred \c S'$. 
  \begin{itemize}
    \item $\c S \tredCom \c S'$ if the transition is derived with use of the rules \iref{trinit}, \iref{trcom}, \iref{trdel}, \iref{trsel};
    \item $\c S \tredInt \c S'$ otherwise.
  \end{itemize}
\end{lemma}
A diverging process is one that has an infinite sequence of internal transitions $(\intStep)$.
\begin{definition}[Divergence]
  A process $P$ diverges if $P \tredInt \c S_1 \tredInt \c S_2 \tredInt \c S_3 \tredInt \ldots$.
  A system $\c S$ diverges if for any $P \in \c S$, $P$ diverges.
\end{definition}
These transitions may spawn new processes, and divergence can be the result of an infinite spawn chain.
\begin{example} Let
    $P \defeq \tconf{\epsilon}{b}{\eapp{(\efix{f}{x}{\espawn{f}})}{\cunit}}$
  where
  $b \defeq \bseq{\tau}{\bseq{\tau}{\orec{(
          \bseq{\tau}{\espawn{\beta}}
    )}{\beta}}}$; $P$ is a diverging process.
\end{example}
To prove progress we first divide a system into its diverging and non-diverging parts. The non-diverging part of the system can
only take a finite number of internal transitions.
\begin{lemma} 
  Let $\tjParSys{}{\c S}$ and $\wst{\c S}$. Then $\c S = \c D, \c{ND}$ and $\c D$ diverges and for some $\c{ND}'$,
  $\c{ND}\tredInt[C\,*]{}\c{ND}'\not\tredInt$.
\end{lemma}
\begin{proof}
  By definition of diverging system (and its negation).
\end{proof} 

The non-diverging part of the system that cannot take any more internal steps consists of processes that are values, or stuck on global
channels or session primitives.

\begin{lemma} 
  Let $\tjParSys{}{\c{ND}}$ and $\wst{\c{ND}}$ and $\c{ND}\not\tredInt$. Then $\c{ND} = \c F, \c W, \c B$ and:
 \begin{description}
   \item[Processes in $\c F$ are finished:] $ \forall\tconf{\Delta}{b}{e} \in \c F.~ \Delta=\epsilon, b=\tau, e=v$.
   \item[Processes in $\c W$ wait on global channels:] $ \forall\tconf{\Delta}{b}{e} \in \c W.~ e = E[\erequest{l}{c}] $ or
     $ e = E[\eaccept{l}{c}]$.
   \item[Processes in $\c B$ block on a session primitive:]
     $\forall P=\tconf{\Delta}{b}{e} \in \c B.~ e = E[e_0]$ and $e_0 \in \{\esend{v}{}, \erecv{}{}~v, \edeleg{v}, \eresume{}{}~v,
       \eselectnew{L}~v, \ematch{v}{L_i \Rightarrow e_i}{i \in I}\}$.
 \end{description}
\end{lemma}
\begin{proof}
  If a process in $\c{ND}$ is not in one of the three categories then it would be able to take an internal step. Moreover, the
  structure of the stack $\Delta$ and behaviour $b$ of finished processes follows from $\tjParSys{}{\c{ND}}$ and $\wst{\c{ND}}$.
\end{proof} 

From the preceding two lemmas we can easily derive the most part of progress.

\begin{corollary} \label{cor:progress-a}
 Let $ \tjParSys{}{\c S}$ and $ \wst{\c S}$. Then $ \c S \intStep^* (\c F,\c D,\c W,\c B) $ such that:
 \begin{description}
   \item[Processes in $\c F$ are finished:] $ \forall\tconf{\Delta}{b}{e} \in \c F.~ \Delta=\epsilon, b=\tau, e=v$.
   \item[Processes in $\c D$ diverge:] $ \forall\tconf{\Delta}{b}{e} \in \c D.~ \tconf{\Delta}{b}{e} \intStep^\infty$.
   \item[Processes in $\c W$ wait on global channels:] $ \forall\tconf{\Delta}{b}{e} \in \c W.~ e = E[\erequest{l}{c}] $ or
     $ e = E[\eaccept{l}{c}]$.
   \item[Processes in $\c B$ block on a session primitive:]
     $\forall P=\tconf{\Delta}{b}{e} \in \c B.~ e = E[e_0]$ and $e_0 \in \{\esend{v}{}, \erecv{}{}~v, \edeleg{v}, \eresume{}{}~v,
       \eselectnew{L}~v, \ematch{v}{L_i \Rightarrow e_i}{i \in I}\}$.
 \end{description}
\end{corollary}

What is missing is that when $\c S$ cannot take any more communication steps $(\Rightarrow_{\mathsf c})$ then all processes in $\c
B$ depend on processes in $\c D$ and $\c W$. This follows by well-stackedness of (well-typed) systems. We write
$(\c F, \c D, \c W, \c B)$ for a system whose finished processes are in $\c F$, diverging processes are in $\c D$, waiting
processes are in $\c W$ and blocked processes are in $\c B$.

\begin{lemma} \label{lem:progress-b}
  Let $\tjParSys{}{\c S}$ and $\wst{\c S}$ and $\c S = (\c F, \c D, \c W, \c B)$ and $\c W, \c B \not\intStep$. Then
  \begin{enumerate}
    \item \label{l23.1} If $P, Q \in \c W$ and $P \leftrightharpoons Q$ then $P, Q \comStep \c S'$, for some $\c S'$.
    \item \label{l23.2} The $(\mapsto)$ dependencies in $\c S$ create a directed acyclic graph.
    \item \label{l23.3} If $P \in \c B$ then there exist $Q, R \in \c D, \c W, \c B$ such that $P \Mapsto (Q,R)$.
  \end{enumerate}
\end{lemma}
\begin{proof}
  The first property follows from $\tjParSys{}{\c S}$ and $\wst{\c S}$.
  Property~\ref{l23.2} is proven by induction on $\wst{\c S}$.

  Property~\ref{l23.3}: Because $P \in \c B$ and $P$ is well-typed, the stack of $P$ is non-empty. Thus, by Property~\ref{l23.1},
  there exists $Q \in \c S$ such that $P \mapsto^+ Q$ is the longest sequence of dependencies \emph{without repetitions} (this is
  possible because of Property~\ref{l23.2}).
  We examine two cases: 
  \begin{itemize}
    \item $P \mapsto^* P' \mapsto Q$ and the top-level frame in the stack of $P'$ has an endpoint $p$ and $\co p$ appears in the top frame
      of $Q$; then $P \leftrightharpoons Q$; therefore $P \Mapsto(P',Q)$.
    \item $P = P_1 \mapsto \ldots \mapsto P_n \mapsto Q$ and the top-level frame in the stack of $P'$ has an endpoint $p$ and $\co p$ appears in a
      frame other than the top one in $Q$; then there exists $R$ such that $Q \mapsto R$ (by $\wst{\c S}$); $R$ cannot be one of
      the processes $P_1,\ldots, P_n$ because of Property~\ref{l23.2}. Moreover $R$ cannot be a process in $\c F$ (because
      processes in $\c F$ have empty stacks due to typing), and $R$ cannot be any other processes in $\c D, \c W, \c B$ because
      the sequence of dependencies is the longest. Thus this case is not possible.
      \qedhere
  \end{itemize}
\end{proof}

\emph{Type Soundness is a direct consequence of Cor.~\ref{cor:progress-a} and Lem.~\ref{lem:progress-b}.}


\clearpage
\label{sec:appendix}
\clearpage
\section{Inference algorithms}
\label{sec:algorithms}
\newcommand{\nif}{\mathsf{\textcolor{blue}{if~}}}
\newcommand{\nand}{\mathsf{\textcolor{blue}{and~}}}
\newcommand{\nfresh}{\mathsf{ ~fresh }}

The algorithm to infer session types at second level of our type system is
called Algorithm $ \algoSI $, and it operates on a syntactical sub-language of
session types called \emph{session structures}, which is defined as follows:
\begin{align*}
\eta ::= & \psi
          \bnf \tend 
          \bnf !T.\eta 
          \bnf ?T.\eta
          \bnf !\eta.\eta
          \bnf ?\eta.\eta
          \bnf \psiint
          \bnf \psiext
\end{align*}
Internal and external choices are removed from the syntax of sessions. In their
place, we have two special variables $ \psiint $ and $ \psiext $.
These variables are bound by special \emph{choice constraints}, which
extend the syntax of Constraints in Fig. \ref{fig:type-syntax} as follows:
\begin{align*}
  C ::= &  
  \ldots \bnf 
  \ceq{\sichoice{i}{I}{\eta}}{\psiint}
          \bnf \ceq{\sechoice{i}{I_1}{I_2}{\eta}}{\psiext}
\end{align*}
During session inference, the constraint set $ C $ might be refined to a new
set $ C' $, containing more precise session types for $ \psiint $ and 
$ \psiext $ (for example by adding new labels to an internal choice, or by
moving an active label to inactive in an external choice), 
or new constraints on types (because of a $ \rho!\tint$ behaviour for example).
Choice constraints in $ C $ are ordered according to subtyping:
\begin{center}
\irule[Inf-IChoice][]
  { \sueq{\psiint}{\eta} \quad \coType{\eta}{\eta'} }
  { \sueq{\psiint}{\eta'} }
\\
\irule[Inf-EChoice][]
  { \sueq{\psiext}{\eta} \quad \coType{\eta}{\eta'} }
  { \sueq{\psiext}{\eta'} }
\end{center}
Abstract interpretation transitions can be naturally extended to 
the sub-language just presented, except for the two cases
when $ b $ pops a label $ L_k $, and the top of the stack contains either $
\psiint $ or $ \psiext $. In these two cases, if $ \sueq{\psiext}{\eta} $ or $
\sueq{\psiext} {\eta} $, $ \eta $ substitutes $ \psiext $ or $
\psiint $ on the stack. 

Section \ref{sec:algoSI} presents the main inference algorithm for Stage 2,
whose main core is Algorithm $ \algoMC $ in Section
\ref{sec:algoMC}.
Sections \ref{sec:algoFunc} and \ref{sec:algoSub} contains the 
auxiliary functions for session sub-type inference. 
Section \ref{sec:algoD} shows the inference algorithm for Stage 3. 

As stated in Sec. \ref{sec:inference}, the definition of Algorithm $\algW$ can
be adapted straightforwardly from \cite{ANN}, and therefore its definition is
not included. The soundness of Algorithm $ \algW $ can be
stated as follows:
\begin{theorem}[Soundness of Algorithm $ \algW $]
If $ \algW ([], e) = (\sigma, t, b, C) $ then $\tjrNorm{C}{[]}{e}{t}{b}$.
\end{theorem}

Regarding completeness, let $ jdg^\star $ be any valid typing judgement for an
expression $ e $.
Completeness first show that Algorithm $ \algW $ always
calculates a judgment $ jdg $ for $ e $. Moreover it also shows that $jdg^\star
$ is a \emph{lazy instance} of $ jdg $: in the sense that there exists a
substitution $ \sigma' $ such that $ jdg^\star $ can always be derived by
further instantiating $ jdg $ with $ \sigma' $ and by subtyping. 
This
second property points to the fact that $ \algW $ calculates \emph{principal
types} for $ e $ (see \cite{ANN}, Sec. 1.5.1, p.30). 
Completeness is stated as follows:

\begin{theorem}[Completeness of Algorithm $ \algW $]
If $\tjrNormAt{C^\star}{[]}{e}{t^\star}{b^\star} $ with $ C^\star $ atomic
(i.e. all type constraints in $ C^\star $ have the form $
\su{\alpha_1}{\alpha_2}$), then $ \algW([], e) = (\sigma, b, t, C)$ and there
exists a substitution $ \sigma^\star $ such that:
\begin{itemize}
  \item $ \cj{C^\star}{C\sigma^\star} $
  \item $ \subBase {C^\star} {b\sigma^\star} {b^\star} $
  \item $ \coTypeBase{C^\star} {t\sigma^\star} {t^\star} $
\end{itemize}
\end{theorem}

\subsection{Algorithm $ \algoSI $}
Let $ K $ be the behavior stack, defined by the following grammar:
\begin{align*}
K ::= \emptyK \bnf \consK{b}{K}
\end{align*}

Let also the application of $ K $ to $ b $, or $ K[b] $, be defined
inductively as follows:
\begin{align*}
    \emptyK[b]       = & b
\\  (\consK{b'}{K})[b] = & K[b;b']
\end{align*}

Since we only work with finite behaviors $ b $, there is always a finite
decomposition $ b = K[b'] $ such that $ b' $ is not the sequential composition
of two sub-behaviors; the decomposition is also unique.
The session inference algorithm $ \algoSI $ is defined as follows:
\label{sec:algoSI}
\begin{lstlisting}[  backgroundcolor=\color{white},  escapeinside={(*}{*)} ]
  $ \algoSI (b,~ C) = (\sigma_2\sigma_1,~ C_2) $
    if  $ (\sigma_1, C_1) = \mcclause \epsilon b C \emptyK $
    and $ (\sigma_2, C_2) = $ choiceVarSubst $ C_1 $

  choiceVarSubst $ (C \uplus \overrightarrow{\singleSu{\eta_i}{\psiint_i}}
  \uplus \overrightarrow{\singleSu{\eta_j}{\psiext_j}}) =
    (\sigma, C\sigma) $
    if   $ \sigma =\overrightarrow{\subst{\psiint_i}{\eta_i}}
                   \overrightarrow{\subst{\psiext_j}{\eta_j}} $
    and  $ \psiint, \psiext \freshfrom RHS(C) \text{ for any } \psiint, \psiext
    $
\end{lstlisting}

\subsection{Algorithm $ \algoMC $}
\label{sec:algoMC}
\begin{lstlisting}[  backgroundcolor=\color{white},  escapeinside={(*}{*)} ]
-- remove terminated frames
$ 
  \mcclause {\St l \tend} b C K  
  = \mcclause {\Delta} b C K
$ (* \label{code:MC-delta-end} *)

-- $ \textcolor{mygreen}{\algoMC}$ terminates with behavior $\textcolor{mygreen}{\tau }$ 
$ 
  \mcclause \Delta \tau C \emptyK    
  = (\sigma,~ C\sigma) 
$ (* \label{code:MC-delta-tau} *)
  if $  \sigma = $ finalize $ \Delta $

-- pop a sub-behavior from the continuation stack 
$ 
  \mcclause \Delta \tau C {\consK b K}    
  = \mcclause \Delta b C K 
$ (* \label{code:MC-delta-ktau} *)

-- push a new frame on the stack
$ 
  \mcclause \Delta {\pusho{l}{\eta}{}} C K 
  = (\sigma_2\sigma_1, C_2)
$ (* \label{code:MC-push} *)
  if  ($\sigma_1, \Delta_1$) = checkFresh$ (l, \Delta) $
  and ($\sigma_2, C_2$) = $
    \mcclause {\stBase {l}{\eta\sigma_1}{\Delta_1}} \tau {C\sigma_1}{K\sigma_1}
  $
      
-- send
$
  \mcclause { \St l \psi } {\popo{\rho}{!T}} C K
    = 
  (\sigma_2\sigma_1, C_2)
$  (*\label{code:MC-send-psi}*)
  if  $ \subEql{l}{\rho} $
  and $ \sigma_1 = \subst\psi {!\alpha.\psi'} $ where $ \alpha, \psi' $ fresh
  and $ (\sigma_2, C_2) 
          = 
        \mcclause { \St l {\psi'}\sigma_1 } \tau {
          C\sigma_1 \cup \singleSu{T}{\alpha} } { K\sigma_1 } 
      $

$ 
  \mcclause {\St l {!\alpha.\eta}}{\popo{\rho}{!T}} C K
    = \mcclause { \St l \eta } \tau {C \cup \singleSu{T}{\alpha}} K
$ (*\label{code:MC-send-eta}*)
  if  $ \subEql{l}{\rho} $ 

-- recv
$ 
  \mcclause { \St l \psi }{\popo{\rho}{?T}} C K
    = (\sigma_2\sigma_1,~ C_2) $(*\label{code:MC-recv-psi}*)
  if  $ \subEql{l}{\rho} $
  and $ \sigma_1 = \subst\psi {?\alpha.\psi'} $ where $ \alpha, \psi' $ fresh 
  and $ (\sigma_2, C_2) = \mcclause {\St l {\psi'}\sigma_1} \tau {
                  C\sigma_1\cup\singleSu{\alpha}{T} } {K\sigma_1} 
      $

$
  \mcclause {\St l{?\alpha.\eta}} {\popo{\rho}{?T}} C K
    = \mcclause {\St l {\eta}} \tau {C \cup \singleSu{\alpha}{T}} K
$(*\label{code:MC-recv-eta}*)
  if $ \subEql{l}{\rho} 
$

-- delegation
$
 \mcclause {\stBase{l}{\psi}{\St{l_d}{\eta_d}}} {\popo{\rho}{!\rho_d}} C K
= (\sigma_2\sigma_1, C_2)
$(*\label{code:MC-del-psi}*)
  if $ \subEql{l}{\rho} $ and $ \subEql{l_d}{\rho_d} $
  and $ \sigma_1 = \subst \psi {!\eta_d.\psi'} $ where $ \psi' $ fresh
  and $ (\sigma_2, C_2) = \mcclause {\St l {\psi'} \sigma} \tau {C\sigma_1
                            } {K\sigma_1} 
  $

$
  \mcclause {\stBase l {!\eta_d.\eta} {\St{l_d}{\eta_d'}}
  } {\popo{\rho}{!\rho_d}} C K
  = (\sigma_2\sigma_1, C_2)
$(*\label{code:MC-del-eta}*)
  if  $ \subEql{l}{\rho} $ and $ \subEql{l_d}{\rho_d} $
  and $ (\sigma_1, C_1) = \nsub(\eta_d', \eta_d, C) $
  and $ (\sigma_2, C_2) = 
    \mcclause {\St l {\eta} \sigma} \tau {C_1} {K\sigma_1} 
$

$  \mcclause {\stBase l {\eta} {\St{l_d}{\eta_d'}}
  }{\popo{\rho}{!\rho_d}} C K
  = (\sigma_2\sigma_1, C_2)
$(*\label{code:MC-del-checkFresh}*)
  if  $ \subNotEql{l_d}{\rho_d} $
  and $ (\sigma_1, \Delta_1) = $ checkFresh$ (l_d, \stBase l \eta
                                                        {\St{l_d}{\eta_d'}}) $
  and $ (\sigma_2, C_2) = \mcclause  {\Delta_1}{\popo{\rho}{!\rho_d}
    } { C\sigma_1} {K\sigma_1} $

-- resume
$ 
  \mcclause {\stBase l \psi \epsilon}{\popo{\rho}{?l_r}} C K
    =  (\sigma_2\sigma_1, C_2) 
$(*\label{code:MC-res-psi}*)
  if  $ \subEql{l}{\rho} $ and $ l \neq l_r $
  and $ \sigma_1 = \subst\psi {?\psi_r'.\psi'} $ where $ \psi', \psi_r' $ fresh
  and $ (\sigma_2, C_2) =  
  \mcclause {\stBase l {\psi'} { \stBase {l_r} {\psi_r'} \epsilon}
            } \tau {C\sigma_1
            }      {K\sigma_1}$
  
$ 
  \mcclause {\stBase l {?\eta_r.\eta} \epsilon} {\popo{\rho}{?l_r}} C K
    = 
    \mcclause {\stBase l {\eta} {\stBase{l_r}{\eta_r}\epsilon}} \tau C K
$(*\label{code:MC-res-eta}*)
  if $ \subEql{l}{\rho} $ and $ l \neq l_r $

$ 
  \mcclause {\stBase l {?\eta} {\St {l_d}{\eta_d}}} {\popo{\rho}{?l_r}} C K = 
  (\sigma_2\sigma_1, C_2)
$(*\label{code:MC-res-checkFresh}*)
  if  $ (\sigma_1, \Delta_1) = $ checkFresh$(l_d, \stBase l {?\eta} 
  {\St {l_d}{\eta_d}} )$ 
  and $ (\sigma_2, C_2) =
    \mcclause {\Delta_1} {\popo{\rho}{?l_r}} {C\sigma_1} {K\sigma_1} $

-- in. choice
$
  \mcclause {\St{l}{\psi}} {\popo{\rho}{!L_k}} C K
  =
    (\sigma_2\sigma_1, C_2) 
$ (*\label{code:MC-ich-psi}*)
  if  $ \subEql{l}{\rho} $
  and $ \sigma_1 = \subst\psi \psiint $ where $ \psi_k, \psiint $ fresh
  and $ (\sigma_2, C_2) = 
    \mcclause {\St l {\psi_k}\sigma} \tau {
            C\sigma_1 \cup \atomSu{\ceq{\psiint}{\sichoice{j}{\{k\}}{\psi}}}
    } {K\sigma_1}
  $

$
\mcclause {\St{l}{\psiint}} {\popo{\rho}{!L_i}} C K
  = \mcclause {\St l {\eta_i}} \tau C K
$(*\label{code:MC-ich-psiint}*)
  if  $ \subEql{l}{\rho} $
  and $ i \in J $ and $ (\ceq{\psiint}{\sichoice{j}{J}{\eta}}) \in C $ 

$ 
  \mcclause {\St{l}{\psiint}} {\popo{\rho}{!L_i}} C K
    = 
  \mcclause {\St l {\psi_i}} \tau {
    C' \uplus \atomSu{\ceq{\psiint}{\sichoice{j}{J,i}{\eta}}}
  } K
$(*\label{code:MC-ich-psiint-not}*) 
  if  $ \subEql{l}{\rho} $
  and $ i \not \in J $ and $
    C = C' \uplus \atomSu{\ceq{\psiint}{\sichoice{j}{J}{\eta}}}  
    $ and $ \eta_i = \psi_i $ fresh 
  
-- ex. choice
$
  \mcclause {\St{l}{\psi}} {\bechoice{}{\rho}{L_i}{b_i}} C K
    = (\sigma_2\sigma_1, C_2) 
$ (*\label{code:MC-psiext-psi}*)
  if  $ \subEql{l}{\rho} $
  and $ \sigma_1 = \subst{\psi}{\psiext} $ where $ \psiext $ fresh
  and $ C_1 = \atomSu{\ceq{\psiext}{\sechoice{i}{I}{\emptyset}{\psi}}} 
    $ where $ \overrightarrow {\psi_i} $ fresh
  and $ (\sigma_2, C_2) = 
      \mcclause {\St{l}{\psiext}\sigma_1
                } {\big(\bechoice{}{\rho}{L_i}{b_i}\big)\sigma_1
                } {C\sigma_1\cup C_1
                } {K\sigma_1} $
                
$
  \mcclause {\St{l}{\psiext}
            } {\bechoiceTrap{j \in I_1\uplus I_2\uplus I_3}{\rho}{L_j}{b_j}
            } C K 
    = \mcclause {\St{l}{\psiext}
                } {\bechoiceTrap{j \in I_1\uplus I_2\uplus I_3}{\rho}{L_j}{b_j}
                } {C_1
                } {K}
$ (*\label{code:MC-psiext-not}*)
  if  $ \subEql{l}{\rho} $ 
  and $ C = C' \uplus \atomSu{ 
                          \ceq \psiext {\sechoice{j}{I_1J_1}{I_2J_2}{\eta}} }
      $ and $ I_3 \freshfrom J_1\uplus J_2 $ and $ I_3, J_1 \neq \emptyset $
  and $ C_1 = C' \cup \atomSu{\ceq {\psiext}
                                   {\sechoice{i}{I_1}{I_2I_3J_1J_2}{\eta}} } $

$ \mcclause {\St{l}{\psiext}
            } {\bechoiceTrap{i\in I}{\rho}{L_i}{b_i}} C K 
    = (\sigma_n\ldots\sigma_0, C_n) $(*\label{code:MC-psiext}*)
  if  $ \subEql{l}{\rho} $ and $ J_1 \subseteq I $ and $ I \subseteq J_1 \uplus
                                                                        J_2 $ 
  and $ C = C' \uplus \atomSu{\ceq{\psiext}{\sechoice{j}{J_1}{J_2}{\eta}}} $
  and for each $ k \in [1\ldots n] $ where $ I = \{L_{j_1}, \ldots, L_{j_n}\} 
  $ and $ n \geq 0 $
    $ (\sigma_k, C_k) = \mcclause {(\St l {\eta_{j_k}}
                                      } {b_{j_k})\sigma_{k-1}\ldots\sigma_0
                                      } {C_{k-1}
                                      } {K\sigma_{k-1}\ldots\sigma_0} $ 
      where $ \sigma_0 = \id $ and $ C_0 = C $
          
-- sequencing
$ \mcclause \Delta {\bseq {b_1} {b_2}} C K
  = 
  \mcclause \Delta {b_1} C {\consK {b_2} K}
  $ (*\label{code:MC-seq}*) 
      
-- internal choice in the behavior      
$ \mcclause  \Delta {b_1\oplus b_2} C K
    = (\sigma_2\sigma_1, C_2) 
    $ (*\label{code:MC-plus}*)
  if  $ (\sigma_1, C_1) = \mcclause \Delta {b_1} C K $
  and $ (\sigma_2, C_2) = \mcclause {\Delta\sigma_1 
                                  } {b_2\sigma_1
                                  } {C_1
                                  } {K\sigma_1}
      $

-- spawn
$ \mcclause \Delta {\spawno{b}{}} C K
    = (\sigma_2\sigma_1, C_2) $(*\label{code:MC-spawn}*)
  if  $ (\sigma_1, C_1) = \mcclause \epsilon b C \epsilon $
  and $ (\sigma_2, C_2) = \mcclause {\Delta\sigma_1} \tau {C_1} {K\sigma_1} $ 

-- rec 
$ \algoMC(\mcconf{\Delta}{\orec{b}{\beta}},~ C)
   =  (\sigma_2\sigma_1, C_2) $ (*\label{code:MC-rec}*)
  if  $ C = C' \uplus \singleSu{b'}{\beta} $
  and $ (\sigma_1, C_1) = \mcclause \epsilon b {C' \cup
                                      \singleSu{\tau}{\beta}} \epsilon $ 
  and $ (\sigma_2, C_2) = \mcclause {\Delta\sigma_1
                                  } \tau {
 (C_1 \backslash \singleSu{\tau}{\beta}) \cup
 (\singleSu{b'}{\beta})\sigma_1 } { K\sigma_1 } $
  
-- behavior variable
$ \mcclause \Delta \beta C K
  = \mcclause \Delta b C K $(*\label{code:MC-beta}*)
  where  $ b = \bigoplus \{b_i \where \exists i. ~(b_i \subseteq \beta) \in C\}
  $ 
 
-- try to close the top session if all previous clauses fail
$ \mcclause \Delta b C K
  = (\sigma_2\sigma_1, C_2) 
$ (*\label{code:MC-closetop}*)
  if  $ (\sigma_1, \Delta_1) = $ closeTop$ (\Delta) $
  and $ (\sigma_2, C_2) = \mcclause { \Delta_1
                                  } { b\sigma_1
                                  } { C\sigma_1
                                  } { K\sigma_1} $
\end{lstlisting}

\subsection{Helper functions} \label{sec:algoFunc}
\begin{lstlisting}[backgroundcolor=\color{white}]
-- checkFresh forces frame l in the input stack to be the closed session 
checkFresh$(l, \epsilon) = (\id, \epsilon) $
checkFresh$(l, \St {l'} \tend) = (\id, \Delta) $
checkFresh$(l, \St l \psi) = (\subst \psi \tend, \Delta\subst \psi \tend) $
checkFresh$(l, \St {l'} \eta) = (\sigma, \stBase {l'} {\eta\sigma} {\Delta'}) $ 
  if $ l' \neq l $ and $ (\sigma, \Delta') $ = checkFresh$ (l, \Delta) $

-- $\color{mygreen}{\mathsf{closeTop}}$ forces the top of the stack to be $ \color{mygreen}{\tend} $ 
closeTop$(\St{l}{\eta})$ = checkFresh$ (l, \St{l}{\eta}) $

-- $\color{mygreen}{\mathsf{finalize}}$ matches the input stack with the empty stack 
finalize $ \epsilon = \id $
finalize $ \St{l}{\tend} = $ finalize $\Delta$
finalize $ \St{l}{\psi}  = \sigma_2\sigma_1 $
  if  $ \sigma_1 = \subst{\psi}{\tend} $
  and $ \sigma_2 = \mathsf{finalize}(\Delta) $
\end{lstlisting}
  
\subsection{Subtype checking}
\label{sec:algoSub}
\begin{lstlisting}[backgroundcolor=\color{white}]
-- end
sub$(\tend, \tend, C) = (\id, C) $

-- send/recv
sub$(!T_1.\eta_1, !T_2.\eta_2, C) = (\sigma_1, C_1 \cup
  \singleSu{T_2}{T_1}) $
  if $(\sigma_1, C_1) = $ sub$ (\eta_1, \eta_2, C) $

sub$(?T_1.\eta_1, ?T_2.\eta_2, C) = (\sigma_1, C_1\cup
    \singleSu{T_1}{T_2})  $
  if $(\sigma_1, C_1) = $ sub$ (\eta_1, \eta_2, C) $

-- deleg/resume
sub$(!\eta_{d1}.\eta_1, !\eta_{d2}.\eta_2, C)
            $ = $(\sigma_2\sigma_1, C_2) $
  if  $ (\sigma_1, C_1) = $ sub$(\eta_{d2}, \eta_{d1}, C)
       $ and $ (\sigma_2, C_2) = $ sub$(\eta_1\sigma_1, \eta_2\sigma_1,
       C_1) $

sub$(?\eta_{r1}.\eta_1, ?\eta_{r2}.\eta_2, C)
            $ = $(\sigma_2\sigma_1, C_2) $
  if  $ (\sigma_1, C_1) = $ sub$(\eta_{r1}, \eta_{r2}, C)
       $ and $ (\sigma_2, C_2) = $ sub$(\eta_1\sigma_1, \eta_2\sigma_1,
       C_1)
       $

-- in. choice
sub$(\psiint_1, \psiint_2, C) = $ f $( I_2, C )$
  if $ (\ceq{\psiint_2}{\sichoiceBase{i}{I_2}{\eta}{2i}})  \in C $ 
  and f $ (\emptyset, C) = (\id, C) $
  and f $ (I\uplus \{k\}, C) $ = f $ (I\uplus \{k\}, C_1) $
        if  $ C =  C' \uplus (\ceq{\psiint_1}{\sichoiceBase{i}{I_1}{\eta}{1i}})
                      \uplus (\ceq{\psiint_2}{\sichoiceBase{i}{I_2}{\eta}{2i}})$
        and $ k \not \in I_1 $ and $ k \in I_2 $
        and $ C_1 =C' \uplus (\ceq{\psiint_1}{\sichoiceBase{i}{I_1
                                \cup \{k\}
                                  }{\eta}{1i}}\oplus \eta_{2_k})
                      \uplus (\ceq{\psiint_2}{\sichoiceBase{i}{I_2}{\eta}{2i}}) 
        $
  and f $ (I\uplus \{k\}, C) = (\sigma_2\sigma_1, C_2)$
        if  $ (\sigma_1, C_1) = $ f $ (I, C) $
        and $ (\ceq{\psiint_1}{\sichoiceBase{i}{I_1}{\eta}{1i}}),
              (\ceq{\psiint_2}{\sichoiceBase{i}{I_2}{\eta}{2i}})  \in C_1 $
        and $ k\in I_1 $ and $ k \in I_2 $
        and $ (\sigma_2, C_2) = $ sub $ 
                                    (\eta_{1k}\sigma_1, \eta_{2k}\sigma_1, C_1)$
                                    
-- ex. choice
sub$(\psiext_1, \psiext_2, C) = $ f $ (J_1 \cup J_2, C) $
  if  $ (\ceq{\psiext_1}{\sechoiceBase{i}{I_1}{I_2}{\eta}{1i}}),
        (\ceq{\psiext_2}{\sechoiceBase{i}{J_1}{J_2}{\eta}{2i}})  \in C $
  and $ I_1 \subseteq J_1 $
  and f $ (\emptyset, C) = (\id, C) $
  and f $ (I\cup \{k\}, C) = (\sigma_2\sigma_1, C_2)$
        if  $ (\ceq{\psiext_1}{\sechoiceBase{i}{I_1}{I_2}{\eta}{1i}}),
              (\ceq{\psiext_2}{\sechoiceBase{i}{J_1}{J_2}{\eta}{2i}})  \in C $
        and ($ k \in I_2 $ or ($k \in I_1 $ and $ k \in J_1 $)  
        and $ (\sigma_1, C_1) = $ f $ (I, C) $
        and $ (\sigma_2, C_2) = $ sub $ 
                                  (\eta_{1k}\sigma_1, \eta_{2k}\sigma_1, C_1) $
  and f $ (I\cup \{k\}, C) = $ f $ (I\cup \{k\}, C_1)$
        if  $ C = C' \uplus
                        (\ceq{\psiext_1}{\sechoiceBase{i}{I_1}{I_2}{\eta}{1i}})
        \uplus (\ceq{\psiext_2}{\sechoiceBase{i}{J_1}{J_2}{\eta}{2i}}) $
        and $ k \in I_1 $ and $ k \in J_2 $  
        and $ C_1 = C' 
        \uplus (\ceq{\psiext_1}{\sechoiceBase{i}{I_1\backslash\{k\}
        }{I_2 \cup \{k\} }{\eta}{1i}})
        \uplus (\ceq{\psiext_2}{\sechoiceBase{i}{J_1}{J_2}{\eta}{2i}}) $ 
  and f $ (I\cup \{k\}, C) = $ f $ (I\cup \{k\}, C_1)$
        if  $ C = C' \uplus
                        (\ceq{\psiext_1}{\sechoiceBase{i}{I_1}{I_2}{\eta}{1i}})
        \uplus (\ceq{\psiext_2}{\sechoiceBase{i}{J_1}{J_2}{\eta}{2i}}) $
        and $ k \not \in I_1 $ and $ k \not \in I_2 $  
        and $ C_1 = C' 
        \uplus (\ceq{\psiext_1}{\sechoiceBase{i}{I_1
        }{I_2 \cup \{k\} }{\eta}{1i}} + ?L_k.\eta_{2k})
        \uplus (\ceq{\psiext_2}{\sechoiceBase{i}{J_1}{J_2}{\eta}{2i}}) $ 
        
-- session inference
sub$(\psi, \eta, C) = (\sigma, C\sigma)$ if $ \sigma = \subst\psi\eta $
sub$(\eta, \psi, C) = (\sigma, C\sigma)$ if $ \sigma = \subst\psi\eta $
\end{lstlisting}

\subsection{Algorithm $ \algoDual $}
\label{sec:algoD}

The duality check algorithm, which we call Algorithm $ \algoDual $, takes the
constraints set $ C $ calculated by the second stage, and returns a larger
constraints set $ C' $. As in Nielson\& Nielson, the algorithm simply halts
when a duality check fails, rather than throwing an exception.

Algorithm $\algoDual$ manipulates a new kind of constraints, called
\emph{duality constraints}. A duality constraint has the form $ (\eta_1 \dual{} \eta_2)
$, where $ \eta_1 $ and $ \eta_2 $ are inference session
types. At the beginning, Algorithm $\DC$ creates a duality
constraint $ \eta_1 \dual{} \eta_2 $ for each channel $ c $ and $ \bar c $
such that session types $ \eta_1 $ and $ \eta_2 $ have been derived,
i.e. such that $ \atomSu{\su{\eta_1}{c}, \su{\eta_2}{\bar c} } \in C $.
For any other channel $ c $ such that $ \singleSu{\eta_1}{c}$ is in $ C $,
but no constraint $ \singleSu{\eta_2}{\bar c} $ is in $ C $, Algorithm $ \DC
$ introduces a constraint $ \eta_1 \dual{} \psi_2$, where $ \psi_2 $ is a
fresh variable. 

After this initial setup, Algorithm $ \DC $ non-deterministically applies one
of the following rules to the configuration $(\id, C)$, until no more rules can
be applied:
$$
\begin{array}{llll}
    (\sigma, C \uplus \singleDualSu{\tend}{\tend})  
  & \hookrightarrow
  & (\sigma, C) 
  &
\\  (\sigma, C \uplus \singleDualSu{!T_1.\eta_1}{?T_2.\eta_2})  
  & \hookrightarrow
  & (\sigma, C \cup \atomSu{\su {T_1 }{T_2}, \duals{\eta_1}{\eta_2} }) 
  &
\\  (\sigma, C \uplus \singleDualSu{!\eta_1'.\eta_1}{?\eta_2'.\eta_2})  
  & \hookrightarrow
  & (\sigma'\sigma, C' \cup \atomSu{\duals{\eta_1\sigma'}{\eta_2\sigma'} }) 
  & \text{if } (\sigma', C') = \mathsf{sub}(C, \eta_1, \eta_2)
\\  (\sigma, C \uplus 
    \singleDualSu{\sichoiceBase{i}{I_0}{\eta}{1i}}
                 {\sechoiceBase{i}{I_1}{I_2}{\eta}{2i}})  
  & \hookrightarrow
  & (\sigma, C \cup \underset{i \in I_0}\bigcup
                                    \atomSu{\duals{\eta_{1i}}{\eta_{2i}} }) 
  & \text{if } I_0 \subseteq I_1
%
\\  (\sigma, C \cup 
    \singleDualSu{\psi_1}
                 {\eta_2})  
  & \hookrightarrow
  & (\sigma'\sigma, C')
  & \text{if } (\sigma', C') = \mathsf{expand}(C, \duals {\psi_1} {\eta_2})
\\  (\sigma, C \cup 
    \singleDualSu{\psiint}
                 {\eta_2})  
  & \hookrightarrow
  & (\sigma'\sigma, C')
  & \text{if } (\sigma', C') = \mathsf{expand}(C, \duals {\psiint} {\eta_2})
\\  (\sigma, C \cup 
    \singleDualSu{\psiext}
                 {\eta_2})  
  & \hookrightarrow
  & (\sigma'\sigma, C')
  & \text{if } (\sigma', C') = \mathsf{expand}(C, \duals {\psiext} {\eta_2})
\end{array}
$$
where the helper function $ \mathsf{sub} $ is the same function from
Algorithm $ \algoSI$ (which returns a substitution and a set of constraints
such that the two input sessions are in the subtyping relation). The helper
function \textsf{expand} is defined as follows:

\begin{longtable}{lll}
   $\mathsf{expand}(C, \duals {\psi_1} \tend  )$ 
     & $= (\sigma, C\sigma)$
     & \text{if } $\sigma = \subst {\psi_1} \tend$
\\ $\mathsf{expand}(C, \duals {\psi_1} {?T.\eta_2} )$ 
     & $= (\sigma, C\sigma)$
     & $\text{if } \sigma = \subst {\psi_1} {!\alpha.\eta_2}$
\\ $\mathsf{expand}(C, \duals {\psi_1} {!T.\eta_2} )$ 
     & $= (\sigma, C\sigma)$
     & $\text{if } \sigma = \subst {\psi_1} {?\alpha.\eta_2}$
\\ $\mathsf{expand}(C, \duals {\psi_1} {!\eta_2'.\eta_2} ) $
     & $= (\sigma, C\sigma)$
     & $\text{if } \sigma = \subst {\psi_1} {!\psi_1'.\psi_1''}$ 
\\ $\mathsf{expand}(C, \duals {\psi_1} {?\eta_2'.\eta_2} ) $
     & $= (\sigma, C\sigma)$
     & $\text{if } \sigma = \subst {\psi_1} {?\psi_1'.\psi_1''}$
\\ $\mathsf{expand}(C, \duals {\psi_1} {\sichoiceBase{i}{I}{\eta}{2i}} )$ 
     & $= (\sigma, C\sigma \cup C')$
     & $\text{if } \sigma = \subst {\psi_1} {\psiext}
       \text{ and } C' = \atomSu  
        { \ceq \psiext {\sechoiceBase{i}{I}{\emptyset}{\psi}{2i}} }$
\\ $\mathsf{expand}(C, \duals {\psi_1} {\sechoiceBase{i}{I_1}{I_2}{\eta}{2i}} )$ 
    & $ = (\sigma, C\sigma \cup C')$
    & $\text{if }  \sigma = \subst {\psi_1} {\psiint}
      \text{ and } C'   = \atomSu {
      \ceq \psiint {\sichoiceBase{i}{I_1}{\psi}{2i}}}$
\\ $\mathsf{expand}(C, \duals \psiext {\sichoiceBase{i}{I_0}{\eta}{2i}})$  
  & $= (\id, C')$
  & $\text{if }~~~ C = C'' \cup \atomSu 
    { \ceq \psiext \sechoiceBase{i}{I_1}{I_2}{\eta}{1i} }$
\\&& $\text{and } C' = C'' \cup \atomSu 
{
  \ceq \psiext {\sechoiceBase {i} {I_1\cup I_0} {I_2\backslash I_0} {\eta} {1i}}
}
                                      \cup \underset{i\in I_0}\bigcup
                                      \singleDualSu{\eta_{1i}}{\eta_{2i}}$ 
\\ $\mathsf{expand}(C, \duals \psiint {\sechoiceBase{i}{I_1}{I_2}{\eta}{2i}})$  
  & $= (\id, C')$
  & $\text{if }~~~~ C = C'' \cup \atomSu { \ceq\psiint
                                          { \sichoiceBase{i}{I_0}{\eta}{1i} } }$
\\&& $\text{and }  C' = C'' \cup \atomSu {
  \ceq \psiext {\sichoiceBase {i} {(I_0\backslash I_2)\cap I_1} {\eta} {1i}}}
                                      \cup \underset{i\in I_0\cap I_1}
                                      \bigcup
                                      \singleDualSu{\eta_{1i}}{\eta_{2i}}$
\\ $\mathsf{expand}(C, \duals \psiint \psiext)$  
  & $= (\id, C')$
  & $\text{if } ~~~~~C = C'' \cup \singleSu {\sichoiceBase{i}{I_0}{\eta}{1i}}
                                      \psiint
                       \cup \singleSu {\sechoiceBase{i}{I_1}{I_2}{\eta}{1i}}
                                      \psiext$
\\ && $\text{ and } C' = C'' 
 \cup \atomSu { \ceq \psiint {\sichoiceBase{i}{ (I_0\backslash I_2) \cup
                                                              I_1}{\eta}{1i} }} 
 \cup \atomSu { \ceq\psiext
  {\sechoiceBase{i}{ I_1 \cap (I_0\backslash I_2) }
                                 {I_2 \cup (I_0\backslash I_1}{\eta}{1i}}}$           
\\ && $\phantom{ and  C' = C''~~ }\cup \underset{i\in (I_0\backslash I_2)\cap
                                      I_1} \bigcup
                                      \singleDualSu{\eta_{1i}}{\eta_{2i}}$
\end{longtable}

where all $ \alpha, \psi, \psiint, \psiext $ variables on the right-hand side
are fresh, and where fresh variables $ \psi_{1i} $ and $ \psi_{2i} $ are
generated, in case index $ i $ is not defined in the starting internal or
external choice. Symmetric rules are omitted.

The duality constraints are increasingly simplified, until no more
simplifications are possible. After Algorithm $\DC$ is finished, the input set
$ C $ fails the duality check when there exists a duality constraint 
$ \suDual{\eta_1}{\eta_2} $ such that neither $ \eta_1 $ or $
\eta_2$ are fresh session type variables $ \psi $, i.e. $\eta_1 \neq
\psi_1$ and $ \eta_2 \neq \psi_2 $.

\clearpage
\section{Finiteness of Abstract Interpretation}
\label{sec:termination}
\subsection{Introduction}
This section proves that the abstract interpretation of a configuration $ \mc
\Delta b $ in a well-formed environment $ C $ always generates a finite
state-space. We first formalize the notion of \emph{behaviour compact} from
Definition \ref{def:wf-constraints}. Then we define a translation from
behaviours with $ \beta $ variables to \emph{ground behaviours}, i.e. behaviours
without $ \beta $s. We show that this translation is fully abstract with
reference to the abstract interpretation semantics. Finally, we show that
a configuration $ \mc \Delta b $ and constraints $ C $ generate a finite
state-space when $ b $ is ground and $ C $ is well-formed. 

\subsection{Finite ground behaviors}
In order to show termination of $ \algoMC $, we show that $ \beta $ variables
can always be represented by an equivalent \emph{finite} behaviour when $ C $ is
well-formed.

Consider the following ordering:
\begin{definition}[Behaviour ordering]\label{def:fin-b-ordering}
For any constraints set $ C $ and behaviour $ b $, the \emph{behaviour ordering}
$ (C, b) \succ (C', b')$ is defined by the following equations:
\begin{align*}
   (C, b_1;b_2) 
&  \succ (C, b_i) 
&  \text{ for } i \in \{1, 2\}
\\ (C, b_1 \oplus b_2)
&  \succ (C, b_i)
&  \text{ for } i \in \{1, 2\}
\\ (C, \espawn b)
& \succ (C, b)
&
\\ (C, \bechoice{}{\rho}{L_i}{b_i})
&  \succ (C, b_i)
&  \text{ for } i \in I
\\ (C, \beta) 
&  \succ (C, b )
&  \text{ if } \singleSu b \beta \in C
\\ (C \uplus \singleSu {b'} \beta, \orec b \beta) 
&  \succ (C \cup \singleSu \tau \beta, b ) 
& 
\end{align*}
\end{definition}

Notice that a behaviour $ b $ structurally decreases on the
right-hand side of $ \succ $, except when $ b $ is a behaviour variable $
\beta$. When $ C $ is not well-formed, a behaviour variable might occur
infinitely often in a chain of $ \succ $ relations. Such is the case when $ b
= \beta $ and $ C = \{ \su \beta \beta \}$, since the pair $ (C, \beta) $ gives
rise to the infinite chain: 
\begin{align*}
(C, \beta) \succ (C, \beta) \succ (C, \beta) \succ \ldots
\end{align*}

However, it can be shown that there is no infinite chain when $ C $ is
well-formed. The reason for this is that in a well-formed $ C $, any cycle on
behaviour constraints must include a recursive constraint of the form
($\su{\orec{b}{\beta}}{\beta'} $). The definition of $ \succ $ replaces such
constraints with the dummy constraint ($ \su \tau \beta $), effectively
breaking the constraint cycle. Since a well-formed $ C $ is a finite set,
eventually there are no more cycles in $ C $. 

This property is crucial to demonstrate that behaviours have a finite
representation in a well-formed $ C $, and it can be proved as follows:

\begin{lemma}[Well-foundedness of $ \succ $]\label{lem:fin-well-founded}
If $ C $ is well-defined, then $ \succ $ is well-founded.
\end{lemma}
\begin{proof}
We need to show that there are no infinite descending chains in $ \succ $, i.e.
for any pair $ (C, b) $ there are no infinite sequences of the following kind:
\begin{align*}
(C, b) \succ (C_1, b_1) \succ (C_2, b_2) \succ \ldots \succ (C_i, b_i) \succ
\ldots
\end{align*}

We prove the lemma by contradiction.
Suppose that there is indeed such an infinite descending chain in $ \succ $.
Let $ \mathcal I = \set{i}{\exists \beta. b_i = \beta}$
There are two cases to consider: either $\mathcal I$ is finite, i.e. 
$ \beta $ variables occur a finite number of times in the infinite chain,
or $ \mathcal I $ is infinite, i.e. $ \beta $ variables occur infinitely often.

If $\mathcal I $ is finite, let $ k $ be the least element in $ \mathcal I $.
By definition, $ b_k $ is the last behaviour in the infinite chain such that
$ b_k = \beta_k $ for some $ \beta_k $.
By well-formedness of $ C $, $ C $ contains the constraint $ (b_{k+1} \subseteq
\beta_k) $ and $ b_{k+1} $ is a finite, well-defined term. Since $ b_{k+1} $ is
finite and $ b_j $ is not a variable for $ j > k $, then the chain must be
finite, because $ \succ $ structurally decreases $ b $ at every step. 
This contradicts the assumption that there is an infinite chain in $ \succ
$, and the lemma is proved.

If $ \mathcal I $ is infinite, then the set 
$ \mathcal B = \set{b_i}{\exists\beta. b_i = \beta}$
(the set of all behaviour variables occurring in the chain) must be finite,
because there are only a finite number of constraints $ (\su b \beta ) $ in a 
well-formed $ C $. Because this set is finite, by the pigeonhole principle
there exist $ b_i $ and $ b_j $ in the infinite chain such that $ b_i = b_j =
\beta $ for some $ \beta $. This can be illustrated as follows:
\begin{align*}
(C, b) \succ \ldots \succ (C_i, \beta) \succ \ldots \succ (C_k, b_k)
\succ (C_{k+1}, b_{k+1}) \succ \ldots \succ (C_j, \beta) \succ \ldots
\end{align*}

This sequence contains a cycle over $ \beta $:
\begin{align*}
 (C_i, \beta) \succ \ldots \succ (C_k, b_k)
\succ (C_{k+1}, b_{k+1}) \succ \ldots \succ (C_j, \beta)
\end{align*}
By the Behaviour-Compact property of
well-formedness (Def. \ref{def:wf-constraints}.3), one of the behaviours in
the cycle must be recursive, i.e. it must have the form $ \orec
{b'}{\beta'} $ for some $ b' $ and $ \beta' $.
Let $ b_k = \orec {b'} {\beta'} $; by definition of $ \succ $ $ b_{k+1} =
b'$. Moreover, $ C_{k}$ must contain the constraint $ (\su
{\orec{b'}{\beta'}} {\beta'}) $, whereas $ C_{k+1}$ contains 
the dummy constraint $ \su \tau {\beta'} $ in its place.
Because of this, the number of recursive constraints $ \su {\orec b \beta }
\beta $ decrease by one element in $ C_{k+1} $. 

Since $ C $ is a finite set, the
infinite chain in $ \succ $ must contain a constraints set $ C' $ such
that $ C' $ contains no recursive constraints anymore. 
Since all cycles in behaviours constraints must contain at least a recursive
behaviour, and since $ C' $ contains no such constraints, then there can be no
infinite chain after $ C' $. This contradicts the hypothesis that there is an
infinite chain from $ (C, b) $, and the lemma is proved.

\end{proof}

\subsection{Behaviour variables elimination }
The occurrence of a variable $ \beta $ in a behaviour $  b $ creates an indirect
link between $ b $ and the constraints $ C $ where $ \beta $ is defined.
This hidden connections introduces cumbersome technical complications when
proving properties of the abstract interpretation semantics. On the contrary, 
\emph{ground behaviours}, i.e. behaviours that do not
contain $ \beta $ variables, are easier to reason about. This section
introduces a translation from any behaviour $ b $ to the ground behaviour $
\groundify b $, and shows that it is fully abstract w.r.t the
abstract operational semantics (provided that $ C $ is well-defined).

The abstract interpretation semantics treats $ \beta $ variables as
place-holders: Rule \RefTirName{ICh}
replaces a $ \beta $ with
any behaviour $ b $ to which $ \beta $ is bound in $ C $; Rule \RefTirName{Rec}
effectively replaces each $ \beta $ variables inside recursive
behaviours $ \orec b \beta $ with $ \tau $.

These observations suggest that a $ \beta $ variable can be substituted either
with the internal choice of all the behaviour it binds in $ C $, or with a $
\tau $ inside recursive behaviours. Such a translation is defined as follows: 

\begin{definition}[Ground translation]\label{def:fin-groundify}
Let $ C $ be a well-formed constraint set.
The \emph{ground translation} $ \groundify{-} :: \behavDomain \rightarrow
\behavDomain $ is the total function defined by the following equations:\\

\begin{tabular}{lll}
 $\groundify {b}$ 
   &= $ b$ 
   & $\text{ if } \beta \freshfrom b \text{ for any } \beta$
\\ $ \groundify {b_1;b_2} $
   &= $\groundify {b_1};\groundify {b_2}$
   & 
\\ $\groundify {b_1\oplus b_2}$ 
   &= $\groundify {b_1}\oplus\groundify {b_2}$
   &
\\ $\groundify {\spawno{b}{}} $
   & = $\spawno{\groundify b}{}$
   &
\\ $\groundify {\bechoicetext{i\in I}{\rho}{L_i}{b_i}} $
   &= $\bechoicetext{i\in I}{\rho}{L_i}{\groundify {b_i}}$
\\ $\groundify \beta $
   &= $\bigoplus\set {\groundify {b_i}}{\su{b_i}\beta\in C}$
   & 
\\ $\groundifyBase {\orec b \beta } {C \uplus \singleSu{b'}{\beta}}$
   &= $\orec {\groundifyBase b {C\cup \singleSu \tau \beta }} \beta$
   &
\end{tabular}
\end{definition}

We now show that, when a constraint set $ C $ is well-formed, the ground
translation of a behaviour $ b $  in $ C $ 
does not expand $ \beta $ variables infinitely, but it 
constructs a \emph{finite} ground behaviour, i.e. a behaviour with a finite
syntax tree:

\begin{lemma}
Let $ C $ be well-formed and $ b $ be a finite behaviour. For any behaviour $ b
$, $ \groundify b $ is a finite ground term.
\end{lemma}
\begin{proof}
By well-founded induction on $ \succ $.

The base case is when $ b $ is a ground term in $ \{\tau, \pusho{l}{\eta}{},
\popo{\rho}{!L} \ldots \}$. These are all ground terms in $ G $, and for these
terms the translation $ \groundify{b} = b $, which is finite and ground by
hypothesis.

If $ b \in\{ b_1;b_2, b_1\oplus b_2, \spawno{b_1}{}, 
\groundify {\bechoicetext{i\in I}{\rho}{L_i}{b_i}} \}$, then the lemma is proved
by the inductive hypothesis, since for example if $ b = b_1;b_2$, then $
\groundify{b_1} $ and $ \groundify{b_2} $ are finite ground terms, and
therefore $ \groundify{b_1};\groundify{b_2}$ is finite and ground too.

Because of well-formedness, there are two cases to consider when  $ b = \beta $:
either $ \beta $ is bound to a unique constraint $ b \neq \orec {b'} \beta $ in
$ C $, or it is bound to multiple $ b_i $ which are not recursive behaviours.
In the case that $ \su b \beta  $ is the only constraint on $ \beta $ in $ C
$, and we can write $ C $ ad $ C' \uplus \singleSu b \beta$.
By definition of translation we have $ \groundify \beta = \groundifyBase
{\orec{b'} \beta} {C' \uplus \singleSu b \beta } =  \orec{\groundifyBase
{b'} {C' \uplus \singleSu b \tau }} \beta$. Since 
$ (C \uplus \singleSu {\orec {b'} \beta} \beta, \orec {b'} \beta) \succ (C \cup
\singleSu \tau \beta, b' )$ holds by definition of $ \succ $, the lemma is
proved by inductive hypothesis. 
In the latter case, when $ \beta $ is bound to multiple non-recursive
behaviours $ b_i $, 
the set of all such $ b_i$ is finite by well-formedness, and 
the lemma is proved by the inductive hypothesis as in the
case $ b_1;b_2$.
\end{proof}

Having proved that the ground translation of a behaviour $ b $ always exists for
well-formed constraints $ C $, we show some property of the translation w.r.t.
the abstract semantics:
\begin{lemma}
 Let $ C $ be well-formed. 
 \begin{enumerate}
  \item if $ \dstep{\Delta}{b}{\Delta}{b'} $, then $
  \dstep{\Delta}{\groundify b}{\Delta'}{\groundify{b'}}$
   \item If $ \dstep{\Delta}{\groundify b}{\Delta'}{b''} $, then there exists a
   $ b' $ such that $ b'' = \groundify{b'}$.
 \end{enumerate}
\end{lemma}
\begin{proof}
By rule induction. 
\end{proof}

The termination of Algorithm $ \algoMC$ hinges on the finiteness of its input
behaviour. Even though the algorithm expands the input constraints set $ C $
with additional type and session constraints, the size of a ground term remains
constant:
 \begin{lemma}[Constraint expansion]\label{lem:fin-c-expand}
For any well-formed $ C $, if $ C' $ does not contain behavior constraints,
then $ \groundify{b} = \groundifyBase{b}{C \cup C'}$
for any $ b $.
\end{lemma}
\begin{proof}
By structural induction on $ b $. 
\end{proof}


\subsection{Finite state-space}
We conclude this section by showing that, given a well-formed $ C $,
all configurations $ \mc \Delta b $ always generate a finite state-space, i.e.
the set of reachable states from $ \mc \Delta b $ is finite.
We prove this result by designing a function that assigns an integer,
or \emph{size}, to any configuration $ \mc \Delta b $, and then show that the
size of a configuration always decreases after taking a step in the abstract
interpretation semantics. Since configurations of size 0 cannot take steps, and
since the size decreases after taking a step in the semantics, the number of
states that a finite configuration $ \mc \Delta b $ can reach is finite.

We first introduce  the size function on behaviours:
\begin{definition}[Behaviour size]\label{def:fin-behav-size}
For any behaviour $ b $, the behaviour size $ \bsize - :: \behavDomain
\rightarrow \mathcal N$ is the total function defined by the following
equations:
\[
\begin{array}{@{}l@{~}l@{\qquad}l@{~}l@{}}
   \bsize \tau &= 0 
 & \bsize \beta&= 0
\\
   \bsize {\pusho{l}{\eta}{}} &= 2
 & \bsize {\popo{\rho}{?\rho'}} &= 2
\\ 
   \bsize {\popo{\rho}{!T}} &= 1
 & \bsize {\popo{\rho}{?T}} &= 1
\\ \bsize {\popo{\rho}{!\rho'}} &= 1
 & \bsize {\popo{\rho}{!L_i}} &= 1
\\ \bsize {\bechoice{i\in I}{\rho}{L_i}{b_i}}  &= 1 + \sum_{i\in I}
     \bsize{b_i} 
 & \bsize{\orec b \beta } &= 1 + \bsize b
\\ 
\bsize{b_1;b_2 } &= 1 + \bsize {b_1} + \bsize {b_2}
 & \bsize{\spawno{b}{} }&=2 + \bsize{b}
\\ 
 \bsize{b_1\oplus b_2 } &= 1 + \bsize {b_1} + \bsize {b_2}
 &&
\end{array}
\]

\end{definition}
According to the definition, $ \tau $ is the behaviour with the smallest size,
zero. Most $ \kpop $ operations have size 1, except for the resume operation $
\popo{\rho}{?\rho'} $, which has size 2. Notice that $ \kpush $ has size 2 as
well. The reason for this difference is that these operations introduce new
frames on the stack in the abstract interpretation semantics, and therefore
have to be counted twice in order for the abstract interpretation semantics to
be always decreasing in size.
The size of the other behaviours is defined inductively.

We now introduce the size of a stack:
\begin{definition}[Stack size]\label{def:stack-size}
The size of a stack $ \Delta $, or $ \bsize{\Delta} $, is defined by the
following equations:
\begin{align*}
  \bsize{\epsilon } = 0 & &\bsize{\St l \eta} = 1 + \bsize \Delta
\end{align*}
\end{definition}

In short, the size of a stack is its length, or total number of frames.
We finally specify the size of configurations:
\begin{definition}[Configuration size]\label{def:conf-size}
The size of a configuration $ \mc \Delta b $, or $ \bsize {\mc \Delta b }$, is
the sum $ \bsize {\mc \Delta b } =   1 + \bsize \Delta + \bsize b $.
\end{definition}

An important property of $ \bsize - $ is that it is invariant to session types:
\begin{lemma}[Session substitution distributivity]\label{lem:fin-subst} 
For any session substitution $ \sigma $, $
\bsize{\mc \Delta b}= \bsize{\mc{\Delta\sigma} {b\sigma}}$. 
\end{lemma}
\begin{proof}
By structural induction on $ \mc \Delta b$. 
\end{proof}

The size of a configuration decreases strictly after each step in the
abstract interpretation semantics. This result provides a
useful induction principle to reason about Algorithm $ \algoMC $, which will be
used to prove a completeness result for $ \algoMC $. The following lemma
expresses this result:
\begin{lemma}  \label{lem:term-step-size-shrinks}
Let $ C $ be well-formed, and let $ b $ be a ground finite behaviour. If 
$ \dstep{\Delta}{b}{\Delta'}{b'} $, then $ \bsize {\mc \Delta b} > \bsize {\mc
{\Delta'} {b'}}$.
\end{lemma}
\begin{proof}
By rule induction. 

If Rule \RefTirName{End} is applied, then $     \dstep {\St l \tend } {b}
{\Delta} {b}$. By definition of size, $ \bsize {\mc {\St l \tend} b} =
\bsize{\St l \tend} + \bsize b = 1 + \bsize{\Delta} + \bsize{b} = 1 +
\bsize{\mc \Delta b} $, which proves the lemma.

By hypothesis $ b $ is a ground term, therefore Rule \RefTirName{Beta} cannot be 
applied.

If Rule  \RefTirName{Push} is applied, then $ 
\dstep {\Delta}           {\pusho{l}{\eta}{} }
          {\St {l}{\eta}}    {\tau}
$ holds. By definition, $ \bsize {\mc \Delta \pusho{l}{\eta}{ }} = 2 +
\bsize{\Delta} = 1 + \bsize{\St l \eta} = 1 + \bsize{\mc {\St l \eta} \tau} $,
which proves the lemma. The case for Rule \RefTirName{Res} is proved similarly.

If Rule \RefTirName{Out} is applied, then $
\dstep {\St {l}{!T.\eta}} {\popo{\rho}{!T'}} {\St {l}{\eta}} {\tau} $
holds. By definition, $ \bsize {\mc{\St {l}{!T.\eta}} {\popo{\rho}{!T'} }} =
\bsize{\St {l}{!T.\eta}} + 1 = 1 + \bsize{\Delta } + 1 = \bsize{\St l \eta} + 1
= \bsize{\mc {\St l \eta}{\tau}} $, which proves the lemma. The cases for Rule
\RefTirName{In} and \RefTirName{Del} are proved similarly.

If Rule \RefTirName{ICh} is applied, then $ \dstep {\Delta} {b_1\oplus b_2 } 
{\Delta}{b_i} $ holds for $ i \in \{1, 2\}$. By definition of size, then 
 $ \bsize {\mc \Delta {b_1\oplus b_2}} = \bsize{\Delta} + 1 + \bsize{b_1} +
 \bsize {b_2 } = 1 + \bsize{\mc \Delta {b_i}} + \bsize{b_j } $ for $ \{i, j\} =
 \{1,2\}$. The lemma is proved by $ 1 + \bsize{\mc \Delta {b_i}} + \bsize{b_j }
 > \bsize{\mc \Delta {b_i}} $ for $ i \in \{1,2\}$.
The cases for Rule \RefTirName{ECh}, \RefTirName{Rec}, \RefTirName{Spn} and
\RefTirName{Tau} are proved similarly.

If Rule \RefTirName{Seq} is applied, then $
\dstep {\Delta}          {b_1;b_2 
            } {\Delta'}{b_1';b_2}$ only if 
             $ \dstep {\Delta} {b_1}{\Delta'}{b_1'} $ holds.
By rule induction $ \bsize {\mc {\Delta} {b_1}} > \bsize{ \mc {\Delta'}{b_1'}}
$. By definition of size, $ \bsize {\mc {\Delta} {b_1;b_2}} = 1 + \bsize{\Delta}
+ \bsize {b_1} + \bsize{b_2} $;  by the previous inequality we have that 
$1  + \bsize{\Delta} + \bsize {b_1} + \bsize{b_2} >  
 1 + \bsize{\Delta'} + \bsize {b_1'} + \bsize{b_2} 
 = \bsize{\mc {\Delta'}{b_1';b_2}}
$, which proves the lemma.
\end{proof}

\subsection{Termination of $ \algoMC$}

In order to prove termination, some well-formedness properties of algorithm $
\algoMC$ need to be proved first. In particular, we want to show that $ \algoMC
$ returns well-formed outputs, given some well-formed arguments as input. 

A configuration $ \mc\Delta b$ is well-formed in configuration $ C $ when any
free variable occurring in $ \mc\Delta b$ is bound by some constraint in $ C $:

\begin{definition}[Variable well-formedness]
A variable $ \gamma $ is \emph{well-formed} in a constraints set $ C $ when, 
if $ \gamma \in \{\beta, \rho, \psi, \psiint, \psiext \} $, then 
there exists a constraint $ (g \subseteq \gamma) \in C $ such that all the free
variables in $ g $ are well-formed in $ C $.
\end{definition}

\begin{definition}[Configuration well-formedness]
A configuration $ \mcconf{\Delta}{b}$  is \emph{well-formed} 
%
in a constraint set $ C $ when all the free variables in 
$ \mcconf \Delta b $ are well-formed in $ C $.
\end{definition}

Given a well-formed constraints set $ C $ and a well-formed configuration $ \mc
\Delta {K[b]} $ in $ C $, Algorithm $ \algoMC $ returns a \emph{session
inference substitution} $ \sigma $ and a \emph{refined} constraints set $ C' $.

Following the approach of \cite[Sec. 2.2.5, p.51]{ANN}, session
inference substitution (or simply session substitution) is a total function from
session variables to session types, defined as follows:
\begin{definition}[Session inference substitution]
An \emph{inference substitution} $ \sigma $ is a total function from
session variables $ \psi $ to sessions $ \eta $.
The \emph{domain} of an inference substitution $ \sigma $
is $ dom(\sigma) = \set{\psi}{\sigma(\psi) \neq\psi } $ and its range is 
$ rg(\sigma) = \bigcup\set{FV(\sigma(\psi))}{\psi \in dom(\sigma)} $.
\end{definition}

It is easy to verify that Algorithm $ \algoMC $ only returns session
substitutions:
\begin{lemma}
If $ \mc\Delta {K[b]}$ is well-formed in $ C $ and $ \mcclause \Delta b k C =
(\sigma, C') $, then $ \sigma $ is a session inference substitution.
\end{lemma}
\begin{proof}
By induction on $ \bsize{\mc\Delta {K[b]}}$. 
\end{proof}

During session type inference, Algorithm $ \algoMC $ might need to subsume
some type $ T $ inside an inferred session type to a more general type $ T' $.
For example, suppose that $ \algoMC$ has inferred session
type $ \eta = !T.\tend $ (for some $ T $) and a set $ C $ for an endpoint $ l $.
If the same endpoint is used in the behaviour $ b = \popo{l}{!T'}$, $ b$
respects session type $ \eta $ only if $ \subType {T'} T $ holds, according to
the abstract interpretation semantics of Fig. \ref{fig:abstract-interpr}.
Algorithm $ \algoMC $ deals with type subsumption by adding new type
constraints in the input set $ C $ in the constraints set $ C' $; for example,
it would add the constraint $ (\su {T'} T) $ in the previous example.
We say that $ C' $ is a \emph{refinement} of $ C $, in the sense $ C' $
is a superset of $ C $ that further specifies types and session types.

Constraints refinement is defined as follows:

\begin{definition}[Constraints refinement]\label{def:constr-ref}
A constraint set $ C $ is a refinement of constraint set $ C' $, or  
$ \cj{C}{C'} $, when:
\begin{itemize}
  \item for all constraints $ (\ceq{\psiint}{\eta'})$ in $ C' $, 
        there exists a session $ \eta $ such that $ \coType{\eta}{\eta'} $ and $
        \subEql \psiint \eta $ hold. 
  \item for all constraints $ (\ceq\psiext {\eta'})$ in $ C' $, 
        there exists session $ \eta $ such that $ \coType{\eta}{\eta'} $ and $
        \subEql \psiext \eta $ hold.
  \item if $ \subBase{C'}{\pusho{l}{\eta'}{}} \beta $, then
        $ \sub{\pusho{l}{\eta}{}} \beta $ and
        $ \coType{\eta}{\eta'} $
  \item for all other constraints $(\su{g}{g'} )$ in $ C' $, $ \sub{g}{g'} $
  holds
\end{itemize}
\end{definition}


Constraint refinements are unaffected by session substitutions:

\begin{lemma}[Substitution invariance]        \label{lem:sound-C-to-Csigma}                   
Let $ C $ and $ C ' $ be well-formed, and $ \sigma $ a session inference
substitution. If $ \cj{C}{C'}$, then $ \cj{C\sigma}{C'\sigma}$.
\end{lemma}
\begin{proof}
By induction on the size of $ C $.
\end{proof}

Algorithm $ \algoMC $ progressively refines the input constraints set $ C $
by refining session variables $\psi$, assigning a lower session type to external
and internal choices bound to variables $ \psiext $ and $ \psiint $, and by
adding type constraints $ \su T {T'} $.
When the algorithm terminates, the resulting set $ C' $ is a refinement of $ C
$, after applying the inferred substitution $ \sigma $ to it. 

\begin{lemma}[Constraint refinement on sub-types]        
\label{lem:sound-C-sub-ext} 
If $ \mathsf{sub}(\eta_1, \eta_2, C) = (\sigma_1, C_1) $, 
then $ \cj{C_1}{C\sigma_1}$.
\end{lemma}
\begin{proof}
By structural induction on $\eta_1$.

Function \textsf{sub} terminates when either $\eta_2$ has the same shape as $
\eta_1$ (i.e. if $ \eta_1 = !T_1.\eta_1'$ then $ \eta_2 = !T_2.\eta_2'$), or $
\eta_2 $ is a variable $ \psi_2 $. In the latter case $ \mathsf(\eta_1,
\psi_2, C) = (\sigma, C\sigma) $ with $ \sigma = \subst {\psi_2}{\eta_1} $, and
the lemma is trivially proved by $ \cj{C\sigma}{C\sigma}$.

Consider the former case, when $ \eta_1 $ and $ \eta_2 $ have the same shape.
When $ \eta_1 = \tend $ the lemma holds trivially. If $ \eta_1 $ is $
!T_1.\eta_1' $, then $ \eta_2 = !T_2.\eta_2' $; by induction we have 
$ \cj{C_1}{C\sigma_1 \cup \singleSu{T_2\sigma_1}{T_1\sigma_1} }$, which implies
$ \cj{C_1}{C\sigma_1} $ by definition of constraint refinement. The case for $
\eta_1 = ?T_1.\eta_1'$ is proved similarly.

The cases for delegate and resume hold directly by inductive hypothesis.

When $ \eta_1 $ is an internal choice, we need to show that $ f (I_2, C) =
(\sigma_1, C_1) $ implies $ \cj{C_1}{C\sigma_1} $. This can be easily proved by
induction of the size of $I_2$. 
The base case (line 21) is trivial. In the inductive case,
if a label $ k $ from $ I_2 $ is missing in $ I_1 $ (lines 22-25), then
$ \sichoicetextBase{i}{I_1 \cup \{k\} }{\eta}{1i}\oplus \eta_{2_k} $
is a subtype of 
$ \sichoicetextBase{i}{I_1  }{\eta}{1i} $ by definition (because $ I_1 \cup
\{k\} \subset I_1 $) and the lemma follows by inductive hypothesis. 
If $ k $ is in $ I_1 $, then the lemma follows directly by inductive hypothesis.

When $ \eta_1 $ is an external choice, we need to show that $ f (J_1 \cup J_2,
C) = (\sigma_1, C_1) $ implies $ \cj{C_1}{C\sigma_1} $. We prove this by
induction on the size of $ J_1 \cup J_2 $. 
Let $ \eta_1 = \sechoicetextBase{i}{I_1}{I_2}{\eta}{1i} $ and 
    $ \eta_2 = \sechoicetextBase{i}{J_1}{J_2}{\eta}{2i} $ be such that $ I_1
    \subseteq J_1 $.
The base case (line 36) is trivial.
In the inductive case, suppose that a label $ k $ is either in the inactive
labels $ J_2 $, or it is in both the active labels $ I_1 $ and $ J_1 $. Then 
it is sufficient to show that 
$ \eta_{1k} $ is a subtype of $ \eta_{2k} $, which holds by inductive hypothesis
on \textsf{sub}. If $ k $ is in the active labels $ I_1 $ but it
is in the inactive labels $ J_2 $, 
then removing $ k $ from $ I_1 $ and adding it
to the inactive labels $ I_2 $ makes $ I_1 \backslash \{k\} $ be a subset of $
J_1 $; under this condition the subsequent call $ \mathsf{f}(I \cup \{k\}, C_1)
$ is proved as in the previous case.
If $ k $ is neither in $ I_1 $ nor $ I_2 $, then $ k $ is added to the inactive
labels $ I_2 $ and the lemma is proved as in the previous case too.

\end{proof}

\begin{lemma}[Constraint refinement]         \label{lem:sound-C-ext}
If $ \algoMC(\mcconf{\Delta}{b}, C, K) = (\sigma_1, C_1) $, 
then $ \cj{C_1}{C\sigma_1}$.
\end{lemma}
\begin{proof}
By induction on $ \bsize{\mcconf {\Delta\sigma} {
\groundifyBase{(K[b])\sigma}{C_1}}} $.

Most cases follow directly from the inductive hypothesis, such as the case for
$ \kpush $:
\begin{lstlisting}[
    backgroundcolor=\color{white}, 
    escapeinside={(*}{*)},
    firstnumber=11,
    firstline=20,
    lastline=28]
-- push a new frame on the stack
$ 
  \mcclause \Delta {\pusho{l}{\eta}{}} C K 
  = (\sigma_2\sigma_1, C_2)
$ (* \label{code:MC-push} *)
  if  ($\sigma_1, \Delta_1$) = checkFresh$ (l, \Delta) $
  and ($\sigma_2, C_2$) = $
    \mcclause {\stBase {l}{\eta\sigma_1}{\Delta_1}} \tau {C\sigma_1}{K\sigma_1}
  $
\end{lstlisting}
At line 14 we have that $ \mcclause {\stBase {l}{\eta\sigma_1}{\Delta_1}}
\tau {C\sigma_1}{K\sigma_1} = (\sigma_2, C_2) $. By inductive hypothesis, we
obtain directly $  \cj {C_2}{C\sigma_2\sigma_1} $.

The lemma is also trivial when $ b $ is an operation that sends
a type $ T $:
\begin{lstlisting}[
    backgroundcolor=\color{white}, 
    escapeinside={(*}{*)},
    firstnumber=16,
    firstline=30,
    lastline=42]
-- send
$
  \mcclause { \St l \psi } {\popo{\rho}{!T}} C K
    = 
  (\sigma_2\sigma_1, C_2)
$  (*\label{code:MC-send-psi}*)
  if  $ \subEql{l}{\rho} $
  and $ \sigma_1 = \subst\psi {!\alpha.\psi'} $ where $ \alpha, \psi' $ fresh
  and $ (\sigma_2, C_2) 
          = 
        \mcclause { \St l {\psi'}\sigma_1 } \tau {
          C\sigma_1 \cup \singleSu{T}{\alpha} } { K\sigma_1 } 
      $
\end{lstlisting}
At line 20 we have $ \mcclause { \St l {\psi'}\sigma_1 } \tau {
C\sigma_1 \cup \singleSu{T}{\alpha} } { K\sigma_1 }  = (\sigma_2, C_2) $. By
inductive hypothesis we obtain directly that $  \cj {C_2}{C\sigma_2\sigma_1
\cup \singleSu{T\sigma_2}{\alpha\sigma_2}
} $, which in turn implies  
$ \cj {C_2}{C\sigma_2\sigma_1} $
by definition of constraints refinement; this proves the lemma. 
The lemma is proved similarly when a type $ T $ is received (lines 25-32), and
when a session is resumed (lines 50-61). When 
the behaviour delegates a session (lines 34-48), the lemma is proved by Lem.
\ref{lem:sound-C-sub-ext}.

The case is more interesting when the behaviour $ b $ extends an inferred
internal choice with a new label:
%
%
\begin{lstlisting}[
    backgroundcolor=\color{white}, 
    escapeinside={(*}{*)},
    firstnumber=73,
    firstline=148,
    lastline=159]
$ 
  \mcclause {\St{l}{\psiint}} {\popo{\rho}{!L_i}} C K
    = 
  \mcclause {\St l {\psi_i}} \tau {
    C' \uplus \atomSu{\ceq{\psiint}{\sichoice{j}{J,i}{\eta}}}
  } K
$(*\label{code:MC-ich-psiint-not}*) 
  if  $ \subEql{l}{\rho} $
  and $ i \not \in J $ and $
    C = C' \uplus \atomSu{\ceq{\psiint}{\sichoice{j}{J}{\eta}}}  
    $ and $ \eta_i = \psi_i $ fresh 
  
\end{lstlisting}
%
By definition of $ \algoMC $ at lines \ref{code:MC-ich-psiint-not} 
and 75, we have:
\begin{align*}
  C    
    &= C' \uplus \atomSu{
\ceq{\psiint}{\eta_1}}
    &
    \eta_1  &= \sichoice{j}{J}{\eta}
\\ C_1 
&= C' \uplus
\atomSu{\ceq{\psiint}{\eta_1'}}
&
\eta_1' &= \sichoice{j}{J,i}{\eta}
\end{align*} 
where $ i \not\in J $.
Notice that $ \suType{\eta_1'}{\eta_1} $ holds by definition of subtyping,
since $ J \subseteq J \cup \{i\} $.  
By Def. \ref{def:constr-ref} of constraint refinement, 
$ \cj {C_1}{C} $ holds too, because $ C' $ is contained in $ C $ verbatim, and
we have just proved that $ \suType {\eta_1'}{\eta_1}$. 

By the inductive hypothesis, if $ \algoMC $ terminates with  
$ (\sigma_2, C_2) $, then $ \cj {C_2}{C_1\sigma_2} $ holds. By 
Lem. \ref{lem:sound-C-to-Csigma}, the previous result 
$ \cj {C_1}{C} $ implies $ \cj {C_1\sigma_2}{C\sigma_2} $. By transitivity of
constraint refinement, we have $\cj {C_2}{C\sigma_2} $ and the lemma is proved.

Let now $ b  $ be an external choice, with $ b = \bechoicetext{}{\rho}{L_i}{b_i}
$.
For ease of exposition, let us first 
consider the case where the $\algoMC $ clause at line \ref{code:MC-psiext}
applies:
\begin{lstlisting}[
    backgroundcolor=\color{white}, 
    escapeinside={(*}{*)},
    firstnumber=89,
    firstline=191,
    lastline=203]
$ \mcclause {\St{l}{\psiext}
            } {\bechoiceTrap{i\in I}{\rho}{L_i}{b_i}} C K 
    = (\sigma_n\ldots\sigma_0, C_n) $(*\label{code:MC-psiext}*)
  if  $ \subEql{l}{\rho} $ and $ J_1 \subseteq I $ and $ I \subseteq J_1 \uplus
                                                                        J_2 $ 
  and $ C = C' \uplus \atomSu{\ceq{\psiext}{\sechoice{j}{J_1}{J_2}{\eta}}} $
  and for each $ k \in [1\ldots n] $ where $ I = \{L_{j_1}, \ldots, L_{j_n}\} 
  $ and $ n \geq 0 $
    $ (\sigma_k, C_k) = \mcclause {(\St l {\eta_{j_k}}
                                      } {b_{j_k})\sigma_{k-1}\ldots\sigma_0
                                      } {C_{k-1}
                                      } {K\sigma_{k-1}\ldots\sigma_0} $ 
      where $ \sigma_0 = \id $ and $ C_0 = C $
\end{lstlisting}
In this case, the final constraints set $ C_n $ is obtained by recursively
refining the input $ C $ through invocations to $ \algoMC $ at line 93. 
By inductive hypothesis, we have that 
each set $ C_1, C_2, \ldots C_n $ is a successive refinement of $ C $,
i.e. $ \cj {C_n}{C_{n-1}\sigma_n} , \ldots, \cj {C_1}{C_0\sigma_0} $ hold.
Since $ C_0 = C $, the last refinement can be written as 
$ \cj {C_1}{C\sigma_0} $. By Lem. \ref{lem:sound-C-to-Csigma} this implies
$ \cj {C_1\sigma_1}{C\sigma_1\sigma_0} $. Since $ \cj{C_2}{C_1\sigma_1} $   
holds by inductive hypothesis, it also follows that
$ \cj{C_2}{C\sigma_1\sigma_0} $ by transitivity of constraint refinement.
By repeating this reasoning, 
we can conclude by transitivity of
constraint refinement that $ \cj{C_k}{C\sigma_k\ldots\sigma_0} $ for any $ k
\leq |I| $.
Since a well-formed behaviour is finite and $ b $ is well-formed by assumption,
the external choice in $ b $ has a finite set of labels $ L_i $ for $ i \in
I$. Therefore $ I $ is a finite set, and we can conclude that 
$ \cj{C_n}{C\sigma_n\ldots\sigma_0} $ for $ n = |I| $, which proves the lemma.

Suppose that either the clause at line \ref{code:MC-psiext-psi} 
or line \ref{code:MC-psiext-not} applies. Notice that in these two cases we
cannot apply the inductive hypothesis directly, since $ \algoMC $ is passed the
entire behaviour $ b $ as input. However, the only clause of $ \algoMC$ that
applies is the third clause at line \ref{code:MC-psiext}. In the first case
(line \ref{code:MC-psiext-psi}) the algorithm creates a new external choice
constraint in $ C $, by associating the session $
\sechoicetext{i}{I}{\emptyset}{\psi} $ to a fresh session variable $ \psiext $.
The first $ \algoMC $ clause cannot be called again, because the stack contains
$ \psiext $ in place of $ \psi $ (line 82). Notice that the set of inactive
labels is the empty set, therefore the second $ \algoMC $ clause does not
apply, since $ I = I_1 $ in this case and therefore $ I_3 = \emptyset$. 
Therefore the only $ \algoMC $ clause that applies is the the third one, and the
lemma is proved as in the first case. By a similar reasoning, the only clause
that applies to the sub-call to $ \algoMC $ in the first clause (line 87) is the
third clause, and the lemma is proved similarly.

\end{proof}

Before proving termination of $ \algoMC $, we prove some further
properties of the auxiliary functions from Sec. \ref{sec:algoFunc}, which $
\algoMC $ uses.
\begin{lemma}[\textsf{finalize}]\label{lem:alg-finalize}
For any well-formed stack $ \Delta $, $\mathsf{finalize}~\Delta $
terminates.
\end{lemma}
\begin{proof}
By structural induction on $ \Delta $. If $ \Delta = \epsilon$, then $
\mathsf{finalize}~\epsilon $ terminates with the empty substitution (line 12).
If $ \Delta = \stBase l \eta {\Delta'}$, then the lemma follows by inductive
hypothesis when $ \eta = \tend$ or $ \eta = \psi $ (lines 13 and 14), or it
terminates with an error otherwise, since no other \textsf{finalize} clause
applies.
\end{proof}

\begin{lemma}[\textsf{checkFresh}]\label{lem:alg-checkFresh}
For any well-formed stack $ \Delta $ and label $ l $,
the function call $\mathsf{checkFresh}(l, \Delta )$ terminates. Moreover, if $ l
$ occurs in $ \Delta $ and the function call $ \mathsf{checkFresh}(l, \Delta ) =
(\sigma, \Delta') $, then $ \bsize{\Delta' } < \bsize{\Delta} $; otherwise $ \bsize{\Delta' } \leq
\bsize{\Delta} $.
\end{lemma}
\begin{proof}
By structural induction on $ \Delta $. The base case is $ \Delta = \epsilon $,
whereby $ \mathsf{checkFresh} $ terminates with $ (\id, \epsilon) $ (line 2),
and the second part of the lemma is vacuously true, since $ l $ cannot occur in
the empty stack $ \epsilon $. In the inductive case when $ \Delta =
\stBase{l'}{\eta}{\Delta'}$, \textsf{checkFresh} terminates immediately 
with a
smaller stack 
if $ l = l' $ and either $ \eta = \tend $ (line 3) or $ \eta = \psi $ (line
4).
If $ l \neq l' $, then \text{checkFresh}  terminates with a smaller stack by
inductive hypothesis. Otherwise it terminates with an error (because no other
clauses apply) and the second condition of the lemma is vacuously true.
\end{proof}

\begin{corollary}[\textsf{closeTop}]\label{cor:alg-closeTop}
For any well-formed stack $ \Delta $, $\mathsf{closeTop}~ \Delta$
terminates. Moreover, if $ \mathsf{closeTop}~\Delta = (\sigma, \Delta') $,
then $ \bsize{\Delta' } < \bsize{\Delta} $.
\end{corollary}
\begin{proof}
By Lem. \ref{lem:alg-checkFresh}.
\end{proof}

\begin{lemma}[\textsf{sub}]\label{lem:fin-sub}
For any session $ \eta_1, \eta_2 $ and well-formed $ C $, $
\mathsf{sub}(\eta_1, \eta_2, C) $ terminates. 
\end{lemma}
\begin{proof}
By structural induction on $ \eta_1 $ and $ \eta_2 $.
\end{proof}

We can now prove termination of $ \algoMC $:
\begin{proposition}[Termination of $ \algoMC $]\label{prop:mc-termination}
For any well-formed $ C $ and $ \mc \Delta {K[b]}$, $
\mcclause{\Delta}{b}{K}{C} $ terminates.
\end{proposition}
\begin{proof}
By lexicographic induction on $ \bsize{\mc \Delta {\groundify{K[b]}}} $
and $ \bsize{\groundify b} $. We proceed by case analysis on the $
\algoMC $ clauses in Sec. \ref{sec:algoMC}.

\begin{itemize}
\item \textbf{Ended session}:
\begin{lstlisting}[
    backgroundcolor=\color{white}, 
    escapeinside={(*}{*)},
    firstline=1,
    lastline=5]
-- remove terminated frames
$ 
  \mcclause {\St l \tend} b C K  
  = \mcclause {\Delta} b C K
$ (* \label{code:MC-delta-end} *)
\end{lstlisting}

This case follows by inductive hypothesis, since $ \bsize{\Delta} < \bsize
{\St l \tend } $ by Def. \ref{def:stack-size}.\\

  \item \textbf{Tau}: 
\begin{lstlisting}[
    backgroundcolor=\color{white}, 
    escapeinside={(*}{*)},
    firstnumber=4,
    firstline=7,
    lastline=12]
-- $ \textcolor{mygreen}{\algoMC}$ terminates with behavior $\textcolor{mygreen}{\tau }$ 
$ 
  \mcclause \Delta \tau C \emptyK    
  = (\sigma,~ C\sigma) 
$ (* \label{code:MC-delta-tau} *)
  if $  \sigma = $ finalize $ \Delta $
\end{lstlisting}

By Lem. \ref{lem:alg-finalize}, \textsf{finalize} terminates for any $ \Delta $,
therefore algorithm $ \algoMC $ always terminates and the proposition is
proved.\\

  \item \textbf{Tau sequence}:
\begin{lstlisting}[
    backgroundcolor=\color{white}, 
    escapeinside={(*}{*)},
    firstnumber=8,
    firstline=14,
    lastline=18]
-- pop a sub-behavior from the continuation stack 
$ 
  \mcclause \Delta \tau C {\consK b K}    
  = \mcclause \Delta b C K 
$ (* \label{code:MC-delta-ktau} *)
\end{lstlisting}

By definition of continuation stack, $ b \cdot K[\tau] = K[\tau;b] $.
By Definition \ref{def:fin-behav-size} of behaviour size, $ \bsize {K[\tau;b]} =
1 + \bsize {K[b]} $.
Therefore the configuration $ \mc \Delta {(b \cdot K)[\tau]}$ has greater size
than configuration $ \mc \Delta {K[b]} $, and the proposition follows by
inductive hypothesis.\\

  \item \textbf{New session push}:
\begin{lstlisting}[
    backgroundcolor=\color{white}, 
    escapeinside={(*}{*)},
    firstnumber=11,
    firstline=20,
    lastline=28]
-- push a new frame on the stack
$ 
  \mcclause \Delta {\pusho{l}{\eta}{}} C K 
  = (\sigma_2\sigma_1, C_2)
$ (* \label{code:MC-push} *)
  if  ($\sigma_1, \Delta_1$) = checkFresh$ (l, \Delta) $
  and ($\sigma_2, C_2$) = $
    \mcclause {\stBase {l}{\eta\sigma_1}{\Delta_1}} \tau {C\sigma_1}{K\sigma_1}
  $
\end{lstlisting}

By Lem. \ref{lem:alg-checkFresh}, $ \mathsf{checkFresh}(l, \Delta) $ always
terminates, and if it terminates successfully then $ \Delta_1 $ has smaller or
equal size to $ \Delta $. By Definition \ref{def:conf-size} of configuration
size and by Lem. \ref{lem:fin-subst}:
\begin{align*}
\bsize{\mc
\Delta {K[\pusho{l}{\eta}{}}]}
 &=  1 + 2 + \bsize {\Delta} + \bsize{K[\tau]}
\\& \geq 1 + 2 + \bsize{\Delta_1} 
               + \bsize{ {K\sigma_1[\tau]}}
\\& = 1 + 1 + \bsize{\stBase {l}{\eta\sigma_1}{\Delta_1}} 
               + \bsize{ {K\sigma_1[\tau]}}
\\&> 1 + \bsize{\stBase {l}{\eta\sigma_1}{\Delta_1}} + \bsize{
{K\sigma_1[\tau]}} \\&= \bsize{\mc{\stBase {l}{\eta\sigma_1}{\Delta_1}}
{K\sigma_1[\tau]}}
\end{align*}

and the proposition is proved by the inductive hypothesis.
Recall that, by definition of size, $ 1 + \bsize\Delta = \bsize{\stBase l \eta
\Delta} $ for any $ \eta $ and $ l $, since $ \bsize - $ is invariant to 
session types.
\\

  \item \textbf{Value send}:
\begin{lstlisting}[
    backgroundcolor=\color{white}, 
    escapeinside={(*}{*)},
    firstnumber=16,
    firstline=30,
    lastline=48]
-- send
$
  \mcclause { \St l \psi } {\popo{\rho}{!T}} C K
    = 
  (\sigma_2\sigma_1, C_2)
$  (*\label{code:MC-send-psi}*)
  if  $ \subEql{l}{\rho} $
  and $ \sigma_1 = \subst\psi {!\alpha.\psi'} $ where $ \alpha, \psi' $ fresh
  and $ (\sigma_2, C_2) 
          = 
        \mcclause { \St l {\psi'}\sigma_1 } \tau {
          C\sigma_1 \cup \singleSu{T}{\alpha} } { K\sigma_1 } 
      $

$ 
  \mcclause {\St l {!\alpha.\eta}}{\popo{\rho}{!T}} C K
    = \mcclause { \St l \eta } \tau {C \cup \singleSu{T}{\alpha}} K
$ (*\label{code:MC-send-eta}*)
  if  $ \subEql{l}{\rho} $ 
\end{lstlisting}

By the inductive hypothesis in both cases.
The other cases when $ b $ is $ \popo{\rho}{!T} $,  $ \popo{\rho}{?T} $,
$ \popo{\rho}{!L_k}$, $ \popo{\rho}{?L_k} $,
$ \popo{\rho}{!\rho_d}$ and $ \popo{\rho}{?l_d}$ are all proved similarly.\\ 

\item \textbf{Session delegation}:
\begin{lstlisting}[
    backgroundcolor=\color{white}, 
    escapeinside={(*}{*)},
    firstnumber=37,
    firstline=78,
    lastline=97]
$
  \mcclause {\stBase l {!\eta_d.\eta} {\St{l_d}{\eta_d'}}
  } {\popo{\rho}{!\rho_d}} C K
  = (\sigma_2\sigma_1, C_2)
$(*\label{code:MC-del-eta}*)
  if  $ \subEql{l}{\rho} $ and $ \subEql{l_d}{\rho_d} $
  and $ (\sigma_1, C_1) = \nsub(\eta_d', \eta_d, C) $
  and $ (\sigma_2, C_2) = 
    \mcclause {\St l {\eta} \sigma} \tau {C_1} {K\sigma_1} 
$

$  \mcclause {\stBase l {\eta} {\St{l_d}{\eta_d'}}
  }{\popo{\rho}{!\rho_d}} C K
  = (\sigma_2\sigma_1, C_2)
$(*\label{code:MC-del-checkFresh}*)
  if  $ \subNotEql{l_d}{\rho_d} $
  and $ (\sigma_1, \Delta_1) = $ checkFresh$ (l_d, \stBase l \eta
                                                        {\St{l_d}{\eta_d'}}) $
  and $ (\sigma_2, C_2) = \mcclause  {\Delta_1}{\popo{\rho}{!\rho_d}
    } { C\sigma_1} {K\sigma_1} $
\end{lstlisting}

By Lem. \ref{lem:fin-sub} and inductive hypothesis in the first clause, and by
\ref{lem:alg-checkFresh} and inductive hypothesis in the second clause.
\\

\item \textbf{External choice}:
\begin{lstlisting}[
    backgroundcolor=\color{white}, 
    escapeinside={(*}{*)},
    firstnumber=77,
    firstline=160,
    lastline=203]
-- ex. choice
$
  \mcclause {\St{l}{\psi}} {\bechoice{}{\rho}{L_i}{b_i}} C K
    = (\sigma_2\sigma_1, C_2) 
$ (*\label{code:MC-psiext-psi}*)
  if  $ \subEql{l}{\rho} $
  and $ \sigma_1 = \subst{\psi}{\psiext} $ where $ \psiext $ fresh
  and $ C_1 = \atomSu{\ceq{\psiext}{\sechoice{i}{I}{\emptyset}{\psi}}} 
    $ where $ \overrightarrow {\psi_i} $ fresh
  and $ (\sigma_2, C_2) = 
      \mcclause {\St{l}{\psiext}\sigma_1
                } {\big(\bechoice{}{\rho}{L_i}{b_i}\big)\sigma_1
                } {C\sigma_1\cup C_1
                } {K\sigma_1} $
                
$
  \mcclause {\St{l}{\psiext}
            } {\bechoiceTrap{j \in I_1\uplus I_2\uplus I_3}{\rho}{L_j}{b_j}
            } C K 
    = \mcclause {\St{l}{\psiext}
                } {\bechoiceTrap{j \in I_1\uplus I_2\uplus I_3}{\rho}{L_j}{b_j}
                } {C_1
                } {K}
$ (*\label{code:MC-psiext-not}*)
  if  $ \subEql{l}{\rho} $ 
  and $ C = C' \uplus \atomSu{ 
                          \ceq \psiext {\sechoice{j}{I_1J_1}{I_2J_2}{\eta}} }
      $ and $ I_3 \freshfrom J_1\uplus J_2 $ and $ I_3, J_1 \neq \emptyset $
  and $ C_1 = C' \cup \atomSu{\ceq {\psiext}
                                   {\sechoice{i}{I_1}{I_2I_3J_1J_2}{\eta}} } $

$ \mcclause {\St{l}{\psiext}
            } {\bechoiceTrap{i\in I}{\rho}{L_i}{b_i}} C K 
    = (\sigma_n\ldots\sigma_0, C_n) $(*\label{code:MC-psiext}*)
  if  $ \subEql{l}{\rho} $ and $ J_1 \subseteq I $ and $ I \subseteq J_1 \uplus
                                                                        J_2 $ 
  and $ C = C' \uplus \atomSu{\ceq{\psiext}{\sechoice{j}{J_1}{J_2}{\eta}}} $
  and for each $ k \in [1\ldots n] $ where $ I = \{L_{j_1}, \ldots, L_{j_n}\} 
  $ and $ n \geq 0 $
    $ (\sigma_k, C_k) = \mcclause {(\St l {\eta_{j_k}}
                                      } {b_{j_k})\sigma_{k-1}\ldots\sigma_0
                                      } {C_{k-1}
                                      } {K\sigma_{k-1}\ldots\sigma_0} $ 
      where $ \sigma_0 = \id $ and $ C_0 = C $
\end{lstlisting}

As shown in the proof of Lem. \ref{lem:sound-C-ext},
the first two clauses recursively call $ \algoMC $ with the same input, except
for an expanded $ \psi $ and $ C_1 $ in the first case (line
\ref{code:MC-psiext-psi}), and a modified $ \psiext $ in the second case (line
\ref{code:MC-psiext-not}).
Let us therefore analyse the third clause first.

Let $ b = \bechoicetext{i\in I}{\rho}{L_i}{b_i} $.
If the third clause is applied (line \ref{code:MC-psiext}), the proposition is a
straightforward consequence of the inductive hypothesis and Lem.
\ref{lem:fin-subst}, since each recursive invocation of $ \algoMC $ has as input
$ b_{j_k} $, which is a  sub-behaviour of the input behaviour 
$ b $.

If the first clause is invoked (line \ref{code:MC-psiext-psi}), 
by Lem. \ref{lem:sound-C-ext} and Lem. \ref{lem:fin-c-expand}
we have $ \bsize{\groundify{\mc \Delta b}} =
\bsize{\groundifyBase{\mc {\Delta\sigma_1}{b\sigma_1}}{C\sigma_1 \cup C_1}} $.
Since the size of the behaviour is the same and the only applicable $
\algoMC $ clause is the third one, the proposition is proved as before.

If the second clause is invoked (line \ref{code:MC-psiext-not}), the proof
proceeds as in the previous case, except for the fact that the constraint set $
C $ is refined instead of expanded. 
\\
 
\item \textbf{Behavior sequence}:
\begin{lstlisting}[
    backgroundcolor=\color{white}, 
    escapeinside={(*}{*)},
    firstnumber=96,
    firstline=205,
    lastline=209]
-- sequencing
$ \mcclause \Delta {\bseq {b_1} {b_2}} C K
  = 
  \mcclause \Delta {b_1} C {\consK {b_2} K}
  $ (*\label{code:MC-seq}*) 
\end{lstlisting}

By definition of stack continuation, $ K[b_1;b_2] = b_2\cdot K[b_1] $, and
therefore 
$ \bsize{\mc\Delta{K[b_1;b_2]}} = \bsize{\mc\Delta{b_2\cdot K[b_1]}}$.
Because of this, and since $ \bsize{b_1} < \bsize{b_1;b_2}$ for any $ b_2$, the
proposition is proved by the second ordering in the lexicographic order
of the inductive hypothesis. 
\\

\item  \textbf{Behavior choice}:
\begin{lstlisting}[
    backgroundcolor=\color{white}, 
    escapeinside={(*}{*)},
    firstnumber=99,
    firstline=211,
    lastline=220]
-- internal choice in the behavior      
$ \mcclause  \Delta {b_1\oplus b_2} C K
    = (\sigma_2\sigma_1, C_2) 
    $ (*\label{code:MC-plus}*)
  if  $ (\sigma_1, C_1) = \mcclause \Delta {b_1} C K $
  and $ (\sigma_2, C_2) = \mcclause {\Delta\sigma_1 
                                  } {b_2\sigma_1
                                  } {C_1
                                  } {K\sigma_1}
      $
\end{lstlisting}

By inductive hypothesis. 
\\

\item \textbf{Recursive behavior}:
\begin{lstlisting}[
    backgroundcolor=\color{white}, 
    escapeinside={(*}{*)},
    firstnumber=109,
    firstline=228,
    lastline=237]
-- rec 
$ \algoMC(\mcconf{\Delta}{\orec{b}{\beta}},~ C)
   =  (\sigma_2\sigma_1, C_2) $ (*\label{code:MC-rec}*)
  if  $ C = C' \uplus \singleSu{b'}{\beta} $
  and $ (\sigma_1, C_1) = \mcclause \epsilon b {C' \cup
                                      \singleSu{\tau}{\beta}} \epsilon $ 
  and $ (\sigma_2, C_2) = \mcclause {\Delta\sigma_1
                                  } \tau {
 (C_1 \backslash \singleSu{\tau}{\beta}) \cup
 (\singleSu{b'}{\beta})\sigma_1 } { K\sigma_1 } $
\end{lstlisting}
    
By inductive hypothesis. By definition of ground term, 
$ \groundifyBase {\orec b \beta } {C \uplus \singleSu{b'}{\beta}} =
  \orec {\groundifyBase b {C\cup \singleSu \tau \beta }} \beta$ holds, and
  therefore the inductive hypothesis can be applied on the first inner call to $
  \algoMC $ at line 112. The second call to $ \algoMC $ terminates by the
  inductive hypothesis as well, since the size of $ \tau $ is strictly less than
  the size of a recursive term $ \orec b \beta $.
\\
    
\item  \textbf{Behavior variables}:
\begin{lstlisting}[
    backgroundcolor=\color{white}, 
    escapeinside={(*}{*)},
    firstnumber=115,
    firstline=239,
    lastline=243]
-- behavior variable
$ \mcclause \Delta \beta C K
  = \mcclause \Delta b C K $(*\label{code:MC-beta}*)
  where  $ b = \bigoplus \{b_i \where \exists i. ~(b_i \subseteq \beta) \in C\}
  $ 
\end{lstlisting}

This case is proved as for the case of behavior choice, after unrolling the
definition of $ \algoMC $. 
\\

\item  \textbf{Forced session termination}:
\begin{lstlisting}[
    backgroundcolor=\color{white}, 
    escapeinside={(*}{*)},
    firstnumber=119,
    firstline=245,
    lastline=253]
-- try to close the top session if all previous clauses fail
$ \mcclause \Delta b C K
  = (\sigma_2\sigma_1, C_2) 
$ (*\label{code:MC-closetop}*)
  if  $ (\sigma_1, \Delta_1) = $ closeTop$ (\Delta) $
  and $ (\sigma_2, C_2) = \mcclause { \Delta_1
                                  } { b\sigma_1
                                  } { C\sigma_1
                                  } { K\sigma_1} $
\end{lstlisting}

By Cor. \ref{cor:alg-closeTop} and inductive hypothesis. 
\\

\end{itemize}

\end{proof}

Termination of $ \algoSI $ follows directly from this result:
\begin{theorem}[Termination of $ \algoSI $]
For any well-formed $ C $ and $ b $, $
\algoSI(b, C) $ terminates.
\end{theorem}
\begin{proof}
By Prop. \ref{prop:mc-termination}.
\end{proof}


\clearpage
\section{Soundness of Algorithm $\algoSI$ }
\label{sec:soundness}
\newcommand{\mypureBase}[2]{ #1 \vdash \mathit{pure}(#2)}
\newcommand{\mypure}[1]{ \mypure C {#1} }
This section shows a soundness result for Algorithm $ \algoSI $.
Namely we show that, given a well-formed $ b $ and $ C $, if 
$ \algoSI $ terminates with a solution $ (\sigma_1, C_1)$, then the
configuration $ \mc \epsilon {b\sigma_1} $ is strongly normalizing under $ C_1
$. 

The core of the proof revolves around proving that $ \algoMC $ is sound. 
However some extra technical
machinery is required before proceeding with the proof, because of 
two issues: the lazy detection of terminated sessions and constraint refinement.

The first issue is that Algorithm $ \algoMC $ detects the termination of
sessions lazily.
During inference, open session variables $ \psi $ in the stack are assigned the
session $ \tend $ only when no other $ \algoMC $ clause is applicable (line
\ref{code:MC-closetop}), or at the very end of session inference at line
\ref{code:MC-delta-tau}, when the stack is empty and the behaviour is $ \tau $.

Because of this laziness, the proof of soundness of $ \algoMC $ requires 
some flexibility to deal with terminated sessions. We therefore introduce the
following equivalence relation on stacks with terminated sessions:

\begin{definition}[Terminated session equivalence]\label{def:mc-term-delta-eq}
\begin{center}
\irule*[][]
  { }
  {\epsilon \equiv \epsilon}
\irule*[][]
  { \Delta \equiv \Delta' }
  { \St l \tend \equiv \Delta'  }
\irule*[][]
  { \Delta \equiv \Delta' }
  { \St l \eta \equiv \St l \eta ' }
  \condBox{ $ \eta \neq \tend $ } 
\end{center}

\end{definition}

The second issue is that constraints in $ C $ may be refined from time to time
in Algorithm $ \algoMC $, and we need to show that strong normalization is
preserved under constraint refinement. 
For example, let the stack contain a variable $ \psiint
$, and let $ C $ contain the constraint $ \singleSu {\psiint}{
\sichoicetext{i}{I}{\eta} } $.  
During session inference the internal choice $ \sichoicetext{i}{I}{\eta} $ might
be expanded, for example by adding a new label $ L_k $. In such a case the
constraint $ \singleSu {\psiint}{\sichoicetext{i}{I}{\eta} } $ is 
refined to $ \singleSu {\psiint}{ \sichoicetext{i}{I, k}{\eta} } $, and $ C $
is updated to some $ C' $ accordingly.
If a behaviour $ b $ is strongly normalizing under $ C $, it is now unclear if
$ b $ is also strongly normalizing in $ C' $.

To address this issue, we introduce the following form of stack sub-typing:

\newcommand{\suDBase}[3]{{#3}\vdash{#1} <: {#2}}
\newcommand{\suD}[2]{\suDBase{#1}{#2}{C}}
\begin{definition}[Stack sub-typing]
Let $ C $ be well-formed. A stack $ \Delta_1 $ is a subtype of stack $ \Delta_2
$, or $ \suD {\Delta_1}{\Delta_2} $, when the following relations are satisfied:
\begin{align*}
    & \suD \epsilon \epsilon
    &
 \\ & \suD {\stBase l {\eta_1} {\Delta_1}} {\stBase l {\eta_2} {\Delta_2}} 
    & \text{ if }
 \coType{\eta_1}{\eta_2} \text{ and } \suD {\Delta_1} {\Delta_2}
\end{align*} 
\end{definition}

The following theorem is crucial in proving that 
that strong normalization is preserved under constraint refinement: 
\begin{theorem}[Liskov's substitution principle]        
  \label{prop:liskov}
    Let $ \cj{C_2}{C_1\sigma} $, $ \suDBase{\Delta_2}{\Delta_1\sigma}{C_2}$.
    If $ \dstepBase{\Delta_1}{b}{\Delta_1'}{b'}{C_1}$, 
    then 
      $ \dstepBase {\Delta_2} {b\sigma} {\Delta_2'} {b'\sigma_1}{C_2} $
      and  
      $ \suDBase{\Delta_2'}{\Delta_1'\sigma}{C_2} $.
  \end{theorem} 
\begin{proof}

By rule induction. The proof of 3 is a trivial consequence of the hypothesis
$ \mypureBase{C_2}{\Delta_2} $,
since the continuation of a pure session is itself pure. We only prove 1 and 2:

\textbf{\textsf{Case} }\RefTirName{End}\textbf{\textsf :}
Suppose that $ \dstepBase {\stBase l \tend {\Delta_1}} {b } {\Delta_1}{b} {C_1}
$. 
By definition of substitution $ (\stBase{l}{\tend}{ \Delta_1})\sigma = 
\stBase{l}{\tend}{ \Delta_1\sigma} $; by definition of stack sub-typing
the hypothesis $ \suDBase{\Delta_2}{(\stBase{l}{\tend}{
\Delta_1})\sigma}{C_2}$ implies that  $
\Delta_2 = \stBase{l}{\tend}{ \Delta_2'} $ for some $ \Delta_2'$ such that $
\suDBase{\Delta_2'}{\Delta_1\sigma}{C_2}$.
Rule \RefTirName{End} yields
$ \dstepBase {\Delta_2} {b\sigma} {\Delta_2'} {b\sigma} {C_2}$, which proves the
proposition together with $ \suDBase{\Delta_2'}{\Delta_1\sigma}{C_2}$

\textbf{\textsf{Case} }\RefTirName{Beta}\textbf{\textsf :}
Let $\dstepBase {\Delta_1} {\beta } {\Delta_1}{b}{C_1} $, assuming that $
\subBase{b}{\beta}{C_1}$.
By definition of refinement, $ \cj {C_2}{C_1\sigma} $ implies that $
\subBase{C_2}{b\sigma}{\beta\sigma} $.
Since $ \beta\sigma = \beta $ by definition of inference substitution,
$ \subBase{C_2}{b\sigma}{\beta} $ holds, and 
therefore the proposition is proved
by applying Rule \RefTirName{Beta} on $ \mc {\Delta_2} {\beta} $.
             
\textbf{\textsf{Case} }\RefTirName{Plus}\textbf{\textsf :}
Let $ \dstepBase {\Delta_1} {b_1 \oplus b_2} {\Delta_1}{b_i} {C_1} $ with $ i
\in \{1, 2\} $.
By definition of substitution $ (b_1\oplus b_2) \sigma = b_1\sigma \oplus
b_2\sigma$. The proposition is proved by straightforward application of Rule
\RefTirName{Plus}, which yields $ \dstepBase {\Delta_2}{b_1\sigma \oplus
b_2\sigma } {\Delta_2}{b_i\sigma}{C_2} $.

\textbf{\textsf{Case} }\RefTirName{Push}\textbf{\textsf :}
Let $ \dstepBase {\Delta_1} {\pusho{l}{\eta}{} } {\stBase {l}{\eta}{\Delta_1}}
{\tau} {C_1}$ with $ l \freshfrom \Delta $.
By application of Rule \RefTirName{Push}, 
$ \dstepBase {\Delta_2} {\pusho{l}{\eta}{}\sigma } 
             {\stBase {l} {\eta\sigma} {\Delta_2}} {\tau}{C_2} $ holds because 
$ \pusho{l}{\eta}{}\sigma = \pusho{l}{\eta\sigma}{} $.
By definition of subtyping 
$ \coTypeBase {C_2}{\eta\sigma}{\eta\sigma} $ holds by reflexivity,
and therefore $ \suDBase
{ \stBase l {\eta\sigma} {\Delta_2}} 
{ \stBase l {\eta\sigma} {\Delta_1\sigma}}
{C_2} $ holds, because $ 
\suDBase{\Delta_2} 
{\Delta_1\sigma}{C_2} $ holds by hypothesis; therefore the proposition is
proved.

\textbf{\textsf{Case} }\RefTirName{Out}\textbf{\textsf :}
Let $ \dstepBase {\stBase {l}{!T.\eta}{\Delta_1}} {\popo{l}{!T_0} }
                 {\stBase {l}{\eta}   {\Delta_1}} {\tau}            {C_1}$,
with  $ \mypureBase{C_1}{T_0}$ and $
\coTypeBase{C_1}{T_0}{T} $.
By definition of substitution $ (\stBase {l}{!T.\eta}{\Delta_1})\sigma = 
\stBase {l}{!T\sigma.\eta\sigma}{\Delta_1\sigma} $ holds.
By definition of stack sub-typing,  
$ \suDBase {\Delta_2}{  \stBase {l}{!T\sigma.\eta\sigma}{\Delta_1\sigma} }{C_2}
$ implies that $ \Delta_2 = \stBase {l}{!T'.\eta'}{\Delta_2'} $ such that 
$ \coTypeBase{C_2}{T\sigma}{T'} $
and 
$ \suDBase{\Delta_2'}{\Delta_1\sigma}{C_2} $.

Since $ \cj{C_2}{C_1\sigma} $ holds by hypothesis, then 
$ \coTypeBase{C_1}{T_0}{T} $ implies $ \coTypeBase{C_2}{T_0\sigma}{T\sigma} $.
By transitivity $ \coTypeBase{C_2}{T_0\sigma}{T\sigma} $ and $
\coTypeBase{C_2}{T\sigma}{T'} $ imply $ \coTypeBase{C_2}{T_0}{T'} $.
Since $ \popo{l}{!T_0}\sigma = \popo{l}{!(T_0\sigma)}$, an application Rule
\RefTirName{Out} yields 
$ \dstepBase {\stBase {l}{!T'.\eta'}{\Delta_2'}} {\popo{l}{!T_0\sigma} }
             {\stBase {l}{\eta'}    {\Delta_2'}} {\tau}                 {C_2}$,
which proves 1.
By definition of session sub-typing, the hypothesis 
$ \suDBase {\stBase {l}{!T'.\eta'}{\Delta_2'} }{  \stBase
{l}{!T\sigma.\eta\sigma}{\Delta_1\sigma} }{C_2} $ implies 
$ \coTypeBase {\eta'}{\eta\sigma} {C_2} $. We have already proved that
$ \suDBase{\Delta_2'}{\Delta_1\sigma}{C_2} $ holds, therefore we can infer
$ \suDBase {\stBase {l}{\eta'}{\Delta_2'} }{  \stBase
{l}{\eta\sigma}{\Delta_1\sigma} }{C_2} $, which proves 2.
 
\textbf{\textsf{Cases} }\RefTirName{In, Del, Res}\textbf{\textsf :}
proved as in case \RefTirName{Out}, using the fact that sub-typing is covariant
for \RefTirName{In, Res}, and it is contravariant for \RefTirName{Del}.

\textbf{\textsf{Case} }\RefTirName{ICh}\textbf{\textsf :}
Let $ \dstepBase {\stBase {l}{\sichoicetext{i}{I}{\eta}}{\Delta_1}}
            {\popo{l}{!L_k} 
            }
            {\stBase l {\eta_k} {\Delta_1} }    {\tau}{C_1}$ 
with $ k \in I$.
By definition of internal choice variables, configuration
$\mc { \stBase l {\sichoicetext{i}{I}{\eta}} {\Delta_1} }
     { \popo{\rho}{!L_k}}
     $
is equivalent to 
$ \mc {\stBase {l}{\psiint}{\Delta_1}}
      {\popo{\rho}{!L_k} }
$ 
with $ \sueqBase{C_1}{\psiint}{ \sichoicetext{i}{I}{\eta}} $ for some $ \psiint
$.
Since $ \cj {C_2}{C_1\sigma} $ holds by hypothesis,
by definition of constraint refinement there exists a session $ \eta' = 
 \sichoicetext{j}{J}{\eta'}
 $
such that 
$ \sueqBase{C_2}{\psiint}{\eta'} $ 
and 
$ 
\coTypeBase{C_2}{\sichoicetext{j}{J}{\eta'}}{\sichoicetext{i}{I}{\eta\sigma}}
$ 
hold, with $ J \subseteq I $.

Since $ J \subseteq I $, $ k \in I $ implies $ k \in J $;
and since $ (\popo{\rho}{!L_k})\sigma = \popo{\rho}{!L_k}$, 
an application of Rule \RefTirName{ICh} yields 
$ \dstepBase {\stBase {l}{\sichoicetext{j}{J}{\eta'}}{\Delta_2'}}
             {(\popo{\rho}{!L_k})\sigma}
             {\stBase {l}{\eta_k'}{\Delta_2'} }    {\tau}{C_2} $, 
which proves 1.
By definition of session sub-typing $ 
\coTypeBase{C_2}{\sichoicetext{j}{J}{\eta'}}{\sichoicetext{i}{I}{\eta}}
$ implies $ \coTypeBase {C_2}{\eta_k'}{\eta_k} $; since 
$ \coTypeBase{C_2}{\Delta_2'}{\Delta_1\sigma} $ holds by hypothesis, then
$ \coTypeBase{C_2}{\stBase l {\eta_k'}\Delta_2'}{\Delta_1\sigma} $ holds too,
which proves 2.

\textbf{\textsf{Case} }\RefTirName{ECh}\textbf{\textsf :} similar to the case
for \RefTirName{ICh}.
%
\end{proof}

Before proving soundness, we also need a lemma about the session sub-type
inference. The $ \mathsf{sub} $ function infers whether a session $ \eta_1 $ is
a sub-type of session $ \eta_2$:
\begin{lemma}[Soundness of $\mathsf{sub}$]\label{lem:sound-sub}
If $ \mathsf{sub}(\eta_1, \eta_2, C) = (\sigma_1, C_1) $, then $ \coTypeBase
{C_1}{\eta_1\sigma_1}{\eta_2\sigma_1}$.
\end{lemma}
\begin{proof}
By structural induction on $ \eta_1$. 
The structure of the proof is similar to the structure of the proof for Lem.
\ref{lem:sound-C-sub-ext}.

The lemma is trivial when $ \eta_1 = \eta_2 = \tend $.
When $ \eta_1 = !T_1.\eta_1' $ and $ \eta_2 =
!T_2.\eta_2'$, $ C $ is expanded with the constraint $\singleSu{T_2}{T_1}$.   
By structural induction $ \eta_1' $ is a subtype of $ \eta_2' $; 
since $ C_1 $ contains the constraint $\singleSu{T_2}{T_1}$,
$ \eta_1 $ is a subtype of $ \eta_2 $ by definition of subtyping. The case for
receive is similar; the cases for delegation and resume follow straightforwardly
by structural induction.

When $ \eta_1 $ is an internal 
choice variable $ \psiint $, function $ \mathsf{sub}$ calls an internal function
$ \mathsf f $ which inspects each label $ L_i $ in $ \eta_2 $, and either
expands $ \eta_1 $ if $ \eta_1 $ does not contain $ L_i $, 
or it just recursively invokes $ \mathsf{sub} $ on each common branch of
the internal choices of $ \eta_1 $ and $ \eta_2 $. The lemma is proved
straightforwardly by the definition of session sub-typing in the former
case, and by the inductive hypothesis in the latter case.

The proof is similar when $ \eta_1 $ is an external choice variable $ \psiext $.

When $ \eta_1 $ is a session variable $ \psi $, 
\textsf{sub} produces the substitution substitutes $ \subst {\psi}{\eta_2} $, 
and the lemma is proved straightforwardly by reflexivity. Similarly when $
\eta_2 = \psi$ holds. 
\end{proof}

Soundness of session inference depends on the following central result:
\begin{proposition}[Soundness of Algorithm $ \algoMC$]\label{prop:mc-soundness} 
Let $ C $ be well-formed in $ \mcconf{\Delta}{b} $.
If $ \algoMC(\mcconf{\Delta}{b}, C, K) = (\sigma_1, C_1) $,
then there exists $ \Delta' $ such that $ \Delta\sigma_1 \equiv \Delta' $ and
$\stackopBase {\Delta'}{K[b]\sigma_1}{C_1}{}$.
\end{proposition} 

\begin{proof}

As in the proof of Prop. \ref{prop:mc-termination},
the proposition is proved by lexicographic induction on $ \bsize{\mc \Delta
{\groundify{K[b]}}} $ and $ \bsize{\groundify b} $.
 
%
%
Let us proceed by case analysis on $ \algoMC(\mcconf{\Delta}{b}, C)$.

\begin{itemize}
\item 
\begin{lstlisting}[
    backgroundcolor=\color{white}, 
    escapeinside={(*}{*)},
    firstline=1,
    lastline=5]
-- remove terminated frames
$ 
  \mcclause {\St l \tend} b C K  
  = \mcclause {\Delta} b C K
$ (* \label{code:MC-delta-end} *)
\end{lstlisting}

Suppose that $ \mcclause {\St l \tend} b C K = (\sigma_1, C_1)$. By the
inductive hypothesis 
on the clause $ \mcclause \Delta b C K $, 
there exists $ \Delta' $ such that $ \Delta\sigma_1 \equiv \Delta' $ and
$\stackopBase {\Delta'}{K[b]\sigma_1}{C_1}{}$.
By Def. \ref{def:mc-term-delta-eq} of stack equivalence, we have $ \stBase l
\eta {\Delta\sigma_1} \equiv {\Delta\sigma_1} $.
Since $ \Delta\sigma_1 \equiv \Delta'$ holds by inductive hypothesis, it
follows that $ (\stBase l \eta \Delta)\sigma_1 \equiv \Delta' $. Therefore 
the proposition is proved by $ \Delta' $.

  \item 
\begin{lstlisting}[
    backgroundcolor=\color{white}, 
    escapeinside={(*}{*)},
    firstnumber=4,
    firstline=7,
    lastline=12]
-- $ \textcolor{mygreen}{\algoMC}$ terminates with behavior $\textcolor{mygreen}{\tau }$ 
$ 
  \mcclause \Delta \tau C \emptyK    
  = (\sigma,~ C\sigma) 
$ (* \label{code:MC-delta-tau} *)
  if $  \sigma = $ finalize $ \Delta $
\end{lstlisting}

By definition of \textsf{finalize}, a $ \Delta\sigma $ is equivalent to the
empty stack $ \epsilon $, since 
\textsf{finalize} succeeds only if it replaces all variables $ \psi $ in $
\Delta $ with $ \tend $, and only $ \tend $ sessions are left in it. 
The lemma is therefore proved by taking $ \Delta' = \epsilon $, because
$ \stackop {} {\tau\sigma} \epsilon {} $ holds trivially for any $ \sigma $ and
$ C $.

  \item 
\begin{lstlisting}[
    backgroundcolor=\color{white}, 
    escapeinside={(*}{*)},
    firstnumber=8,
    firstline=14,
    lastline=18]
-- pop a sub-behavior from the continuation stack 
$ 
  \mcclause \Delta \tau C {\consK b K}    
  = \mcclause \Delta b C K 
$ (* \label{code:MC-delta-ktau} *)
\end{lstlisting}
By inductive hypothesis.

  \item 
\begin{lstlisting}[
    backgroundcolor=\color{white}, 
    escapeinside={(*}{*)},
    firstnumber=11,
    firstline=20,
    lastline=28]
-- push a new frame on the stack
$ 
  \mcclause \Delta {\pusho{l}{\eta}{}} C K 
  = (\sigma_2\sigma_1, C_2)
$ (* \label{code:MC-push} *)
  if  ($\sigma_1, \Delta_1$) = checkFresh$ (l, \Delta) $
  and ($\sigma_2, C_2$) = $
    \mcclause {\stBase {l}{\eta\sigma_1}{\Delta_1}} \tau {C\sigma_1}{K\sigma_1}
  $
\end{lstlisting}

If this clause succeeds, then also the inner calls to \textsf{checkFresh} does,
which implies $ \Delta\sigma_1 \equiv \Delta_1$. 
Notice that $ \Delta_1 $ contain no frame of the form $ (l:\eta) $ for any
$ \eta $, i.e. $ l \freshfrom \Delta_1 $. Because of this and by taking $
\Delta' = \Delta_1 $, the abstract interpretation semantics allows the following
transition:
\begin{align*}
  \stackstep{_2}{\Delta'}{K[\pusho{l}{\eta}{}]}{  \stBase l \eta {\Delta'}}{
  K[\tau]}{}
\end{align*}

Since the inner call $ 
    \mcclause {\stBase {l}{\eta\sigma_1}{\Delta_1}} \tau {C\sigma_1}{K\sigma_1}
$ succeeds too, the lemma follows by inductive
hypothesis.

  \item
\begin{lstlisting}[
    backgroundcolor=\color{white}, 
    escapeinside={(*}{*)},
    firstnumber=16,
    firstline=30,
    lastline=42]
-- send
$
  \mcclause { \St l \psi } {\popo{\rho}{!T}} C K
    = 
  (\sigma_2\sigma_1, C_2)
$  (*\label{code:MC-send-psi}*)
  if  $ \subEql{l}{\rho} $
  and $ \sigma_1 = \subst\psi {!\alpha.\psi'} $ where $ \alpha, \psi' $ fresh
  and $ (\sigma_2, C_2) 
          = 
        \mcclause { \St l {\psi'}\sigma_1 } \tau {
          C\sigma_1 \cup \singleSu{T}{\alpha} } { K\sigma_1 } 
      $
\end{lstlisting}

Let us assume that the above clause has been used. 
Because the domain of inference substitutions is session variables $ \psi $, 
in this case we have:
\begin{align*}
\mc { \big(\stBase l \psi \Delta \big)\sigma_2\sigma_1 }
    { K[\popo{\rho}{!T}]\sigma_2\sigma_1 }
 = &~
\mc { \stBase l {!\alpha.(\psi'\sigma_2)} {\Delta\sigma_2\sigma_1} }
    { (K\sigma_2\sigma_1)[\popo{\rho}{!T}] }
\end{align*}
From the algorithm we have that $ \subEql{l}{\rho} $ holds.
The inner call to $ \algoMC $ has  
$ C_1 = C\sigma_1 \cup \singleSu{T}{\alpha} $ as input constraints set. By 
Lem. \ref{lem:sound-C-ext}, $ \cj {C_2} {C_1\sigma_2} $, and therefore
$ \subEqlBase{l}{\rho}{C_2} $ and 
$ \cj {C_2} {(\su{T}{\alpha})\sigma_2}$ hold. The latter implies  
 $ \cj {C_2} {\su{T}{\alpha}} $ because $ \sigma_2 $ is an inference
 substitution (and therefore any type variable $ \alpha' $ in $ T $ is fresh
 from the domain of $ \sigma_2 $).
Since $ \subEqlBase{l}{\rho}{C_2} $ and $ \cj {C_2} {\su{T}{\alpha}} $ both
hold, the abstract interpretation semantics allows the following transition:
\begin{align*}
  \stackstepStatBase
  { \stBase l {!\alpha.(\psi'\sigma_2)} {\Delta\sigma_2\sigma_1} }{
  (K\sigma_2\sigma_1)\popo{\rho}{!T} } { \stBase l {\psi'\sigma_2}
  {\Delta\sigma_2\sigma_1} } { (K\sigma_2\sigma_1)[\tau] } {C_2}
\end{align*} 
By inductive hypothesis, $ \stackopBase { \stBase l {\psi'\sigma_2}
  {\Delta\sigma_2\sigma_1} } { (K\sigma_2\sigma_1)[\tau] } {C_2} {} $ holds.
Because of this, and since $\stackstepStatBase
  { \stBase l {!\alpha.(\psi'\sigma_2)} {\Delta\sigma_2\sigma_1} }{
  (K\sigma_2\sigma_1)\popo{\rho}{!T_c} } { \stBase l {\psi'\sigma_2}
  {\Delta\sigma_2\sigma_1} } { (K\sigma_2\sigma_1)[\tau] } {C_2}$ is the only
  transition allowed by the abstract interpretation semantics, it follows that $ 
\stackopBase
  { \big(\stBase l \psi \Delta \big)\sigma_2\sigma_1 }
  { K[\popo{\rho}{!T}]\sigma_2\sigma_1 }
  {C_2}{}
$ holds, which proves the lemma.

\item The other cases for $ \popo{\rho}{!T} $,  $ \popo{\rho}{?T} $,
$ \popo{\rho}{!L_k}$, $ \popo{\rho}{?L_k} $,
$ \popo{\rho}{!\rho_d}$ and $ \popo{\rho}{?l_d}$ are all proved similarly,
except for this case:
\begin{lstlisting}[
    backgroundcolor=\color{white}, 
    escapeinside={(*}{*)},
    firstnumber=37,
    firstline=78,
    lastline=87]
$
  \mcclause {\stBase l {!\eta_d.\eta} {\St{l_d}{\eta_d'}}
  } {\popo{\rho}{!\rho_d}} C K
  = (\sigma_2\sigma_1, C_2)
$(*\label{code:MC-del-eta}*)
  if  $ \subEql{l}{\rho} $ and $ \subEql{l_d}{\rho_d} $
  and $ (\sigma_1, C_1) = \nsub(\eta_d', \eta_d, C) $
  and $ (\sigma_2, C_2) = 
    \mcclause {\St l {\eta} \sigma} \tau {C_1} {K\sigma_1} 
$
\end{lstlisting}


The only difference with the previous case is that we have to prove that $
\coTypeBase{C'}{\eta_{d'}\sigma}{\eta_d\sigma} $, which follows 
by Lem. \ref{lem:sound-sub}.

\item Suppose that the following $\algoMC $ clause has been applied:
\begin{lstlisting}[
    backgroundcolor=\color{white}, 
    escapeinside={(*}{*)},
    firstnumber=96,
    firstline=205,
    lastline=209]
-- sequencing
$ \mcclause \Delta {\bseq {b_1} {b_2}} C K
  = 
  \mcclause \Delta {b_1} C {\consK {b_2} K}
  $ (*\label{code:MC-seq}*) 
\end{lstlisting}

The lemma is proved directly by inductive hypothesis.

\item  Suppose that the following $ \algoMC $ clause has been used:
\begin{lstlisting}[
    backgroundcolor=\color{white}, 
    escapeinside={(*}{*)},
    firstnumber=99,
    firstline=211,
    lastline=220]
-- internal choice in the behavior      
$ \mcclause  \Delta {b_1\oplus b_2} C K
    = (\sigma_2\sigma_1, C_2) 
    $ (*\label{code:MC-plus}*)
  if  $ (\sigma_1, C_1) = \mcclause \Delta {b_1} C K $
  and $ (\sigma_2, C_2) = \mcclause {\Delta\sigma_1 
                                  } {b_2\sigma_1
                                  } {C_1
                                  } {K\sigma_1}
      $
\end{lstlisting}
    
This case is proved by Lem. \ref{lem:sound-C-ext}, inductive hypothesis and
Prop. \ref{prop:liskov}.

\end{itemize}
\end{proof}

Soundness of Algorithm $ \algoSI $ can be stated as follows:

\begin{theorem}[Soundness of Algorithm $\algoSI$]
Let $ C $ be a well-formed constraints set.
If $ \algoSI(b, C) = (\sigma_1, C_1) $, then $
\stackopBase{\epsilon}{b\sigma_1}{C_1}{}$.
\end{theorem}

\begin{proof}
The definition of Algorithm $ \algoSI $ is:
\begin{lstlisting}[
    backgroundcolor=\color{white}, 
    escapeinside={(*}{*)},
    firstnumber=1,
    firstline=1,
    lastline=4]
  $ \algoSI (b,~ C) = (\sigma_2\sigma_1,~ C_2) $
    if  $ (\sigma_1, C_1) = \mcclause \epsilon b C \emptyK $
    and $ (\sigma_2, C_2) = $ choiceVarSubst $ C_1 $
\end{lstlisting}

By definition of \textsf{choiceVarSubst}, $
C_2 = C_1\sigma_2 $, which implies $ \cj{C_2}{C_1\sigma_2} $.
The theorem follows directly by soundness of $ \algoMC $, i.e. Prop.
\ref{prop:mc-soundness}, and by Lem. \ref{lem:sound-C-ext}, 
since $ \epsilon \equiv \epsilon $ holds. 
\end{proof}

\section{Completeness of Algorithm $\algoSI$ }
\label{sec:completeness}
\newcommand{\refines}[2]{ #1 \vdash #2}
\newcommand{\nfeq}[1]{=_{\mathsf{NF}(#1)}}

\begin{lemma}   \label{lem:compl-sub-sigma}
If $ \cj {C}{C'\sigma} $ and $ \forall \psi.
\coType{\sigma(\psi)}{\sigma'(\psi)}$, then $ \cj{C}{C'\sigma'}$.
\end{lemma}
\begin{proof}
By structural induction on constraints over $ \beta $, $c $, $\psiint $ and $
\psiext$.
\end{proof}
\begin{lemma}[Completeness of $\mathsf{sub}$] \label{lem:compl-sub}
Let $ \eta $ and $ \eta' $ be well-formed in $ C$ and $ C_0\sigma$. 
If $ \coType{\eta_1\sigma}{\eta_2\sigma} $ and
$ \cj {C}{C_0\sigma}$, then $ \mathsf{sub}(\eta_1,
\eta_2, C_0) = (\sigma_1, C_1) $ terminates, and there exists a substitution 
$ \sigma' $
such that $ \forall \psi \in dom(\sigma).
\coType{\sigma(\psi)}{\sigma'(\sigma_1(\psi))} $ and $ \cj C {C_1\sigma'}$.
\end{lemma}
\begin{proof}
By structural induction on $ \eta_1 $. Let $ \eta_1' = \eta_1\sigma $
and $ \eta_2' = \eta_2\sigma$ be such that $ \coType {\eta_1'}{\eta_2'}$.

If $ \eta_1 $ is a variable $ \psi_1 $, then \textsf{sub} terminates at line
52 with $ \sigma_1 = \subst {\psi_1}{\eta_2} $ and $ C_1 = C_0\sigma_1 $,
regardless of the shape of $ \eta_2 $. 
Moreover $ \sigma $ can be decomposed into $ \sigma = \subst {\psi_1} {\eta_1'}
\cdot \sigma''$ for some $ \sigma'' $. The
lemma is proved by taking $ \sigma' = \sigma''$, because
$ \coType{\sigma(\psi)}{\sigma'(\sigma_1(\psi))} $ holds for all $
\psi \in dom(\sigma) $: for $\psi = \psi_1$ we have $ \sigma(\psi_1) = \eta_1'$,
which is a subtype of $ \sigma'(\sigma_1(\psi)) = \sigma'(\eta_2) = \eta_2'$ by 
hypothesis;
for any other $ \psi $, $ \sigma(\psi) $ and $ \sigma'(\psi
)$ are equal by
definition. Moreover $ \cj{C}{C_1\sigma'} $ follows by Lem. 
\ref{lem:compl-sub-sigma}.

If $ \eta_1 $ is a session of the form $ !T_1.\eta_1''$, then either $ \eta_2 =
\psi_2 $, or $ \eta_2 = !T_2.\eta_2'' $ with $ \coType{T_2}{T_1} $ holding by
definition of subtyping. The first case is proved similarly to the case for $
\eta_1 = \psi_1$, but using the clause at line 53. The second case is proved by
inductive hypothesis, since \textsf{sub} in this case only adds a new constraint 
$\singleSu{T_2}{T_1}$ to $ C_0 $ in the clause at line 5. 
The proof for the case $ \eta_1 = ?T_1.\eta'' $ is proved similarly at line 8,
considering that session sub-typing is covariant instead of contravariant.
When $ \eta_1 = !\eta_{1d}.\eta_1'' $ and $ \eta_2 \neq \psi_2$, the inner
clauses at line 13 terminates by inductive hypothesis, and the lemma is proved
similarly by transitivity of the session subtyping relation and of constraint 
refinement.
Similarly, when $ \eta_1 = ?\eta_{1r}.\eta_1''$ and $ \eta_2 \neq \psi_2 $, 
the clauses at line 16 terminate by inductive hypothesis and 
the proof is by inductive hypothesis and transitivity too.
When $ \eta_1 = \psiint_1$ and $ \eta_2 \neq \psi_2$, the lemma is proved by a
induction on the size of the internal choice indexes $ I_2 $,
where $ \eta_2 = \psiint_2 $ and $
\seql{\psiint_2}{\sichoicetext{i}{I_2}{\eta}}$ holds. Similarly,
when $ \eta_1 = \psiext_1$ and $ \eta_2 \neq \psi_2$, then $ \eta_2 = \psiext_2
$ and $ \seql{\psiext_2}{\sechoicetext{j}{J_1}{J_2}{\eta}}$ hold, and the proof
follows by induction on the size of $ J_1 \cup J_2 $ over the inner call to
function \textsf{f}.
\end{proof}

The completeness of Algorithm $ \algoSI $ relies on the completeness
of Algorithm $\algoMC $, stated as follows:
\begin{theorem}[Completeness of Algorithm $\algoMC $]
\label{prop:mc-completeness}
Let $ \mc \Delta {K[b]} $ be well-formed in $ C $ and in $ C_0\sigma $. If
$ \stackopBase
{ (\Delta}
{ K[b])\sigma }
{ C }
{ } $
and $ \cj { C } { C_0\sigma }$, then there exists $ \sigma' $ such that
$ \algoMC(\mcconf{\Delta}{b}, C_0)
= (\sigma_1, C_1) $ terminates,
$ \forall \psi \in dom(\sigma).
\coType{\sigma(\psi)}{\sigma'(\sigma_1(\psi))}
$ and
$ \cj {C} {C_1\sigma'} $
\end{theorem}

\begin{proof}

The lemma is proved by induction of the lexicographic order between the
execution size $ \execSize{\Delta\sigma_1}{K[b]\sigma_1}{C_1} $ and length of a
behavior $ b $ (i.e. structural induction).


{\hskip\labelsep \color{darkgray}\sffamily\bfseries {Base case.}}
The base case is when the execution size is equal to 1, that is, when
$ \execStates {\Delta}{K[b]\sigma} C = \{ \mc \Delta {K[b]\sigma}\} $, and
$ \mc \Delta {K[b]\sigma} \not\rightarrow_{C} $. Since  $ \mc \Delta
{K[b]\sigma} $ is strongly normalizing by hypothesis, then $ \Delta = \epsilon
$ and $ K[b] = \tau $, which implies that $ K = \emptyK $ and $ b = \tau $.
In such a case, the clause $ \algoMC(\mc \epsilon \tau, C_0, \emptyK) = (\id,
C_0) $ at line \ref{code:MC-delta-tau} terminates trivially.
The lemma is proved by taking $ \sigma' = \sigma $, because
$ \sigma'(\id(\psi)) = \sigma'(\psi) = \sigma(\psi) $ for any $ \psi $ and
therefore
$\coType{\sigma(\psi)}{\sigma'(\sigma_1(\psi))}$ holds by reflexivity, 
and because $ \cj C {C_0\id = C_0} $ holds by hypothesis.

{\hskip\labelsep \color{darkgray}\sffamily\bfseries {Inductive case. }}
In the inductive case, the execution size $ \execSize {\Delta\sigma}{K[b]\sigma}
C $ is greater than 1, and therefore there exists a configuration
$ \mc {\Delta'}{b'} $ such that $ \dstep {\Delta\sigma}{K[b]\sigma}{\Delta'}{b'}
$ holds. 
It is easy to show that $ \dstep {\Delta\sigma}{K[b]\sigma}{\Delta'}{b'}
$ holds if and only if $ b $ is a $ \kpop, \kpush $
We proceed by rule induction:
\begin{itemize}
  \item [\RefTirName{End}:]
  Suppose that Rule \RefTirName{End} has been used:
  \begin{align*}
\dstep {\stBase l \tend {\Delta\sigma}} {K[b]\sigma}
{\Delta\sigma}{K[b]\sigma}
\end{align*}
There are two cases to consider: either $ \Delta = \stBase l {\tend}
{\Delta'} $, or $ \Delta = \stBase l {\psi} {\Delta'}
$ and $ \psi\sigma = \tend $.
In the first case, the clause $ \mcclause {\stBase l \tend {\Delta}}
b {C_0} K = \mcclause \Delta b {C_0} K $
and the lemma is proved by the inductive hypothesis. In the second case,
we need to show that $ \mcclause {\stBase l \psi {\Delta}} b {C_0} K $
terminates; this can be proved by structural induction on $ K[b] $.
If $ b = \tau $ and $ K = \emptyK $, then $ \algoMC $ terminates by applying $
\mathsf{finalize} $ on $ \Delta $, and the proposition is proved as in the
base case. If $ b = \tau $ and $ K \neq \emptyK $, then the proposition is
proved by inductive hypothesis on the clause at line \ref{code:MC-delta-ktau}.

If $ b = \pusho{l}{\eta'}{} $, then $ \mathsf{checkFresh} $ closes $ \psi $
in the clause at line \ref{code:MC-push}, because a well-formed stack cannot
push the same label $ l $ twice. In all other cases either the proposition
follows by induction, or $ K[b] $ must have the form of a $ \kpop $
operation. In the latter case, $ l $ cannot be contained in region $ \rho $
of the $ \kpop $ operation by well-formedness of the stack, and therefore the
only applicable clause in the last one at line \ref{code:MC-closetop}, that
calls \textsf{closeTop}.

\item[\RefTirName{Push}:] Suppose that the following transition is taken:
\begin{align*}
\dstep {\Delta\sigma}  {(K[\pusho{l}{\eta}{}])\sigma }
{\stBase {l}{\eta\sigma}{\Delta\sigma}}   {K[\tau]}
&& \text{if } l \freshfrom \Delta
\end{align*}

We first need to show that $ \mcclause \Delta {\pusho{l}{\eta}{}} {C_0}
{K} $ terminates.
Suppose that the frames with closed sessions $ \tend $ are removed
from the top of the stack, as in the case for \RefTirName{End}.
The first clause that matches $ \mcclause \Delta {\pusho{l}{\eta}{}} {C_0}
{K} $ is the one at line \ref{code:MC-push}.
By hypothesis $ l \freshfrom \Delta $ holds, therefore the inner call to
\textsf{checkFresh} can only return the identity substitution.
Therefore $ \sigma_1 = \id $ and $ \Delta_1 = \Delta$ hold, and there exists a
substitution $ \sigma '' $ such that the starting substitution $ \sigma $ can
be split in the composition of $ \sigma_1 $ and $ \sigma'' $, i.e.
$\sigma = \sigma''\sigma_1 = \sigma''$.
The lemma is then proved by inductive hypothesis on the inner call to $
\algoMC $ with the smaller configuration $ \mc {\stBase
{l}{(\eta\sigma_1)\sigma''}{(\Delta_1\sigma_1)\sigma''}} {\tau}
= \mc {\stBase{l}{\eta\sigma}{\Delta\sigma}} {\tau} $.
%

\item[\RefTirName{Out}:]
Let $ (\mc{\Delta}{K[b]})\sigma = \mc {\stBase l {!T.\eta\sigma}
{\Delta'\sigma}} {\popo{\rho}{!T'} } $, and suppose that the following
transition is taken:
\begin{align*}
\dstep {\stBase l {!T.\eta\sigma} {\Delta'\sigma}} {K[\popo{\rho}{!T'}]}
{\stBase l {\eta\sigma}{\Delta'\sigma}} {K[\tau]}
&& \text{if } \nbox{ C\vdash\seql{\rho}{l},~ \mathit{pure}(T'),~
\subt{T'}{T}}
\end{align*}

We must consider two cases: $ \Delta = \stBase l {\psi} {\Delta'} $ and $
\sigma(\psi) = !T.\eta $, or $ \stBase l {!T.\eta}{\Delta'} $.

In the first case the call to
$ \mcclause {\stBase l {\psi}{\Delta'}} {\popo{\rho}{!T'}} {C_0} K $ is
matched by the clause at line \ref{code:MC-send-psi}, and it produces the
substitution $\sigma_1 = \subst{\psi}{!\alpha.\psi'} $ and a new constraint $
\su{T'}{\alpha}$ by definition, where $ \psi' $ and $ \alpha' $ are fresh
variables.
By hypothesis $ \sigma $ has the form $ \subst \psi
{!T.\eta'} $, and it can be decomposed into $ \sigma'' = \subst \psi
{!T.\psi'} $ and $ \sigma'''= \subst {\psi'} {\eta' } $.
The lemma is proved by applying the inductive hypothesis on the
inner call to $ \algoMC $ and by taking $ \sigma' = \sigma'''\subst{\alpha} T$,
since $ \sigma'' $ can be further decomposed into $ \sigma_1$ and $\subst \alpha
T $ and therefore the substitution $ \subst {\alpha} T $ guarantees that
$\sigma'\sigma_1  = \sigma $; therefore $ \forall \psi\in dom(\sigma).
\coType{\sigma(\psi)}{\sigma'(\sigma_1(\psi))}$ follows by reflexivity.
Notice that $ \cj{C}{C_1\sigma'} $ holds too by definition of constraints
refinement, because the new constraint $ \su{T'}{\alpha}\sigma $ becomes $
\su{T'}{T} $ when $ \sigma' = \sigma'''\subst{\alpha}{T} $ is applied to
it, and because $ \coType{T'}{T} $ holds by hypothesis from the side-conditions
on Rule \RefTirName{Out}.

If $ \Delta = \stBase l {!T.\eta}{\Delta'} $, then the
call $ \mcclause \Delta {\popo{\rho}{!T'}} {C_0} K $
is matched by the clause at line \ref{code:MC-send-eta}. The lemma is proved
straightforwardly by taking $ \sigma' = \sigma $, since Algorithm $ \algoMC$
returns the identity substitution in this case.

\item[\RefTirName{In}:] This case is proved similarly to [\RefTirName{Out}],
using the clauses at lines \ref{code:MC-recv-psi} and \ref{code:MC-recv-eta},
recalling that session sub-typing is covariant instead of contravariant for
inputs.

\item[\RefTirName{Del}:] This case is proved similarly to [\RefTirName{Out}],
using the clauses at lines \ref{code:MC-del-psi} and \ref{code:MC-del-eta},
together with Lem.
\ref{lem:compl-sub} in the latter case. 
Notice that $ \seql{l_d}{\rho_d}$ holds by hypothesis, and
therefore the clause at line \ref{code:MC-del-checkFresh} is not applicable.

\item[\RefTirName{Res}:] This case is proved similarly to [\RefTirName{Out}],
using the clauses at lines \ref{code:MC-res-psi} and \ref{code:MC-res-eta}.
Notice that the clause at line \ref{code:MC-res-checkFresh} cannot be called,
because the labels in the behavior $ b = \popo{\rho}{?l_r} $ match the label in
the stack $ \Delta $ by hypothesis.

%


\item[\RefTirName{ICh}:] 
Let $ (\mc{\Delta}{K[b]})\sigma = 
\mc {\stBase l {\psiint}
{\Delta'\sigma}} {K\sigma[\popo{\rho}{!L_j}] } $, and suppose that the following
transition is taken:
\begin{align*}
\dstep {\stBase {l}{\sichoice{i}{I}{\eta}}{\Delta'}}
            {K\sigma[\popo{\rho}{!L_k}] 
            & }
            {\stBase {l}{\eta_k}{\Delta'} }    {K\sigma[\tau]}
            & \text{if } (j \in I),~ C\vdash\seql{\rho}{l}
\end{align*}

There are two cases to consider: either 
$ \Delta = \stBase l \psi {\Delta'} $, or $ \Delta = \stBase l \psiint
{\Delta'} $.

If $ \Delta = \stBase l \psi {\Delta'} $, then the only clause that
matches the call to $ \mcclause{\Delta}{b}{K}{C_0}$ is the one at line \ref{code:MC-ich-psi}.
Since $ \psi\sigma = \psiint $, substitution $ \sigma $ can be decomposed as $
\sigma = \sigma'\subst{\psi}{\psiint} $. Moreover $
C \vdash \seql{\psiint}{\sichoicetext{i}{J}{\eta}} $ and $ k \in J $ hold by
the side conditions of Rule \RefTirName{ICh}.
Algorithm $ \algoMC $ first produces the substitution $ \sigma_1 = \subst \psi
\psiint $, and introduces a new constraint $
\seql{\psiint}{\sichoicetext{i}{\{k\}}{\psi}}$ in $ C_0 $, where $ \psi_k$ is a
fresh variable. Let $ \sigma'' =
\sigma'\subst{\psi_k}{\eta_k}\sigma_1$.
Since $ \mc{\stBase l {\psi} {\Delta'}}{K\sigma''[\popo{\rho}{L_k}]}
\Downarrow_C$ holds by hypothesis, and since $ (\mc{\stBase l {\psi_k}
{\Delta'}}{K[\tau]}) = \mc {\stBase l
{\eta_k}{\Delta'\sigma'\sigma_1}}{K\sigma'\sigma_1[\tau]}
$ holds because $ \psi_k$ is a fresh variable, then 
$ \mc {\stBase l {\eta_k}{\Delta'\sigma'\sigma_1}}{K\sigma'\sigma_1[\tau]}
\Downarrow_C$ holds too. 
Moreover, since $ \cj C {C_0\sigma} $ holds by hypothesis and 
$ C \vdash \seql{\psiint}{\sichoicetext{i}{J}{\eta}}$ where $ J $ contains $ k $
both hold by the side condition of Rule \RefTirName{ICh},
then $ \cj C {C_0\sigma'' \cup \singleSu{\psiint}{\psi_k\sigma''}} $ holds too,
because by definition of constraint refinement and of session sub-typing 
$ \coType{\sichoicetext{j}{J}{\eta}}{\sichoicetext{i}{\{k\}}{\eta}}$.
Therefore the inductive hypothesis applies to the inner call of 
$ \mcclause{\stBase l {\psi_k} {\Delta'}}{\tau}{K}{C_0 \cup
\singleSu{\psiint}{\psi_k}} $, and which also proves the theorem, since $
\psi_k $ is not in the domain of $ \sigma$.

In the case that $ \Delta = \stBase l \psiint {\Delta'} $, then $ C_0 $ must
contain a constraint $ \seql{\psiint}{\sichoicetext{i}{I}{\eta}} $ by definition
of constraint well-formedness. There are two sub-cases to consider: either $ k $
is in $ I $, or it is not. In the former case the theorem is proved
straightforwardly by inductive hypothesis on the inner clause at line
\ref{code:MC-ich-psiint}. If $ k $ is not in $ I $, then the only applicable
clause of Algorithm $ \algoMC $ is the one at line \ref{code:MC-ich-psiint-not},
which adds $ k $ to $ \psiint $, and the theorem follows by inductive hypothesis
on the inner call to $ \algoMC $ with the extended internal choice.

\item[\RefTirName{ECh}:] Suppose that the following transition is taken:
\begin{align*}
\dstep {\St {l}{\sechoice{i}{I_1}{I_2}{\eta}}}
           {\bechoice{j\in J}{\rho}{L_j}{b_j} 
           &}
            {\St {l}{\eta_k} }{b_k}
            &\text{if }\nbox{
             k\in J,~ C\vdash\seql{\rho}{l},~
             \\I_1\subseteq J \subseteq I_1 \cup I_2
            } 
\end{align*}
The proof of this case is similar to the proof for Rule \RefTirName{ICh}: if 
$ \Delta = \stBase l \psi {\Delta'} $, then Algorithm $ \algoMC $ terminates by
applying the clause at line \ref{code:MC-psiext-psi}, which first creates a new
constraint on $ \psiext $, and then recursively calls the clause at line
\ref{code:MC-psiext}, which is proved by inductive
hypothesis and transitivity of session sub-typing. 
If $ \Delta = \stBase l {\sechoice{i}{I_1}{I_2}{\eta}} {\Delta'} $, then either
the clause at line \ref{code:MC-psiext-not} is called, in case either some
labels in $ J $ from the behavior are missing from $ I_1 $ or $ I_2 $, or in
case the active labels in $ I_1 $ are more than the labels in $ J$; or the
clause at line \ref{code:MC-psiext} is called, in case the session type
designated by $ \psiext$ contains all the labels $ J$in the behavior, and the
active labels $ I_1 $ are included in $ J $.
In the former case the constraint in the session type are adjusted appropriately
and the theorem is proved as in the first case, because the clause at line
\ref{code:MC-psiext} is the only clause that matches the adjusted session type.
In the latter case the proof is by inductive hypothesis and transitivity as in
the first case.

\ref{code:MC-psiext-not}
\ref{code:MC-psiext}

\item[\RefTirName{Spn}:] Suppose that the following transition is taken:
\begin{align*}
\dstep {\Delta} { K[\espawn{b}] &} {\Delta} {K[\tau]}
          & \text{if } \stackop {} {b} {\epsilon}{}
\end{align*}

Assume that the frames with closed sessions $ \tend $ are removed from the top
of the stack, as in the case for \RefTirName{End}, resulting in a stack $
\Delta'$. 
The clause $ \mcclause{\Delta'}{\espawn{b}}{K}{C_0}$ is matched at line
\ref{code:MC-spawn}. Since $ \mc \epsilon b $ and $ \mc {\Delta'}{K[\tau]} $ are
smaller configurations than $ \mc {\Delta'}{\espawn{b}}$, the theorem follows by
inductive hypothesis and transitivity of session sub-typing.

\item[\RefTirName{Rec}:] This case is proved similarly to the case for Rule
\RefTirName{Spn}, with the exception that the clause at line \ref{code:MC-rec}
is called, and that the environment $ C_0 $ is properly manipulated to swap the
recursive constraint $ \su{\orec{\beta}b}{\beta}$ with $ \su{\tau}{\beta}$.


\item[\RefTirName{Plus}:] Suppose that the following transition is taken:
\begin{align*}
\dstep {\Delta} {K[b_1 \oplus b_2]} {\Delta}{K[b_i]}
&&\text{if } i \in \{1, 2\}
\end{align*}

Assume that the frames with closed sessions $ \tend $ are removed from the top
of the stack, as in the case for \RefTirName{End}.
The only applicable clause in this case is the one at line \ref{code:MC-plus}.
The inductive hypothesis can be applied directly on the first inner call to $
\algoMC $, which implies that $ \sigma $ can be refined into $ \sigma'\sigma_1$.
By inductive hypothesis, $ \coType{\sigma(\psi)}{\sigma'(\sigma_1(\psi))}$ holds
for all $ \psi $ in the domain of $ \sigma $. Notice that, by construction of
algorithms \textsf{sub} and $ \algoMC $, $ \sigma'\sigma_1$ produces super types
only in the case of delegation, therefore it is easy to show that 
$ (\mc{\Delta}{K[b_1 \oplus b_2})\sigma'\sigma_1 \Downarrow_C$ holds as well.
Therefore $ (\mc{\Delta}{K[b_2})\sigma'\sigma_1 \Downarrow_C$ holds too, and the
theorem is proved by the inductive hypothesis.

\item[\RefTirName{Beta}:] Suppose that the following transition is taken:
\begin{align*}
\dstep {\Delta\sigma} {\beta\sigma}
{\Delta\sigma} {b}
&&         \text{if } \sub{b}{\beta}
\end{align*}

Since $ \mc \Delta \beta\sigma $ is
well-formed in both $ C $ and $ C_0\sigma$, variable $ \beta $
is constrained not only in $ C $ (by assumption $\sub{b}{\beta}$ holds), but it
is also constrained to some $ b' $ in $ C_0\sigma $, i.e. $
\subBase{C_0\sigma}{b'}{\beta\sigma}$. By definition of constraint refinement,
$ \cj C {C_0\sigma} $ implies that  $ b' = b\sigma $, because there is no
sub-typing relation defined for behaviors, therefore $ b\sigma $ and $ b'$ must
be equal. The lemma is then proved directly by inductive hypothesis.

Assume that the frames with closed sessions $ \tend $ are removed from the top
of the stack, as in the case for \RefTirName{End}.
If the top of the stack in $ \Delta $ is not $ \tend$, the only applicable
clause is the one at line \ref{code:MC-beta}, whereby
$ \mcclause \Delta \beta {C_0} K = \mcclause \Delta {b_0} {C_0} K $
with $ b_0 = \bigoplus \{b_i \where \exists i. ~(b_i \subseteq \beta) \in
C\}$. By definition of $ \algoMC $, the only clause that matches the inner
call $ \mcclause \Delta {b_0} {C_0} K $ is the clause for internal choice at
line \ref{code:MC-plus}. Since the execution size of $ \mc {\Delta\sigma}
{\beta\sigma} $ is greater than the execution size of each $ \mc
{\Delta\sigma} {b_i\sigma} $ component, the proof follows by inductive
hypothesis, in the same way as for Rule \RefTirName{Plus}.

\item[\RefTirName{Seq}:] Suppose that the following transition is taken:
\begin{align*}
\dstep {\Delta}          {K[b_1;b_2]
&} {\Delta'}{K[b_1';b_2]}
& \text{if } \dstep {\Delta} {b_1}{\Delta'}{b_1'}
\end{align*}

Assume that all closed $\tend $ sessions are removed from $ \Delta $, resulting
in $ \Delta' $, as for the case \RefTirName{End}. 
The clause $ \mcclause{\Delta'}{b_1;b_2}{K}{C_0}$ becomes
$ \mcclause{\Delta'}{b_1}{b_2 \cdot K}{C_0}$ at line \ref{code:MC-seq}, and the
lemma is proved by inductive hypothesis, because the execution size of 
$\mc {\Delta'}{K[b_1;b_2]} $ is equal to the execution size of 
$\mc {\Delta'}{b_2\cdot K[b_1]} $, but the size of $ b_1 $ is smaller than the
size of $ b_1;b_2$.

\item[\RefTirName{Tau}:] 
Suppose that the following transition is taken:
\begin{align*}
\dstep {\Delta}          {K[\tau;b]
&} {\Delta}{K[b]}
&\end{align*}

Assuming that all closed session are removed from $\Delta$ as in the case for 
\RefTirName{End} and that a stack $ \Delta'$ is returned, the clause at line
\ref{code:MC-seq} is called first, whereby $ b_2 $ is pushed on the stack $ K $. Then the clause 
at line \ref{code:MC-delta-ktau} is called recursively, whereby the $ \tau $
behavior is discarded and $ b_2 $ is popped back from the stack. The lemma 
follows by inductive hypothesis on the smaller configuration $ \mc {\Delta'}
{K[b_2]}$.
\end{itemize}
\end{proof}

Completeness for session type inference can be stated as follows:
\begin{theorem}[Session type inference completeness]
Let $ \mc \epsilon {b\sigma^\star} $ and $ C^\star $ be well-formed.
If $ \stackopStatBase{\epsilon}{b\sigma^\star}{C^\star}$ and $
\refines{C^\star}{ C} $, then $ \algoSI(b, C) = (\sigma_1, C_1) $ and there
exists $\sigma $ such that $ \refines{C^\star}{C_1\sigma} $ and $
\forall \psi \in dom(\sigma^\star).
\sub{\sigma^\star(\psi)}{\sigma\sigma_1(\psi)}
$.
\end{theorem}
\begin{proof}
The proof follows by applying Proposition \ref{prop:mc-completeness}  on the
configuration $ \mc \epsilon b $ under $ C $ first.
Algorithm $\algoMC $ returns $ \sigma_1 $ and $ C_1 $ such that there exists a $
\sigma'$ such that $ \sub{\sigma^\star(\psi)}{\sigma'\sigma_1(\psi)} $ for any 
$\psi$ in the domain of $ \sigma^\star $, and $ \cj{C}{C_1\sigma'}$.
Since the call to \textsf{choiceSubst} simply substitutes $ \psiint $ and $
\psiext $ variables with their (unique) relative internal and external
choices in $ C_1 $, the new substitution $ \sigma_2 $ and $ C_2 $ that this
function returns does not change the typing of sessions, and therefore
$ \forall \psi \in dom(\sigma). \coType{\sigma(\psi)}{\sigma'
(\sigma_2(\sigma_1(\psi)))} $ holds, and 
$ \cj{C}{C_2\sigma'} $ follows by $ \cj{C}{C_1\sigma'}$, which proves the
lemma.
\end{proof}

%
%

%
%

\end{document}